\documentclass[11pt,a4paper]{article}
\usepackage{theapa}
\usepackage{graphicx,nicefrac}
\usepackage{tabularx,booktabs,multirow}
\usepackage{todonotes}
\usepackage{amssymb,amsmath,amsthm}
\usepackage{mathtools}
\usepackage{cleveref}
\usepackage{subcaption}
\usepackage{amsmath,amsfonts}
\usepackage{xcolor}
\usepackage{xspace}
\usepackage{bbm}
\usepackage{tikz}
\usepackage{makecell}
\usepackage{authblk}
\usepackage{fullpage}
\usepackage[inline]{enumitem}
\usetikzlibrary{arrows,decorations.pathmorphing,decorations.pathreplacing,
	backgrounds,positioning,fit,matrix}
\usetikzlibrary{shapes,calc}
\newcommand{\problemdef}[3]{
	\begin{center}
		\begin{minipage}{0.95\textwidth}
			\noindent
			\textsc{#1}
			
			\vspace{2pt}
			\setlength{\tabcolsep}{3pt}
			\begin{tabularx}{\textwidth}{@{}lX@{}}
				\textbf{Input:} 		& #2 \\
				\textbf{Question:} 	& #3
			\end{tabularx}
		\end{minipage}
	\end{center}
}

\newtheorem{observation}{Observation}
\newtheorem{claimS}{Claim}

\newtheorem{theorem}{Theorem}
\newtheorem{corollary}{Corollary}
\newtheorem{proposition}{Proposition}
\newtheorem{definition}{Definition}

\newtheorem{lemma}{Lemma}

\newcommand{\stubsEq}{\textsc{S-Eq-Stub}\xspace}

\newcommand{\sEq}{\textsc{S-Eq}\xspace}

\newcommand{\stubjEq}{\textsc{J-Eq-Stub}\xspace}

\newcommand{\jEq}{\textsc{J-Eq}\xspace}

\newcommand{\NN}{\mathbb{N}}

\newcommand{\bigO}{\mathcal{O}}

\usepackage{algorithm}
\usepackage{algpseudocode}
\algnewcommand\algorithmicinput{\textbf{Input:}}
\algnewcommand\algorithmicoutput{\textbf{Output:}}
\algnewcommand\Input{\item[\algorithmicinput]}
\algnewcommand\Output{\item[\algorithmicoutput]}

\usepackage{thmtools} 
\usepackage{thm-restate}
\usepackage{amssymb}

\usetikzlibrary{decorations.pathreplacing,calligraphy}

\definecolor{c_T1}{rgb}{0.61, 0.87, 1.0}
\definecolor{c_T2}{rgb}{0.64, 0.0, 0.0}
\definecolor{seagreen}{rgb}{0.61, 0.87, 1.0}

\def\edgeDotted(#1,#2){\draw[majarr] (#1) edge ($(#1)!0.35!(#2)$)  
	($(#1)!0.65!(#2)$) edge (#2);
	\node[] at ($(#1)!0.5!(#2)$) {$\dots$};
}
\newcommand{\fitellipsis}[3] 
{\draw[#3] let \p1=(#1), \p2=(#2), \n1={atan2(\y2-\y1,\x2-\x1)}, 
	\n2={veclen(\y2-\y1,\x2-\x1)}
	in ($ (\p1)!0.5!(\p2) $) ellipse [x radius=\n2/2+2ex, y radius=5ex, 
	rotate=\n1];
}

\newcommand{\fitellipsiss}[3] 
{\draw[#3] let \p1=(#1), \p2=(#2), \n1={atan2(\y2-\y1,\x2-\x1)}, 
	\n2={veclen(\y2-\y1,\x2-\x1)}
	in ($ (\p1)!0.5!(\p2) $) ellipse [x radius=\n2/2+2ex, y radius=4ex, 
	rotate=\n1];
}
 
\Crefname{observation}{Observation}{Observations}



\title{Equilibria in Schelling Games: Computational Hardness and Robustnesss\thanks{An extended abstract of 
the paper is scheduled to appear in the proceedings of the \emph{21st International Conference on Autonomous Agents and Multiagent Systems (AAMAS~2022)}.}}

\author{Luca Kreisel} 
\author{Niclas Boehmer}
\author{Vincent Froese}
\author{Rolf Niedermeier}

\affil{\small
  Technische Universit\"at Berlin, Faculty~IV, Algorithmics and Computational 
  Complexity\protect\\
  \{l.kreisel,niclas.boehmer,vincent.froese\}@tu-berlin.de}

\date{\today}

\begin{document}

\maketitle

\begin{abstract}
	In the simplest game-theoretic formulation of Schelling's model of 
segregation on graphs, agents of two 
different types each select their own vertex in a given graph so as to
maximize the fraction of agents of their type in their occupied
neighborhood. Two ways of modeling agent movement here are either to allow two 
agents to swap their vertices or to allow an agent to jump to a free vertex. 

The contributions of this paper are twofold. First, we prove that deciding the 
existence of a swap-equilibrium and a jump-equilibrium in this 
simplest 
model of Schelling 
games is NP-hard, thereby 
answering questions left open by Agarwal et al. [AAAI~'20] and Elkind et al.\ 
[IJCAI~'19]. 
Second, we introduce two measures for the robustness of equilibria in Schelling 
games in terms of the minimum number of 
edges or the minimum number of vertices that need to be deleted to make an 
equilibrium unstable. We prove 
tight lower and upper bounds on the edge- and vertex-robustness of 
swap-equilibria in Schelling 
games 
on 
different graph classes.
	
\end{abstract} 

\section{Introduction}\label{chap:intro}
Schelling's model of segregation \cite{schelling1969models,schelling1971models} 
is a simple random process that aims at explaining segregation patterns 
frequently observed in real life, e.g., in the context of residential 
segregation 
\cite{massey1988dimensions,williams2016racial}. In Schelling's model, one 
considers agents of two 
different types. Each agent is initially placed uniformly at random on 
an individual vertex of some given graph (also called \emph{topology}), where 
an 
agent 
is  called \emph{happy} if at least a $\tau$-fraction of its neighbors is of 
its type for some given tolerance parameter $\tau \in (0,1]$. Happy agents do 
not change location, whereas, depending on the model, unhappy agents either 
randomly swap vertices with other 
unhappy agents or randomly jump to empty vertices. 
\citeA{schelling1969models,schelling1971models}  observed that even 
for tolerant agents with $\tau \sim \frac{1}{3}$, segregation patterns (i.e., large connected areas where agents have only neighbors of their type) are 
likely to occur. Over the last 50 years, Schelling's model has been thoroughly 
studied 
both from an empirical (see, e.g., 
\citeS{DBLP:conf/atal/CarverT18,clark2008understanding})
and a theoretical (see, e.g., 
\citeS{DBLP:conf/focs/BarmpaliasEL14,DBLP:conf/soda/BhaktaMR14,DBLP:conf/stoc/BrandtIKK12,DBLP:conf/soda/ImmorlicaKLZ17})
perspective in various disciplines including computer science, economics, 
physics, and sociology. Most works focused on 
explaining under which circumstances and how quickly segregation patterns occur.

In Schelling's model it is assumed that unhappy agents move randomly and 
only 
care about whether their tolerance threshold is met. As this seems unrealistic, 
Schelling games, which are a game-theoretic 
formulation of Schelling's model where agents move strategically in order to 
maximize their individual utility, have recently 
attracted 
considerable attention 
\cite{AGARWAL2021103576,bilo2020topological,ChauhanLM18,Echzell0LMPSSS19,KanellopoulosKV20}.
However, there is no unified formalization of the agents' utilities
in the different game-theoretic formulations. In this paper, we assume that 
all agents want to maximize the fraction of agents of their type in their 
occupied neighborhood. In contrast, for instance, \citeA{ChauhanLM18} and \citeA{Echzell0LMPSSS19} assumed that the utility of an agent~$a$ depends 
on the 
minimum of its threshold parameter $\tau$ and the fraction of agents of 
$a$'s type in the occupied neighborhood of $a$ 
(as done by \citeA{AGARWAL2021103576} and \citeA{bilo2020topological}, we assume that 
$\tau=1$).
Along a different dimension, in 
the works of
\citeA{ChauhanLM18}  and  \citeA{AGARWAL2021103576}, the utility of (some) agents also depends on 
the 
particular vertex they occupy. For instance, \citeA{AGARWAL2021103576} assume that there exist some agents 
which are \emph{stubborn} and have a favorite vertex in the graph 
which they 
never leave (in our model, as also done before by
\citeA{bilo2020topological}, \citeA{Echzell0LMPSSS19}, and \citeA{KanellopoulosKV20}, we assume that 
the agents'
behavior does not depend on their specific vertex, so there are no stubborn 
agents).
The main focus in previous works and our work lies on the analysis 
of certain pure equilibria in Schelling games, where it is typically either assumed 
that 
agents can swap their vertices or jump to empty vertices. While we mostly focus 
on swap-equilibria, we also partly consider jump-equilibria. As 
the existence and other specifics of equilibria crucially depend on the 
underlying topology, a common approach
is 
to consider different  graph~classes 
\cite{bilo2020topological,ChauhanLM18,KanellopoulosKV20}.    

Our paper is divided into two parts. First, we prove that deciding the 
existence of a swap- or jump-equilibrium in Schelling 
games (as formally defined in \Cref{se:SG}) is NP-hard. 
Second, we introduce a new perspective for the analysis of equilibria in 
Schelling games: Stability under changes.
Considering the original motivation of modeling 
residential segregation, it might for example occur that agents move away and 
leave the city.
In the context of swap-equilibria on which we focus in the second part, this 
corresponds to deleting the vertex occupied by the leaving agent (or all edges 
incident to it), as unoccupied vertices can never get occupied again and 
are also irrelevant for computing agents' utilities.
An interesting task now is to find a more ``robust'' equilibrium that remains stable in such a changing environment, as for non-robust equilibria it could be that a small change can cause the reallocation of all agents.
We formally define a worst-case 
measure for the \textit{robustness} of an equilibrium with respect to vertex/edge deletion as the 
maximum integer~$r$ such that the deletion of any set of at most~$r$ 
vertices/edges leaves the equilibrium stable. Clearly, the robustness depends 
on the underlying topology. That is why we 
study different graph classes with respect to the robustness of~equilibria.

\subsection{Related Work}	
Most of the works on Schelling games focused on one of three 
aspects: existence and 
complexity of computing equilibria 
\cite{AGARWAL2021103576,bilo2020topological}, game dynamics 
\cite{bilo2020topological,ChauhanLM18,Echzell0LMPSSS19}, and price of 
anarchy and stability 
\cite{AGARWAL2021103576,bilo2020topological,KanellopoulosKV20}.
The first area is closest to our paper, so we review some
results here. On the negative side, 
\citeA{AGARWAL2021103576} showed that a jump- and swap-equilibrium may fail to 
exist even on tree topologies and that checking their existence is NP-hard on 
general 
graphs in the 
presence of stubborn agents. On a tree, the existence of a jump- and swap equilibrium can be checked in polynomial time~\cite{AGARWAL2021103576}. On the positive side,  \citeA{AGARWAL2021103576} showed that a jump-equilibrium always exists
on 
stars, paths, and cycles. Concerning swap-equilibria,   
\citeA{Echzell0LMPSSS19} showed that a swap-equilibrium always exists on regular 
graphs (and in particular cycles).  Recently, 
\citeA{bilo2020topological} proved that a swap-equilibrium is guaranteed 
to exist on paths and grids, and they obtained further results for 
the restricted
case where
only adjacent agents are allowed to swap.

Lastly, different from 
the three above mentioned directions, \citeA{DBLP:journals/corr/abs-2012-02086} and  \citeA{DBLP:journals/corr/abs-220106904} studied finding Pareto-optimal 
assignments and assignments maximizing the summed utility of all agents in 
Schelling games and \citeA{DBLP:conf/atal/ChanIT20}  proposed a 
generalization of Schelling games where, among others, multiple agents can 
occupy the same 
vertex.
 
Interpreting the neighborhood of an agent as its coalition, Schelling games 
also share characteristics with hedonic games  
\cite{DBLP:journals/geb/BogomolnaiaJ02,dreze1980hedonic} and in particular 
with 
(modified) fractional hedonic games 
\cite{DBLP:journals/teco/AzizBBHOP19,DBLP:journals/jair/Bilo0FMM18,DBLP:journals/aamas/MonacoMV20}
and hedonic diversity games 
\cite{DBLP:conf/aaai/BoehmerE20,DBLP:conf/atal/BredereckEI19}. The main 
difference between hedonic games and Schelling games is that in Schelling 
games agents are placed on a topology so each agent has its own coalition and coalitions may overlap. Recently, 
\citeA{DBLP:conf/atal/BodlaenderHJOOZ20} proposed a model 
that can be somewhat considered as 
both a  
generalization of hedonic games and Schelling games:  
In their model, 
a set of agents with cardinal preferences over each other need to be placed on
the 
vertices of a graph and the utility of an agent is its summed  utility that 
it derives from all adjacent
agents.  

Analyzing the robustness of outcomes of decision processes has become a popular topic in  
algorithmic game theory 
\cite{DBLP:journals/mp/AghassiB06,DBLP:journals/jair/BoehmerBHN21,DBLP:journals/ai/BredereckFKNST21,DBLP:journals/jet/Kets11,DBLP:conf/alt/Perchet20,DBLP:journals/jet/TakahashiT20}. 
For instance, in the context of  hedonic games, 
\citeA{DBLP:conf/ijcai/IgarashiOSY19} studied stable outcomes that 
remain stable even after some agents have been deleted and, in the context of 
stable matching, 
\citeA{DBLP:conf/esa/MaiV18,DBLP:journals/corr/abs-1804-05537} and
\citeA{DBLP:conf/ec/ChenSS19} studied stable matchings that remain 
stable 
even if the agents' preferences partly change.

\subsection{Our Contributions}	
The contributions of this paper are twofold. 
In the first, more technical part, we prove that deciding the 
existence of a jump- or swap-equilibrium 
in Schelling games where all agents want to maximize the fraction of agents of 
their type in their occupied neighborhood is NP-hard. Notably, our 
technically involved
results
strengthen results by  \citeA{AGARWAL2021103576}, who proved the NP-hardness of these 
problems making decisive use of the existence of stubborn agents (which never 
leave their vertices).

Having analyzed the existence of equilibria, in the second, more 
conceptual part, we introduce a notion for 
the robustness of an equilibrium under vertex/edge deletion: We say that an equilibrium has vertex/edge-robustness $r$ if it 
remains stable under the deletion of any set of at most $r$ vertices/edges but not under the deletion of~$r+1$ vertices/edges.
We study the existence of swap-equilibria with a given robustness.
We restrict our attention to swap-equilibria, as for jump-equilibria, the robustness heavily depends on both
the underlying topology and the specific numbers of agents of each type.
Providing meaningful bounds on the robustness of a jump-equilibrium is 
therefore rather cumbersome.

\begin{table}[t]
	\caption{Overview of robustness values of swap-equilibria for various graph 
		classes.
		For each considered class, a swap-equilibrium always exists (see 
		\citeA{bilo2020topological} and \citeA{Echzell0LMPSSS19} and 
		\Cref{sec:swap:topInfluence}). Moreover, there exists a Schelling 
		game on a  
		graph from this class 
		with two swap-equilibria whose robustness match the depicted lower and 
		upper 
		bound. For bounds marked 
		with $\dagger$, on a graph from this class, an equilibrium with 
		this robustness is guaranteed to exist in every Schelling game with at 
		least four agents of the one and two agents of the other~type.}
	\label{tab:results}
	\centering
	\begin{tabular}{l lll l}
			\toprule
			& \multicolumn{2}{c}{edge-robustness} &
			\multicolumn{2}{c}{vertex-robustness} \\
			\cmidrule{2-3} \cmidrule{4-5}
			& \makecell{lower \\ bound}& \makecell{upper 
				\\ bound} & \makecell{lower \\ bound}& \makecell{upper 
				\\ bound}\\
			\midrule
			Cliques (Ob. \ref{edgedel:clique}) & $0^\dagger$ & $0^\dagger$ &  
			$|V(G)|^\dagger$ &  $|V(G)|^\dagger$\\
			Cycles (Pr. \ref{swap:cycle}) & $0^\dagger$  & 
			$0^\dagger$ &  $0^\dagger$ &  $0^\dagger$ \\
			Paths (Th. \ref{path})& $0^\dagger$ & $|E(G)|^\dagger$ & $0^\dagger$ 
			& $|V(G)|^\dagger$ \\
			Grids (Th. \ref{grid})&  $0$ 
			& $1$ & $0$ & $1$ \\
			$\alpha$-Caterpillars (Pr. \ref{cater}) & $\alpha$ & 
			$|E(G)|^\dagger$ &  $\alpha$ & 
			$|V(G)|^\dagger$\\
			\bottomrule
	\end{tabular}
\end{table}

In our analysis of the robustness of swap-equilibria, we follow the approach 
from 
most previous works and 
investigate the influence of the structure of the 
topology~\cite{AGARWAL2021103576,bilo2020topological}. That 
is, we 
show tight
upper and 
lower bounds on the edge- and vertex-robustness of swap-equilibria in Schelling 
games on 
topologies 
from various 
graph classes, summarized in 
\Cref{tab:results}. We prove that the edge- and vertex-robustness of 
swap-equilibria on a graph class can be arbitrarily far apart, as 
on cliques it is always sufficient to delete a single edge to make every 
swap-equilibrium unstable while every swap-equilibrium remains stable after the 
deletion of any subset of vertices. In contrast to this, all of our other 
lower and upper bounds are the same (and tight) for vertex- and edge-robustness and can be proven using similar arguments. 
We further show that on paths 
there exists a swap-equilibrium that can be made unstable by deleting a single 
edge or vertex and a swap-equilibrium that remains stable after the deletion of 
any 
subset of edges or vertices, implying that the difference between the 
edge/vertex-robustness of the 
most and least robust equilibrium  
can be arbitrarily large. This suggests that, in practice, one should be cautious when dealing with equilibria if robustness is important.

As an example of a nontrivial graph class where every 
swap-equilibrium has 
robustness 
larger than zero, we define $\alpha$\textit{-star-constellation graphs} (see 
\Cref{sub:GC} for a definition). 
We show that every 
swap-equilibrium on an $\alpha$-star-constellation graph has edge- and 
vertex-robustness at least 
$\alpha$. 
Moreover, independent of our robustness notion, we obtain a precise characterization of swap equilibria on $\alpha$-star-constellation graphs and a polynomial-time algorithm for checking whether a swap equilibrium exists on a graph from this class. 
Lastly, we prove that a 
swap-equilibrium always exists on a subclass of 
$\alpha$\textit{-star-constellation 
	graphs} and caterpillar graphs which we call
$\alpha$-caterpillars (see \Cref{fig:ex}).

\section{Preliminaries}
\label{ch:prelims}
Let $\mathbb{N}$ be the set of positive integers and $\mathbb{N}_0$ the set 
of non-negative integers. For two integers $i<j\in \mathbb{N}_0$, we denote 
by~$[i,j]$ the set $\{i,i+1,\dots, j-1, j\}$ and by~$[j]$ the set $[1,j]$. 
Let~$G=(V,E)$ be an undirected graph. Then, 
$V(G)$ is the vertex set of $G$ and $E(G)$ is the edge set of $G$. For a 
subset $S\subseteq E$
of edges,  $G-S$ denotes the graph obtained from $G$ by 
deleting all edges from $S$. For a subset $V'\subseteq V$ of vertices, 
$G[V']$ denotes the graph $G$ induced by~$V'$. 
Overloading notation, for a subset $V'\subseteq V$, we sometimes write $G-V'$ 
to denote the graph $G$ induced by $V\setminus V'$, that is, $G-V'=G[V\setminus 
V']$. For a 
vertex $v\in V(G)$, we denote by $N_G(v)$ 
the set of vertices adjacent to $v$ in~$G$. The degree 
$\deg_G(v):=|N_G(v)|$ of $v$ is the 
number of vertices adjacent to~$v$ in~$G$. Lastly, $\Delta(G) := 
\max_{v \in V(G)} \deg_G(v)$ is the  maximum degree of a vertex in 
$G$.

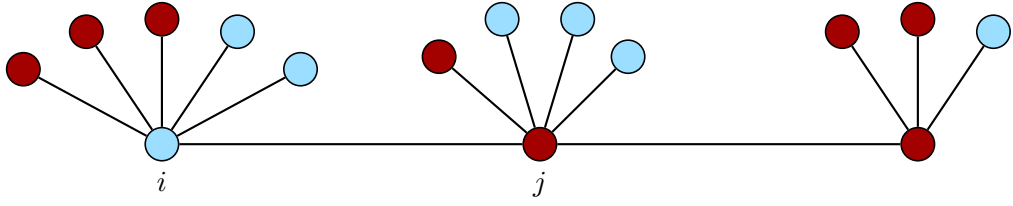
\begin{figure}[t]
	\tikzstyle{alter}=[circle, minimum size=12.5pt, draw, inner sep=1pt, 
	semithick] 
	\tikzstyle{majarr}=[draw=black, thick]
	\centering
	\begin{tikzpicture}[auto]
		
		\node[alter, fill=c_T1, label=below:$i$] at (0ex,0ex) (a) {};
		
		\node[alter, fill=c_T2] at (-11ex,6ex) (a1) {};
		
		\node[alter, fill=c_T2] at (-6ex,9ex) (a2) {};
		
		\node[alter, fill=c_T2] at (0ex,10ex) (a3) {};
		
		\node[alter, fill=c_T1] at (6ex,9ex) (a4) {};
		
		\node[alter, fill=c_T1] at (11ex,6ex) (a5) {};
		
		\node[alter, fill=c_T2, label=below:$j$] at (30ex,0ex) (b) {};
		
		\node[alter, fill=c_T2] at (22ex,7ex) (b1) {};
		
		\node[alter, fill=c_T1] at (27ex,10ex) (b2) {};
		
		\node[alter, fill=c_T1] at (33ex,10ex) (b3) {};
		
		\node[alter, fill=c_T1] at (37ex,7ex) (b4) {};
		
		\node[alter, fill=c_T2] at (60ex,0ex) (c) {};
		
		\node[alter, fill=c_T2] at (54ex,9ex) (c1) {};
		
		\node[alter, fill=c_T2] at (60ex,10ex) (c2) {};
		
		\node[alter, fill=c_T1] at (66ex,9ex) (c3) {};
		
		\draw[majarr] (a) edge  (a1);
		\draw[majarr] (a) edge  (a2);
		\draw[majarr] (a) edge  (a3);
		\draw[majarr] (a) edge  (a4);
		\draw[majarr] (a) edge  (a5);
		\draw[majarr] (b) edge  (b1);
		\draw[majarr] (b) edge  (b2);
		\draw[majarr] (b) edge  (b3);
		\draw[majarr] (b) edge  (b4);
		\draw[majarr] (c) edge  (c1);
		\draw[majarr] (c) edge  (c2);
		\draw[majarr] (c) edge  (c3);
		\draw[majarr] (a) edge  (b);
		\draw[majarr] (b) edge  (c);
		
		\end{tikzpicture}
	\caption{\label{fig:ex} Schelling game with $|T_1|=7$ and $|T_2|=8$. 
		Agents from $T_1$ are drawn in light blue and agents from $T_2$ 
		in dark red. Let $\mathbf{v}$ be the depicted assignment. It holds that 
		$u_i(\mathbf{v})=u_j(\mathbf{v})=\frac{1}{3}$ and $u_i(\mathbf{v}^{i 
			\leftrightarrow 
			j})=u_j(\mathbf{v}^{i \leftrightarrow j})=\frac{1}{2}$. Thus, 
		$\mathbf{v}$ 
		is not a swap-equilibrium, as $i$ and $j$ have a profitable swap. The 
		displayed graph is a $2$-star-constellation graph and a 
		$2$-caterpillar.} 
\end{figure}

\subsection{Schelling Games} \label{se:SG}
A \textit{Schelling game} is defined by a set $N = [n]$ of 
$n \in\NN$ (strategic) agents partitioned into two
\textit{types} 
$T_{1}$ and $T_{2}$ and an undirected graph $G=(V, E)$ 
with $|V| \geq n$, called the \textit{topology}. The strategy of agent  $i \in 
N$  consists of picking some \textit{position} $v_i \in V$ with~$v_i \neq 
v_j$ for $i\neq j \in N$. The \textit{assignment vector} 
$\mathbf{v}=\left(v_{1}, \ldots, v_{n}\right)$ contains the positions of all 
agents.  A vertex $v \in V$ is called \textit{unoccupied} in $\mathbf{v}$ if 
$v \neq v_i$ 
for all~$i \in N$. In the following, we refer to an agent~$i$ and its position 
$v_i$ interchangeably. For example, we say  agent $i$ has an edge to agent $j$ 
if $\{v_i, v_j\} \in E$. For an agent~$i \in T_l$ with $l\in 
\{1,2\}$, we call 
all other agents $F_{i}= T_l \setminus\{i\}$ of the same type  
\textit{friends} of $i$. The set of $i$'s \textit{neighbors} on topology $G$ is 
$N^G_{i}(\mathbf{v})\coloneqq\left\{j \neq i\mid \left\{v_{i}, v_{j}\right\} 
\in\right. E(G)\}$, and~$a^G_i(\mathbf{v})\coloneqq|N^G_{i}(\mathbf{v}) \cap 
F_{i}|$ is the number of friends in the neighborhood of $i$ in 
$\mathbf{v}$.

Given an assignment $\mathbf{v}$, the utility of agent $i$ on topology $G$ 
is: 
\begin{equation*}
u_i^{G}(\mathbf{v}) \coloneqq \begin{cases}
0 &\text{if } N^G_{i}(\mathbf{v}) = \emptyset \text{,}\\
\nicefrac{a^G_i(\mathbf{v})}{|N^G_{i}(\mathbf{v})|}& \text{otherwise.}
\end{cases}
\end{equation*}
If the topology is clear from the context,  we omit the superscript $G$.

Given some assignment $\mathbf{v}$, agent $i \in N$, and an unoccupied vertex~$v$, we denote by~$\mathbf{v}^{i \rightarrow v}=(v^{i \rightarrow v}_1,\dots, 
v^{i \rightarrow v}_n)$ the assignment obtained from  
$\mathbf{v}$ where $i$ \emph{jumps} to $v$, that is, $v_{i}^{i \rightarrow 
	v}=v$ and $v_{j}^{i \rightarrow v}=v_j$ for all $j\in N\setminus \{i\}$.
A jump of an agent $i$ to vertex $v$ is called \textit{profitable} 
if it improves $i$'s utility, that is, $u_{i}\left(\mathbf{v}^{i \rightarrow 
	v}\right)>u_{i}(\mathbf{v})$.
Note that an agent can jump to any unoccupied vertex .
For two agents~$i\neq j \in N$ and some assignment~$\mathbf{v}$, we define 
$\mathbf{v}^{i \leftrightarrow j}=(v^{i \leftrightarrow j}_1,\dots, 
v^{i \leftrightarrow j}_n)$ as the assignment that is obtained by 
\emph{swapping} the vertices of $i$ and $j$, that is, $v_{i}^{i \leftrightarrow 
	j}=v_j$, $v_{j}^{i \leftrightarrow j}=v_i$, and~$v_{k}^{i \leftrightarrow 
	j}=v_k$ 
for all $k\in N\setminus \{i,j\}$.
Note that any two agents (independent of the vertices they occupy) can perform a swap.
A 
swap of two agents $i,j \in N$ is called \textit{profitable} 
if it improves $i$'s and $j$'s utility, that is, $u_{i}\left(\mathbf{v}^{i 
	\leftrightarrow j}\right)>u_{i}(\mathbf{v})$ and $u_{j}\left(\mathbf{v}^{i 
	\leftrightarrow j}\right)>u_{j}(\mathbf{v})$ (see \Cref{fig:ex} for an example).
Note that a swap between agents of the same type is never profitable.

An assignment $\mathbf{v}$ is a \emph{jump/swap-equilibrium} if no profitable 
jump/swap exists.
Note that a jump-equilibrium is simply a Nash equilibrium for our Schelling game.
Following literature conventions, in cases where we allow 
agents to 
swap, 
we assume that $n=|V(G)|$, while in cases where we allow agents to jump, we 
assume that~$n<|V(G)|$. 

\subsection{Graph Classes} \label{sub:GC}
A \textit{path} of length $n$ is a graph $G=(V,E)$ with $V=\{v_1, \dots ,v_n\}$ 
and $E=\{ \{v_i, v_{i+1}\} \mid i \in [n-1]\}$.
A \textit{cycle} of length $n$ is a graph $G=(V,E)$ with $V=\{v_1, \dots, 
v_n\}$ 
and $E=\{ \{v_i, v_{i+1}\} \mid i \in [n-1]\} \cup \{\{v_n,v_1\}\}$.
We call a graph $G=(V,E)$ where every pair of vertices is 
connected by an edge a \textit{clique}. 
For $x,y\geq 2$, we define the $(x \times y)$-\textit{grid} as the 
graph $G=(V,E)$ with $V=\{ (a,b) 
\in \NN \times \NN \mid a\leq x, b\leq y\}$ and $E=\{\{(a,b),(c,d)\} \mid 
|a-c|+|b-d|=1\}$.
An  $x$\textit{-star} with $x \in \NN$ is a graph $G=(V,E)$ with $V=\{v_0, 
\dots, v_x\}$ and $E=\{ \{v_0,v_i\} \mid i \in [x]\}$; the vertex 
$v_0$ is called the \textit{central vertex} of the star. 
We say that a connected graph $G=(V,E)$ is an 
\textit{$\alpha$-star-constellation graph} for some $\alpha \in \NN_0$ if  it 
holds 
for all 
$v \in V$ with $\deg_G(v) > 1$ that $|\{ w \in N_G(v) \mid \deg_G(w)=1\}| \geq 
|\{ w \in N_G(v) \mid \deg_G(w)>1\}| + \alpha$. That is, the graph $G$ consists 
of stars where the central vertices can be connected by edges 
such that every central vertex is adjacent to at least $\alpha$ more degree-one 
vertices  than other central vertices. Thus, an $\alpha$-star constellation graph consists of (connected) stars forming a \emph{constellation} of stars, which gives this class its name. An \emph{$\alpha$-caterpillar} is an 
$\alpha$-star-constellation graph where the graph restricted to non-degree-one 
vertices forms a path  (see \Cref{fig:ex} for an example).

\section{NP-Hardness of Equilibria Existence}\label{ch:np_nostubborn}
In this section, we prove the NP-hardness of the following two problems: 

\problemdef{Swap/(Jump)-Equilibrium [S/(J)-Eq]}
{A connected topology $G$ and a set $N=[n]$ of agents with 
	$|V(G)|=n$ 
	\big($|V(G)|>n$\big) partitioned into types $T_1$ and $T_2$.}
{Is there an assignment $\mathbf{v}$ of agents to vertices where no 
	two agents have a 
	profitable swap (no agent has a profitable jump)?} 

\citeA{AGARWAL2021103576} proved that
deciding the 
existence of a swap- or jump-equilibrium in a Schelling game with 
\emph{stubborn} and \emph{strategic} agents is NP-hard. In their 
model, 
a strategic agent wants to maximize the fraction of agents of its type in its 
occupied
neighborhood (like the agents in our definition) and a stubborn agent has a 
favorite vertex which it never leaves (in our definition, such agents do 
not exist). 
Formally, 
\citeA{AGARWAL2021103576} proved the NP-hardness of the 
following problems:

\problemdef{Swap/(Jump)-Equilibrium with Stubborn Agents [S/(J)-Eq-Stub]}
{A connected topology $G$ and a set $N=[n]=R \dot\cup S$ of agents with 
	$|V(G)|=n$ 
	\big($|V(G)|>n$\big) partitioned into types $T_1$ and $T_2$ and a set $V_S 
	=\{s_i \in V(G) \mid i \in S\}$ of 
	vertices, where $R$ is the set of strategic and~$S$~the set of stubborn agents.}
{Is there an assignment $\mathbf{v}$ of agents to vertices with  
	$v_i=s_i$ for $i \in S$ such that no two strategic agents have a 
	profitable swap (no strategic agent has a profitable jump)?} 

Both known hardness reductions heavily rely on the existence of stubborn agents.
We show that it is 
possible to polynomial-time reduce \textsc{S/J-Eq-Stub} to \textsc{S/J-Eq}. In the following 
two 
subsections, we first prove this for swap-equilibria and 
afterwards consider jump-equilibria.

\subsection{Swap-Equilibria}\label{sub:sq}
This subsection is devoted to proving the following theorem: 

\begin{restatable}{theorem}{seqhard}
\label{swapHardness:noStubborn}
\sEq{} is NP-complete. 
\end{restatable}

To prove that \sEq{} is NP-hard, we devise a polynomial-time many-one reduction from a 
restricted version of \stubsEq:
\begin{lemma}\label{le:stEqEx}
	\stubsEq{} remains NP-hard even on instances satisfying the following 
	two conditions:
	\begin{enumerate}
		\item For every vertex $v \notin V_S$ not occupied by a stubborn 
		agent there 
		exist two adjacent vertices $s_i, s_j \in V_S$ occupied by stubborn agents $i 
		\in T_1$ and $j \in T_2$.
		\item For both types there are at least five strategic and three stubborn 
		agents.
	\end{enumerate}
\end{lemma}
\begin{proof}
	The lemma follows directly from the reduction by  \citeA{AGARWAL2021103576}, since, in their reduction, all constructed 
	instances satisfy both 
	properties.  To verify that these properties hold, we give a short 
	description of their construction below.
	
	The reduction is from \textsc{Clique}. An instance of \textsc{Clique} 
	consists of a connected graph $H=(X,Y)$ and an integer $\lambda$. Without loss 
	of generality, they assume that $\lambda>5$. Given an instance of 
	\textsc{Clique}, 
	they construct a Schelling game with agents $N=R \dot{\cup} S$ partitioned into 
	types $T_1$ and $T_2$ on a topology $G$. $R$ contains $\lambda$ strategic 
	agents of 
	type $T_1$ and $|X|+5$  strategic agents of type $T_2$. Note that we thus have 
	at least five strategic agents from each type. The stubborn agents will be 
	defined together with the topology.
	
	The graph $G$ consists of three 
	subgraphs $G_1,G_2$ and $G_3$, which are connected by single edges. $G_1$ is an 
	extended 
	copy of the given graph $H$ with added degree-one vertices $W_v$ for every $v 
	\in X$. The vertices in $W_v$ are occupied by stubborn agents from both types. 
	$G_2$ is a complete bipartite graph, 
	where one of the partitions is fully occupied by stubborn agents from both types. $G_3$ is constructed by adding degree-one vertices to a given
	tree $T$. For every vertex $v\in V(T)$, at least ten vertices are added, half of which are occupied by stubborn agents from $T_1$ and the other 
	half are occupied by stubborn agents from $T_2$. Thus, we have (more than) three stubborn agents of each type. Furthermore, it is easy to see that every vertex not occupied by a stubborn agent is adjacent to at least one stubborn agent from each type.
\end{proof}

\paragraph*{Idea Behind Our Reduction.}	
Given an 
instance of \stubsEq{} on a topology~$G'$ with vertices of stubborn 
agents $V_{S'}$, we construct a 
Schelling game without stubborn agents on a topology $G$ that 
simulates the given game. To create~$G$, we modify the graph~$G'$ by 
adding new vertices and connecting these new vertices to vertices 
from~$V_{S'}$. 
Moreover, we replace each stubborn agent by a strategic agent and add further 
strategic agents. In the construction, we ensure that if there exists a
swap-equilibrium~$\mathbf{v}'$ in the given game, then~$\mathbf{v}'$ can be extended to a swap-equilibrium $\mathbf{v}$  in the constructed game 
by replacing each stubborn agent with a strategic agent of the same type and 
filling empty vertices with further strategic agents. One particular 
challenge 
here is to ensure that the strategic agents that replace stubborn agents 
do not have a profitable swap in $\mathbf{v}$. For this, recall that, by \Cref{le:stEqEx},  we 
assume 
that in $G'$, every vertex not occupied by a stubborn agent in $\mathbf{v'}$ 
is 
adjacent to at least one stubborn agent of each type. Thus, in $\mathbf{v}'$ 
each strategic agent $i$ is always adjacent to at least one friend and has 
utility~$
u_i^{G'}(\mathbf{v'})\geq \nicefrac{1}{\Delta(G')}
$. Conversely, by swapping with agent $i$, an agent~$j$ of the other type can 
get 
utility at most
$ u_j^{G'}(\mathbf{v'}^{i \leftrightarrow j}) \leq 
\nicefrac{\Delta(G')-1}{\Delta(G')}$.
Our idea is now to ``boost'' the utility of a strategic agent $j$ that replaces 
a stubborn 
agent in $\mathbf{v}$ 
by adding enough degree-one neighbors only adjacent to $v_j$ in $G$, which we 
fill with agents of $j$'s type when extending $\mathbf{v}'$ to $\mathbf{v}$ 
such that $
u_j^{G}(\mathbf{v})\geq \nicefrac{\Delta(G')-1}{\Delta(G')} \geq 
u_j^{G}(\mathbf{\mathbf{v}}^{i \leftrightarrow j})$.

Moreover, we ensure that if there exists a 
swap-equilibrium $\mathbf{v}$ in the constructed game, then $\mathbf{v}$ 
restricted to $V(G')$, where some (strategic) agents are replaced by the 
designated stubborn 
agents of the same type, is a swap-equilibrium in the given game.
Note that the neighborhoods of all vertices in~$V(G')\setminus V_{S'}$ are the 
same in $G$ and $G'$ and thus every swap that is profitable in the assignment in the 
given game would also be 
profitable in $\mathbf{v}$. So the remaining challenge here is to
design $G$ in such a way that the 
vertices occupied by stubborn agents of some type in the input game have to be 
occupied by agents of the same type in every swap-equilibrium in the 
constructed game. We achieve this by introducing an asymmetry 
between the types in the construction.

\paragraph*{Construction.}
We are given an instance $\mathcal{I}'$ of \stubsEq{} consisting of a 
connected 
topology 
$G'$, a set of agents $[|V(G')|]=N'=R' \dot\cup S'$ partitioned into types 
$T_1'$ and~$T_2'$ with at least five strategic and three 
stubborn 
agents 
of each type, and a set $V_{S'} =\{s_i \in V(G') \mid i \in S'\}$ of vertices occupied by stubborn agents 
with each vertex $v \notin V_{S'}$ being adjacent to two vertices 
$s_i, s_j \in 
V_{S'}$ with $i \in T'_1$ and $j \in T'_2$. We denote the sets 
of vertices occupied by stubborn agents from $T_1'$ and $T_2'$ as $V_{S'_1}$ 
and 
$V_{S'_2}$, respectively. We construct an instance $\mathcal{I}$ of 
\sEq{} 
consisting of 
a 
topology $G=(V,E)$ and types $T_1$ and $T_2$ as follows. 

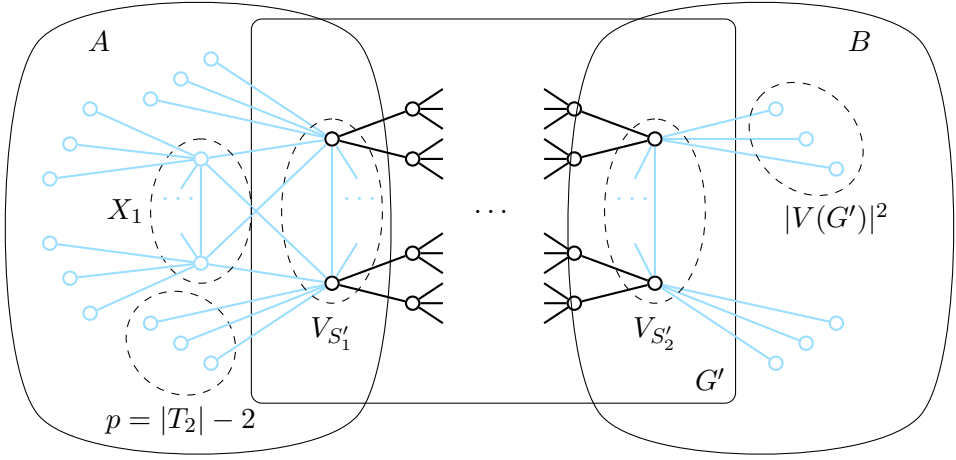
\begin{figure}[t!]
	\tikzstyle{alter}=[circle, minimum size=5pt, draw, inner sep=1pt, thick] 
	\tikzstyle{majarr}=[draw=black, thick]
	\centering
		\begin{tikzpicture}[scale=0.8]
		\draw[rounded corners] (0, 0) rectangle (0.5\textwidth,0.4\textwidth) 
		node[yshift=-0.3\textwidth, xshift=-2ex]{$G'$};
		
		\node[alter, align=center] at (8ex,12ex) (a) {};
		
		\node[alter, color=seagreen] at (-10ex,8ex) (a_1) {};
		\node[alter, color=seagreen] at (-7ex,6ex) (a_2) {};
		\node[alter, color=seagreen] at (-4ex,4ex) (a_3) {};
		\fitellipsis{a_1}{a_3}{ dashed}
		\node[ yshift=-6.2ex] at ($(a_1)!0.5!(a_3)$) {$p=|T_2|-2$};
		
		\draw[majarr, color=seagreen] (a) edge (a_1);
		\draw[majarr, color=seagreen] (a) edge (a_2);
		\draw[majarr, color=seagreen] (a) edge (a_3);
		
		\node[alter] at (8ex,\textwidth*0.4-12ex) (b) {};
		
		\node[alter, color=seagreen] at (-10ex,\textwidth*0.4-8ex) (b_1) {};
		\node[alter, color=seagreen] at (-7ex,\textwidth*0.4-6ex) (b_2) {};
		\node[alter, color=seagreen] at (-4ex,\textwidth*0.4-4ex) (b_3) {};
		
		\draw[majarr, color=seagreen] (b) edge (b_1);
		\draw[majarr, color=seagreen] (b) edge (b_2);
		\draw[majarr, color=seagreen] (b) edge (b_3);
		
		\fitellipsis{a}{b}{ dashed}
		\node[ yshift=-4ex] at (a) {$V_{S'_1}$};
		
		\draw[majarr, color=seagreen] (a) edge (b);
		\draw[majarr, color=seagreen] (a) edge ($(a)+(2.5ex,+4ex)$);
		\draw[majarr, color=seagreen] (b) edge ($(b)+(2.5ex,-4ex)$);
		\node[color=seagreen] at ($(b)+(3ex,-6ex)$) {$\dots$};
		
		\foreach \x in {a,b} {
			\node[alter] at ($(\x)+(8ex,+3ex)$) (g_\x_1) {};
			\node[alter] at ($(\x)+(8ex,-2ex)$) (g_\x_2) {};
			
			\foreach \y in {g_\x_1,g_\x_2} {
				
				\draw[majarr] (\x) edge (\y);
				
				\draw[majarr] (\y) edge ($(\y)+(3ex,+2ex)$);
				\draw[majarr] (\y) edge ($(\y)+(3ex,0ex)$);
				\draw[majarr] (\y) edge ($(\y)+(3ex,-2ex)$);
			} 
		}
		
		\node[alter, color=seagreen] at (-5ex,14ex) (x) {};
		\node[alter, color=seagreen] at (-5ex,\textwidth*0.4-14ex) (y) {};
		
		\draw[majarr, color=seagreen] (x) edge (y);
		\draw[majarr, color=seagreen] (x) edge (a);
		\draw[majarr, color=seagreen] (x) edge (b);
		\draw[majarr, color=seagreen] (y) edge (a);
		\draw[majarr, color=seagreen] (y) edge (b);
		
		\draw[majarr, color=seagreen] (x) edge ($(x)+(-2ex,+3ex)$);
		\draw[majarr, color=seagreen] (y) edge ($(y)+(-2ex,-3ex)$);
		\node[color=seagreen] at ($(y)+(-2ex,-4ex)$) {$\dots$};
		
		\fitellipsis{x}{y}{ dashed}
		\node[ xshift=-6ex] at ($(x)!0.5!(y)$) {$X_1$};
		
		\node[alter, color=seagreen] at (-20ex,16ex) (x_1) {};
		\node[alter, color=seagreen] at (-18ex,12.5ex) (x_2) {};
		\node[alter, color=seagreen] at (-16ex,9ex) (x_3) {};
		
		\draw[majarr, color=seagreen] (x) edge (x_1);
		\draw[majarr, color=seagreen] (x) edge (x_2);
		\draw[majarr, color=seagreen] (x) edge (x_3);
		
		\node[alter, color=seagreen] at (-20ex,\textwidth*0.4-16ex) (y_1) {};
		\node[alter, color=seagreen] at (-18ex,\textwidth*0.4-12.5ex) (y_2) {};
		\node[alter, color=seagreen] at (-16ex,\textwidth*0.4-9ex) (y_3) {};
		
		\draw[majarr, color=seagreen] (y) edge (y_1);
		\draw[majarr, color=seagreen] (y) edge (y_2);
		\draw[majarr, color=seagreen] (y) edge (y_3);
		
		\node[alter] at (\textwidth*0.5-8ex,12ex) (c) {};
		
		\node[alter, color=seagreen] at (\textwidth*0.5+10ex,8ex) (c_1) {};
		\node[alter, color=seagreen] at (\textwidth*0.5+7ex,6ex) (c_2) {};
		\node[alter, color=seagreen] at (\textwidth*0.5+4ex,4ex) (c_3) {};
		
		\draw[majarr, color=seagreen] (c) edge (c_1);
		\draw[majarr, color=seagreen] (c) edge (c_2);
		\draw[majarr, color=seagreen] (c) edge (c_3);
		
		\node[alter] at (\textwidth*0.5-8ex,\textwidth*0.4-12ex) (d) {};
		
		\node[alter, color=seagreen] at 
		(\textwidth*0.5+10ex,\textwidth*0.4-15ex) 
		(d_1) {};
		\node[alter, color=seagreen] at 
		(\textwidth*0.5+7ex,\textwidth*0.4-12ex) 
		(d_2) 
		{};
		\node[alter, color=seagreen] at 
		(\textwidth*0.5+4ex,\textwidth*0.4-9ex) 
		(d_3) 
		{};
		
		\draw[majarr, color=seagreen] (d) edge (d_1);
		\draw[majarr, color=seagreen] (d) edge (d_2);
		\draw[majarr, color=seagreen] (d) edge (d_3);
		\fitellipsis{d_1}{d_3}{ dashed}
		\node[ yshift=-4ex] at (d_1) {$|V(G')|^2$};
		
		\fitellipsis{c}{d}{ dashed}
		\node[ yshift=-4ex] at (c) {$V_{S'_2}$};
		
		\draw[majarr, color=seagreen] (c) edge (d);
		\draw[majarr, color=seagreen] (c) edge ($(c)+(-2ex,+4ex)$);
		\draw[majarr, color=seagreen] (d) edge ($(d)+(-2ex,-4ex)$);
		\node[color=seagreen] at ($(d)+(-2ex,-6ex)$) {$\dots$};
		
		\foreach \x in {c,d} {
			\node[alter] at ($(\x)+(-8ex,+3ex)$) (g_\x_1) {};
			\node[alter] at ($(\x)+(-8ex,-2ex)$) (g_\x_2) {};
			
			\foreach \y in {g_\x_1,g_\x_2} {
				
				\draw[majarr] (\x) edge (\y);
				
				\draw[majarr] (\y) edge ($(\y)+(-3ex,+2ex)$);
				\draw[majarr] (\y) edge ($(\y)+(-3ex,0ex)$);
				\draw[majarr] (\y) edge ($(\y)+(-3ex,-2ex)$);
			} 
		}
		
		\node[] at (\textwidth*0.25, \textwidth*0.2) {$\dots$};
		
		\draw [] plot [smooth cycle] coordinates { (-3.5,-0.2) 
			(-3.5,6) 
			(1.75,6) (1.75,-0.2)   };
		\draw [] plot [smooth cycle] coordinates { (5.75,-0.2) 
			(5.75,6) 
			(11,6) (11,-0.2)   };
		\node[] at (-2.5,6) {{$A$}};
		\node[] at (10, 6) {{$B$}};
		\end{tikzpicture}
	\caption{The constructed topology $G$. Modifications made to the given 
		graph 
		$G'$ are colored light blue. Note that 
		the vertices in $V_{S'_1} \cup X_1$ and $V_{S'_2}$ each form a clique.}
	\label{np:G}
\end{figure}
The graph $G$ (sketched in \Cref{np:G}) is an extended copy of 
the given graph~$G'$ and contains all vertices and edges from $G'$. We add 
three 
sets of vertices $M_1, X_1$, and $M_2$ as specified below. For every vertex $v 
\in V_{S'_2}$, we insert $|V(G')|^2$ degree-one 
vertices only adjacent to $v$ and add them to $M_2$. We connect the
vertices in $ V_{S'_2}$ to form a clique.
Let $$q \coloneqq \Delta(G')+|V(G')|^2 + |V_{S'_2}|$$ and note that~$q$ is an upper 
bound on the degree of a vertex from $V_{S'_2}$.
Let $X_1$ be a set of $s -|V_{S'_1}|$ vertices, where $$s\coloneqq q \cdot 
(|T'_2|+|M_2| + \Delta(G) )   + 1$$ (note that $s> |V_{S'_1}|$). Thus, 
$|V_{S'_1} \cup X_1|=s$ (we 
use this property to introduce the mentioned asymmetry between the two types).
We connect the vertices in $V_{S'_1} \cup X_1$ 
to form a clique. Let $$p 
\coloneqq |T_2'|+|M_2| -2 >\Delta(G')^2$$ (this choice of $p$ is 
important to ensure that vertices in $V_{S'_1}$ are occupied by 
agents from~$T_1$ in each swap-equilibrium of the constructed game). For every vertex $v \in V_{S'_1} 
\cup 
X_1$, we insert $p$ degree-one vertices only adjacent to $v$ and add them to 
$M_1$. Notably, the neighborhood of all vertices in $V(G')\setminus V_{S'}$ is 
the same in $G'$ and $G$.

The set $N =T_1 \dot \cup T_2$ of agents is defined as follows. We have 
$|T_1|=|T'_1|+|M_1|+|X_1|$ agents in $T_1$ and $|T_2|=|T'_2|+|M_2|$ agents in 
$T_2$. By the construction of $X_1$ above, we have that $|V_{S'_1} \cup X_1|=q 
\cdot (|T'_2|+|M_2| + \Delta(G') )  + 1 = q \cdot (|T_2|+ \Delta(G'))+1$. It 
also holds that $p = |T_2'|+|M_2| -2 = |T_2| -2$.

\paragraph*{Proof of Correctness.}
Let $A \coloneqq M_1 \cup 
V_{S'_1} 
\cup X_1$ and $B \coloneqq M_2 \cup V_{S'_2}$ (the vertices from $A$ should be 
occupied by agents from $T_1$, while the vertices from $B$ should be occupied by 
agents from $T_2$).
We start with showing the (easier) forward 
direction of the correctness of the reduction: 

\begin{lemma}\label{swapEqEx:For}
	If the given instance $\mathcal{I}'$ of \stubsEq{} admits a 
	swap-equilibrium, then the 
	constructed 
	instance $\mathcal{I}$ of \sEq{} admits a swap-equilibrium.
\end{lemma}
\begin{proof}
	Assume that there exists a 
	swap-equilibrium $\mathbf{v}'$ in the given instance $\mathcal{I}'$ with 
	stubborn agents. 
	Note that in $\mathbf{v}'$, the vertices in $V_{S'_1}$ and $V_{S'_2}$ are 
	occupied by 
	stubborn 
	agents
	from 
	$T_1'$ and $T_2'$, respectively. We now define an assignment $\mathbf{v}$ 
	for the 
	Schelling 
	game $\mathcal{I}$ without stubborn agents  and prove that it is a 
	swap-equilibrium: In $\mathbf{v}$, the 
	vertices in $V(G) \cap V(G')$ are occupied by agents of the same type as 
	the 
	agents 
	in $\mathbf{v}'$. The vertices in $X_1$ are occupied by agents from $T_1$. 
	For 
	$t \in \{1,2\}$, the added degree-one vertices in $M_t$ are occupied by 
	agents 
	from $T_t$. Hence, the vertices in $A=M_1 \cup V_{S'_1} \cup X_1$ are 
	occupied 
	by 
	agents from $T_1$ and the vertices in $B=M_2 \cup V_{S'_2}$ are occupied by 
	agents 
	from $T_2$. Note that in $\mathbf{v}$, exactly $|T_1|$ agents from $T_1$ 
	and $|T_2|$ agents from $T_2$ are assigned as we have 
	$|T_1|=|T'_1|+|M_1|+|X_1|$ and $|T_2|=|T'_2|+|M_2|$.
	
	Next, we prove that $\mathbf{v}$ is  a swap-equilibrium on $G$ in 
	$\mathcal{I}$ by showing 
	that 
	no profitable swap exists. We first observe that the utility of an agent 
	$i\in T_1$ on 
	a vertex $w \in M_1 \cup X_1$ is $u_i(\mathbf{v})=1$, since 
	$N_G(w)\subseteq A$ 
	and all agents in $A$ are from $T_1$.  The same holds analogously for 
	all
	agents on a vertex in $M_2$. Therefore, the agents on vertices in $M_1 \cup 
	M_2 
	\cup X_1$ cannot be involved in a profitable swap.  Note further that 
	the 
	utility 
	of any agent on a vertex $v \in V(G') \setminus V_{S'}$  is the same in 
	$\mathbf{v}$ and $\mathbf{v}'$, as the neighborhood of $v$ is identical in 
	$G$ and $G'$ and all vertices in the neighborhood are occupied 
	by  agents of the same type in $\mathbf{v}$ and $\mathbf{v}'$. Since 
	$\mathbf{v}'$ is a swap-equilibrium, 
	a 
	profitable swap must therefore involve at least one agent $i$ with $v_i \in 
	V_{S'}$.
	
	Let $Y=V_{S'_1} \cup X_1$ if $v_i \in V_{S'_1}$ and $Y=V_{S'_2}$ otherwise. 
	Denote the 
	number of degree-one neighbors adjacent to $v_i$ that were added to $G'$ in 
	the construction of $G$ by $x$. By 
	construction 
	of $G$, it holds that $x>\Delta(G)^2$ and that the agent $i$ is, among 
	others, adjacent to the 
	vertices 
	in $Y \setminus \{v_i\}$ and the $x>\Delta(G')^2$ degree-one neighbors 
	(all these vertices are occupied by friends). As the only other neighbors of 
	$v_i$ need to come from $V(G')\setminus V_{S'}$, it holds that $\deg_G(v_i) 
	\leq x + |Y 
	\setminus 
	\{v_i\}| + \Delta(G')$. Thus, the utility of agent $i$ on $v_i$ is:
	$$
	u_i^{G}(\mathbf{v})\geq \frac{x + |Y \setminus \{v_i\}|}{x + |Y \setminus 
		\{v_i\}| + \Delta(G')} > \frac{x}{x+\Delta(G')}> 
	\frac{\Delta(G')^2}{\Delta(G')^2+\Delta(G')}=\frac{\Delta(G')}{\Delta(G')+1}.$$
	We now distinguish between swapping $i$ with an agent $j$ with $v_j\in 
	V_{S'}$ and with $v_j\in V(G')\setminus V_{S'}$ (note that this exhausts all cases 
	as 
	we have already argued above that all agents placed on newly added vertices can 
	never be part of a profitable swap).
	First, consider swapping $i$ with an agent $j$ of the other type with $v_j 
	\in 
	V_{S'}$. On vertex $v_j$, agent $i$ can at most have $\Delta(G')$ adjacent 
	friends, as all vertices that are connected to $v_j$ by edges added  
	in the construction (that are, vertices in $A$ or $B$) are 
	occupied by friends of $j$. 
	It holds that $\deg_G(v_j)\geq \Delta(G')^2$, since $v_j$ is 
	adjacent 
	to at least $\Delta(G')^2$ degree-one neighbors. Therefore, the swap can 
	not be 
	profitable, as
	\begin{align*}
	u_i^{G}(\mathbf{v})\geq \frac{\Delta(G')}{\Delta(G')+1}> 
	\frac{1}{\Delta(G')} 
	= 
	\frac{\Delta(G')}{\Delta(G')^2} \geq u_i^{G}(\mathbf{v}^{i \leftrightarrow 
		j}).
	\end{align*}
	Hence, consider swapping $i$ with an agent $j$ of the other type on $v_j 
	\in 
	V(G') \setminus V_{S'}$. Recall that by \Cref{le:stEqEx} we have assumed 
	that for every vertex $v \notin V_{S'}$ not occupied by a stubborn agent, there 
	exist two adjacent vertices $s_i, s_j \in V_{S'}$ occupied by stubborn agents 
	$i 
	\in T_1$ and $j \in T_2$. Since $v_j \in V(G') \setminus V_{S'}$, the agent 
	$j$ 
	is thus adjacent to at least one friend. More precisely, we have $
	u_j^{G'}(\mathbf{v}')\geq \frac{1}{\Delta(G')}$.
	As noted above, the neighborhood of $v_j$ is identical in $G$ and $G'$. We 
	therefore have $u_j^{G}(\mathbf{v})= u_j^{G'}(\mathbf{v}')$. Thus, by 
	swapping 
	with 
	agent $j$, agent $i$ can at most get the following~utility $
	u_i^{G}(\mathbf{v}^{i \leftrightarrow j})\leq 
	\frac{\Delta(G')-1}{\Delta(G')}$.
	It follows that swapping $i$ and $j$ cannot be profitable, as:
	\begin{align*}
	u_i^{G}(\mathbf{v})\geq 
	\frac{\Delta(G')}{\Delta(G')+1}>\frac{\Delta(G')-1}{\Delta(G')}\geq 
	u_i^{G}(\mathbf{v}^{i \leftrightarrow j})
	\end{align*}
	Summarizing, no profitable swap is possible and $\mathbf{v}$ is a 
	swap-equilibrium  for the constructed Schelling game $\mathcal{I}$.
\end{proof}

It remains to prove the backwards direction of the reduction. 
For this, we start by making some definitions. 
Afterwards, we prove in 
\Cref{swapEqEx:A} that in any swap-equilibrium $\mathbf{v}$ in $\mathcal{I}$ 
all vertices 
from $M_1 \cup 
V_{S'_1} 
\cup X_1$ are occupied by agents from $T_1$ and in \Cref{swapEqEx:B} that all 
vertices from $M_2 \cup 
V_{S'_2}$ are occupied by agents from $T_2$. Using this, we conclude the proof 
by proving the backwards 
direction of the correctness of the reduction in \Cref{swapEqEx:Back}.

Recall that $A= M_1 \cup 
V_{S'_1} 
\cup X_1$ and $B = M_2 \cup V_{S'_2}$. Observe that the subgraph $G[A]$ 
(see \Cref{np:G_A}) consists of $q \cdot (|T_2|+ \Delta(G')) +1$ stars which 
each have $|T_2|-1$ vertices. The  central vertices of these stars are 
connected such that they form a clique. As $q = 
\Delta(G')+|V(G')|^2 
+|V_{S'_2}| > |V_{S'_2}| \geq 3$, note that $G[A]$ consists of at least 
three stars.
Note that 
$V(G)\setminus A = (V(G') \setminus V_{S'}) \cup B = (V(G') \setminus V_{S'}) 
\cup 
(M_2 \cup V_{S'_2})$. The subgraph $G[V(G)\setminus A]$ (shown in 
\Cref{np:G_VA}) 
is 
connected, since every vertex in $V(G') \setminus V_{S'}$ is connected to a 
vertex in $V_{S'_2}$ (by our assumption concerning $\mathcal{I}'$ from 
\Cref{le:stEqEx}) and the vertices in $V_{S'_2}$ form a 
clique 
(by the construction of $G$). Additionally, all vertices in $M_2$ are adjacent 
to exactly one vertex in $V_{S'_2}$. We start by proving 
\Cref{swapEqEx:A,swapEqEx:B}, which state that in every swap-equilibrium, the 
vertices in $A$ and $B$ have to be occupied by agents from $T_1$ and $T_2$, 
respectively.

\begin{figure}[t!]
	\centering
	\begin{subfigure}[t]{0.475\textwidth}
		\tikzstyle{alter}=[circle, minimum size=12.5pt, draw, inner sep=1pt , semithick] 
		\tikzstyle{majarr}=[draw=black, thick]
		\centering
		\begin{tikzpicture}[auto]
		
		\node[alter, color=blue] at (2ex,0ex) (a) {};
		\node[alter, color=blue] at (13ex,0ex) (b) {};
		\node[alter, color=blue] at (-3ex,15ex) (c) {};
		\node[alter, color=blue] at (18ex,15ex) (d) {};
		\node[alter, color=blue] at (7.5ex,12ex) (e) {};
		
		\foreach \x in {a,b,c,d,e} {
			\foreach \y in {a,b,c,d,e} {
				\if\x \y\else
				\draw[majarr] (\x) edge  (\y);
				\fi
			}
		}
		
		\foreach \x in {c,d,e} {
			
			\node[alter, xshift=-4ex, yshift=6ex] at (\x) (\x_1) {};
			\node[alter, yshift=7ex] at (\x) (\x_2) {};
			\node[alter,xshift=4ex, yshift=6ex] at (\x) (\x_3) {};
			
			\draw[majarr] (\x) edge  (\x_1);
			\draw[majarr] (\x) edge  (\x_2);
			\draw[majarr] (\x) edge  (\x_3);
		}
		
		\foreach \x in {a,b} {
			
			\node[alter, xshift=-3.5ex, yshift=-5ex] at (\x) (\x_1) {};
			\node[alter, yshift=-6ex] at (\x) (\x_2) {};
			\node[alter, xshift=3.5ex, yshift=-5ex] at (\x) (\x_3) {};
			
			\draw[majarr] (\x) edge  (\x_1);
			\draw[majarr] (\x) edge  (\x_2);
			\draw[majarr] (\x) edge  (\x_3);
		}
		
		\node[color=blue, xshift=-5ex] at ($(a)!0.5!(c)$) { $S'_1 \cup X_1$};
		
		\fitellipsiss{c_1}{c_3}{ dashed}
		\node[ yshift=+6ex] at ($(c_1)!0.5!(c_3)$) {\small 
			$|T_2|-2$};
		
		\end{tikzpicture}
		
		\caption{The subgraph $G[A]$. Observe that $G[A]$ consists of $q \cdot 
			(|T_2|+ \Delta(G')) +1$ stars which each have $|T_2|-1$ vertices  and 
			that are connected such that the central vertices form a clique.  }
		\label{np:G_A}
	\end{subfigure}
	\hfill
	\begin{subfigure}[t]{0.475\textwidth}
		\tikzstyle{alter}=[circle, minimum size=12.5pt, draw, inner sep=1pt, semithick] 
		\tikzstyle{majarr}=[draw=black, thick]
		\centering
		\begin{tikzpicture}[auto]
		
		\node[alter] at (0ex,0ex) (a) {};
		\node[alter] at (0ex,10ex) (b) {};
		\node[alter] at (5ex,18ex) (c) {};
		
		\node[yshift=-5ex] at (a) {$V(G')\setminus V_{S'}$};
		
		\node[alter, color=blue] at (15ex,0ex) (x) {};
		\node[alter, color=blue] at (20ex,10ex) (y) {};
		\node[alter, color=blue] at (15ex,18ex) (z) {};
		
		\draw[ majarr,color=blue] (x) edge  (y);
		\draw[majarr,color=blue] (x) edge  (z);
		\draw[majarr,color=blue] (z) edge  (y);
		
		\node[color=blue, yshift=-5ex] at (x) {$S'_2$};
		
		\draw[majarr] (b) edge  (c);
		\draw[majarr] (a) edge  (x);
		\draw[majarr] (b) edge  (x);
		\draw[majarr] (b) edge  (y);
		\draw[majarr] (c) edge  (z);
		
		\foreach \x in {x,y,z} {
			
			\node[alter, xshift=+5ex, yshift=+4ex, color=red] at (\x) (\x_1) {};
			\node[alter, xshift=+5ex, color=red] at (\x) (\x_2) {};
			\node[alter, xshift=+5ex, yshift=-4ex, color=red] at (\x) (\x_3) {};
			
			\draw[majarr, color=red] (\x) edge  (\x_1);
			\draw[majarr, color=red] (\x) edge  (\x_2);
			\draw[majarr, color=red] (\x) edge  (\x_3);
		}
		
		\node[color=red, yshift=-2ex, xshift=4ex] at (x_3) {$M_2$};
		
		\end{tikzpicture}
		
		\caption{Visualization of the subgraph $G[V(G)\setminus A]$. Note that  
			$G[V(G)\setminus A]$ 
			is connected, since every vertex in $V(G') \setminus V_{S'}$ is 
			adjacent to at least one vertex in $S'_2$. }
		\label{np:G_VA}
	\end{subfigure}
	\caption{Schematic visualization of the induced subgraphs $G[A]$ and 
		$G[V(G) 
		\setminus 
		A]$. Recall that 
		$A = S'_1 \cup M_1 \cup X_1$. }
\end{figure}
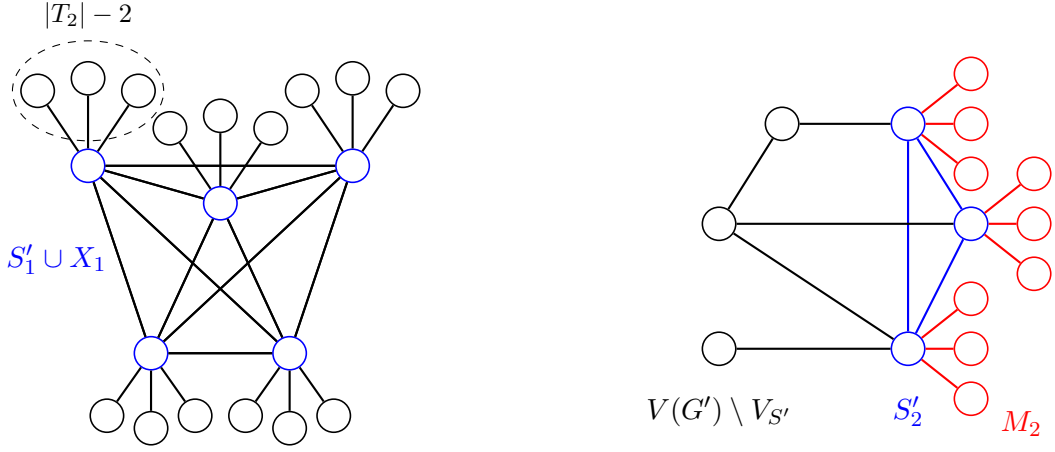

\begin{lemma}\label{swapEqEx:A}
	In any swap-equilibrium $\mathbf{v}$ in $\mathcal{I}$, all vertices in $A$ 
	are occupied by 
	agents from~$T_1$.
\end{lemma}
\begin{proof}
	For the sake of contradiction, assume that $\mathbf{v}$ is a swap-equilibrium 
	in $\mathcal{I}$ where $x>0$ agents from $T_2$ are placed on vertices from 
	$A$. We distinguish the following three cases based on the number of agents 
	from $T_2$ in $A$ and prove that $\mathbf{v}$ cannot be stable.
	\paragraph{Case 1:} Assume that $x<|T_2|-1$. Since all stars in $G[A]$ have 
	$|T_2|-1>x$ vertices, at least one of the stars has to contain agents from both 
	types. Thus, there exists an agent $i \in T_t$ for some $t \in \{1,2\}$ on a 
	degree-one vertex in $A$ with $u_i(\mathbf{v})=0$. Since $|T_2|-x>0$ and 
	$|T_1|>|A|$, there have to be agents from both $T_1$ and $T_2$ in 
	$G[V(G)\setminus 
	A]$. As observed before, $G[V(G)\setminus A]$ is connected. Thus, there exist 
	agents $i' \in T_t$ and $j' \in T_{t'}$  with $t' \neq t$ in $G[V(G)\setminus 
	A]$ 
	that are adjacent and hence have $u_{i'}(\mathbf{v})<1$ and 
	$u_{j'}(\mathbf{v})<1$. Then, swapping $i$ and $j'$ is profitable, since we 
	have $u_i(\mathbf{v})=0<u_i(\mathbf{v}^{i \leftrightarrow j'})$ and 
	$u_{j'}(\mathbf{v})<1=u_{j'}(\mathbf{v}^{i \leftrightarrow j'})$.
	
	\paragraph{Case 2:} Assume that $x=|T_2|$. Recall that all stars in $G[A]$ 
	have $|T_2|-1<x$ vertices.  Thus, there have to be agents from $T_2$ on at 
	least two stars. Since it also holds that $x< 2\cdot (|T_2|-1)$, there are 
	agents from both types on at least one of the stars.  Let $v_j$ be the central 
	vertex of this star occupied by some agent $j$. There exists an agent $i \in 
	T_t$ for some $t \in  \{1,2\}$ on a degree-one vertex adjacent to $v_j$ with 
	$u_{i}(\mathbf{v})=0$. The agent $j$ is from type $T_{t'}$ with $t' \neq t$. As 
	noted before,  $G[A]$ contains at least $3$ stars. Since $x< 2\cdot (|T_2|-1)$, 
	there also have to be agents from $T_1$ on at least two of the stars.  We make 
	a case distinction based on the types of the agents on central vertices in $G[A]$.
	
	First, suppose that all central vertices are occupied by agents from $T_{t'}$. 
	As noted above, it holds for both types that agents of this type occupy 
	vertices from at least two stars. Hence, there exists an agent $i' \in T_t$ on 
	a degree-one vertex adjacent to a central vertex $w \neq v_j$ with 
	$u_{i'}(\mathbf{v})=0$. Let $j' \in T_{t'}$ be the agent on $w$. Then, swapping 
	$i$ and $j'$ is profitable, since $u_{j'}(\mathbf{v})<1=u_{j'}(\mathbf{v}^{i 
		\leftrightarrow j'})$ and $u_{i}(\mathbf{v})=0<u_{i}(\mathbf{v}^{i 
		\leftrightarrow j'})$.
	
	Now, consider that all central vertices different from $v_j$ are occupied by 
	agents from 
	$T_{t}$. Then, there exists an agent $j' \in T_{t'}$ with  
	$u_{j'}(\mathbf{v})=0$ on a degree-one vertex adjacent to a central 
	vertex $w \neq v_j$ that is occupied by an agent from $T_t$. Swapping $i$ and 
	$j'$ is profitable, since  $u_{j'}(\mathbf{v})=0<1=u_{j'}(\mathbf{v}^{i 
		\leftrightarrow j'})$ and $u_{i}(\mathbf{v})=0<1=u_{i}(\mathbf{v}^{i 
		\leftrightarrow j'})$.
	
	Therefore, the central vertices different from $v_j$ have to be occupied by 
	agents from 
	both types. That is, there exists an agent $j' \in T_{t'}$ on a central vertex 
	$v_{j'} \neq v_j$ and an agent $i' \in T_{t}$ on a central vertex $v_{i'} 
	\neq v_j$. We have $u_{j'}(\mathbf{v})<1$, since $v_{j'}$ is adjacent to 
	$v_{i'}$. Then, swapping $i$ and $j'$ is profitable, since 
	$u_{j'}(\mathbf{v})<1=u_{j'}(\mathbf{v}^{i \leftrightarrow j'})$ and 
	$u_{i}(\mathbf{v})=0<u_{i}(\mathbf{v}^{i \leftrightarrow j'})$.
	
	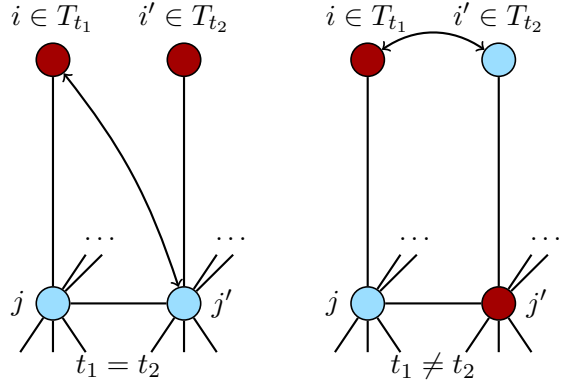
\begin{figure}[t]
		\tikzstyle{alter}=[circle, minimum size=12.5pt, draw, inner sep=1pt, semithick] 
		\tikzstyle{majarr}=[draw=black, thick]
		\centering
		\begin{tikzpicture}[auto,scale=1.3]
		
		\node[alter, fill=c_T1, label=left:$j$] at (0ex,0ex) (a) {};
		\node[alter, fill=c_T1, label=right:$j'$] at (8ex,0ex) (b) {};
		
		\node[alter, fill=c_T2, label=above:$i \in T_{t_1}$] at (0ex,15ex) (c) 
		{};
		\node[alter, fill=c_T2, label=above:$i' \in T_{t_2}$] at (8ex,15ex) (d) 
		{};
		
		\draw[majarr] (a) edge  (b);
		\draw[majarr] (a) edge  (c);
		\draw[majarr] (b) edge  (d);
		
		\draw[majarr] (a) edge  ($(a)+(2ex,-3ex)$);
		\draw[majarr] (a) edge  ($(a)+(-2ex,-3ex)$);
		\draw[majarr] (a) edge  ($(a)+(0ex,-3ex)$);
		
		\draw[majarr] (a) edge  ($(a)+(2ex,3ex)$);
		\draw[majarr] (a) edge  ($(a)+(3ex,3ex)$);
		\node[] at ($(a)+(3ex,4ex)$) {\small $\dots$};
		
		\draw[majarr] (b) edge  ($(b)+(2ex,-3ex)$);
		\draw[majarr] (b) edge  ($(b)+(-2ex,-3ex)$);
		\draw[majarr] (b) edge  ($(b)+(0ex,-3ex)$);
		
		\draw[majarr] (b) edge  ($(b)+(2ex,3ex)$);
		\draw[majarr] (b) edge  ($(b)+(3ex,3ex)$);
		\node[] at ($(b)+(3ex,4ex)$) {\small $\dots$};
		\node[yshift=-5ex] at ($(a)!0.5!(b)$) {$t_1 = t_2$};
		\draw[<->, thick , thick , bend left=10] (c) to (b);
		
		\tikzset{shift={(0.2\textwidth,0)}}
		
		\node[alter, fill=c_T1, label=left:$j$] at (0ex,0ex) (a) {};
		\node[alter, fill=c_T2, label=right:$j'$] at (8ex,0ex) (b) {};
		
		\node[alter, fill=c_T2, label=above:$i \in T_{t_1}$] at (0ex,15ex) (c) 
		{};
		\node[alter, fill=c_T1, label=above:$i' \in T_{t_2}$] at (8ex,15ex) (d) 
		{};
		
		\draw[majarr] (a) edge  (b);
		\draw[majarr] (a) edge  (c);
		\draw[majarr] (b) edge  (d);
		
		\draw[majarr] (a) edge  ($(a)+(2ex,-3ex)$);
		\draw[majarr] (a) edge  ($(a)+(-2ex,-3ex)$);
		\draw[majarr] (a) edge  ($(a)+(0ex,-3ex)$);
		
		\draw[majarr] (a) edge  ($(a)+(2ex,3ex)$);
		\draw[majarr] (a) edge  ($(a)+(3ex,3ex)$);
		\node[] at ($(a)+(3ex,4ex)$) {\small $\dots$};
		
		\draw[majarr] (b) edge  ($(b)+(2ex,-3ex)$);
		\draw[majarr] (b) edge  ($(b)+(-2ex,-3ex)$);
		\draw[majarr] (b) edge  ($(b)+(0ex,-3ex)$);
		
		\draw[majarr] (b) edge  ($(b)+(2ex,3ex)$);
		\draw[majarr] (b) edge  ($(b)+(3ex,3ex)$);
		\node[] at ($(b)+(3ex,4ex)$) {\small $\dots$};
		\node[yshift=-5ex] at ($(a)!0.5!(b)$) {$t_1 \neq t_2$};
		\draw[<->, thick , bend left=35] (c) to (d);
		
		\end{tikzpicture}
		
		\caption{\label{np:case3_1} Case 3 in the proof of \Cref{swapEqEx:A} (there 
			are 
			exactly $|T_2|-1$ agents from $T_2$ in $A$) where agents from $T_2$ occupy two 
			or more stars.}
	\end{figure}
	
	\paragraph{Case 3:} Assume that $x=|T_2|-1$. If the $x$ agents from $T_2$ 
	occupy vertices from two or more stars, then at least two stars contain agents 
	from both types (illustrated in \Cref{np:case3_1}). That is, there exists an 
	agent $i \in T_{t_1}$ for some $t_1 \in \{1,2\}$ with $u_{i}(\mathbf{v})=0$ on 
	a degree-one vertex adjacent to an agent $j \in T_{t'_1}$ with $t'_1 \neq t_1$. 
	The agent $j$ on the central vertex has $u_{j}(\mathbf{v})<1$. Without loss of 
	generality, assume that $t_1 = 1$ and  thus $t'_1 =2$. As argued above, another 
	star has to contain agents from both types. Thus, there exists another agent 
	$i' \in T_{t_2}$ for some $t_2 \in \{1,2\}$ with $u_{i}(\mathbf{v})=0$ on a 
	degree-one vertex adjacent to an agent $j' \in T_{t'_2}$ with $t'_2 \neq t_2$. 
	Again, the agent $j'$ on the central vertex has $u_{j'}(\mathbf{v})<1$. If $t_2 
	= 1$, then swapping $i$ and $j'$ is profitable, as
	$u_{i}(\mathbf{v})=0<u_{i}(\mathbf{v}^{i \leftrightarrow j'})$ and 
	$u_{j'}(\mathbf{v})<1=u_{j'}(\mathbf{v}^{i \leftrightarrow j'})$. Otherwise, if 
	$t_2 = 2$, then swapping $i$ and $i'$ is profitable, as 
	$u_{i}(\mathbf{v})=0<1=u_{i}(\mathbf{v}^{i \leftrightarrow i'})$ and 
	$u_{i'}(\mathbf{v})=0<1=u_{i'}(\mathbf{v}^{i \leftrightarrow i'})$. 
	Summarizing, if the $x$ agents from $T_2$ occupy vertices from two or more 
	stars, then $\mathbf{v}$ cannot be a swap-equilibrium.
	
	Hence, the agents from $T_2$ have to occupy all $|T_2|-1$ vertices of one of 
	the stars in~$A$. Let agent $i \in T_2$ be the agent on 
	the  central vertex $v_i$ of this star. Observe that $\deg_G(v_i)\geq 
	(|T_2|-2)+ q \cdot (|T_2|+ \Delta(G'))$, since $v_i$ is adjacent to $|T_2|-2$ 
	degree-one neighbors and the $ q \cdot (|T_2|+ \Delta(G'))$ vertices in 
	$(V_{S'_1} 
	\cup X_1) \setminus \{v_i\}$.  It follows that agent $i  \in T_2$ has utility: 
	\begin{align*}
	u_i(\mathbf{v})\leq \frac{|T_2|}{(|T_2|-2)+ q \cdot 
		(|T_2|+ \Delta(G'))}    < \frac{|T_2|}{|T_2|\cdot q} 
	= 
	\frac{1}{q}.                \end{align*}
	
	Since $x=|T_2|-1$, there is one agent $i' \in 
	T_2$ not placed on a vertex from $A$. As noted above, $G[V(G)\setminus A]$ is 
	connected. Thus, 
	agent $i'$ is adjacent to an agent $j \in T_1$ on $v_j \in V(G) \setminus A$. 
	Recall that $V(G)\setminus A = (V(G') \setminus V_{S'}) \cup (M_2 \cup 
	V_{S'_2})$. 
	If 
	$v_j \in V_{S'_2}$, then we have $\deg_G(v_j)\leq \Delta(G') + |V(G')|^2 + 
	|V_{S'_2}|$.  If $v_j \in M_2$, then we have $\deg_G(v_j)=1$. If $v_j \in  
	V(G') 
	\setminus V_{S'}$, then we have  $\deg_G(v_j)\leq \Delta(G')$. In any case, it 
	holds that $\deg_G(v_j)\leq \Delta(G') + |V(G')|^2 + |V_{S'_2}|=q$ ($q$ was 
	defined in the construction). Since $v_j$ 
	is 
	adjacent to $v_{i'}$, agent $i$ has at least one adjacent friend after swapping 
	to 
	$v_j$, and its utility is:                                                  
	\begin{align*}
	u_i(\mathbf{v}^{i  \leftrightarrow j })\geq 
	\frac{1}{q} > u_i(\mathbf{v}).                                             
	\end{align*}
	Thus, $i$ wants to swap with $j$.
	Note that at least one of the at most $q$ neighbors 
	of $j$ (specifically, agent $i'$) is not a friend of $j$.  Observe that the 
	neighborhood of the central vertex $v_i$ consists of $|T_2|-2$ degree-one 
	neighbors, the $ q 
	\cdot (|T_2|+ \Delta(G'))$ vertices in $(S_1 \cup X_1) \setminus \{v_i\}$ and 
	at most $\Delta(G')$ neighbors in $V(G') \setminus V_{S'_1}$.  Thus, we have 
	$\deg_G(v_i)\leq q \cdot (|T_2|+ \Delta(G')) + |T_2|-2 + \Delta(G') <  (q+1) 
	\cdot (|T_2|+ \Delta(G'))$.   Moreover, there are at least $q \cdot (|T_2|+ 
	\Delta(G'))$ agents of type $T_1$ adjacent to $v_i$ (all agents in $(S_1 \cup 
	X_1) \setminus \{v_i\}$). Therefore, it holds that agent $j \in T_1$ has 
	utility: 
	
	\begin{align*}
	u_j(\mathbf{v}^{i  \leftrightarrow j }) \geq \frac{q \cdot 
		(|T_2|+ \Delta(G'))}{ (q+1) \cdot (|T_2|+ 
		\Delta(G'))}=\frac{q}{q+1}> \frac{q-1}{q}\geq 
	u_j(\mathbf{v}).                                            
	\end{align*}
	Hence, swapping $i$ and $j$ is profitable and $\mathbf{v}$ cannot be a 
	swap-equilibrium. 
	
	Since we have exhausted all 
	possible cases, the lemma follows.
\end{proof}

Next, we prove that all vertices in $B$ have to be occupied by agents from 
$T_2$.
\begin{lemma} \label{swapEqEx:B}
	In any swap-equilibrium $\mathbf{v}$ in $\mathcal{I}$, all vertices in $B$ are 
	occupied by agents from~$T_2$.
\end{lemma}
\begin{proof}
	By \Cref{swapEqEx:A}, vertices from $A$ are only occupied by  
	agents from $T_1$ in $\mathbf{v}$. The remaining $y\coloneqq |T_1|-|A|= |T'_1|- 
	|V_{S'_1}|< 
	|V(G')\setminus V_{S'}|$ agents from $T_1$ and all agents from $T_2$ occupy 
	vertices in $G[V(G)\setminus A]$. Since it holds that $y<|V(G')\setminus 
	V_{S'}|$, 
	there  exists an agent $j \in T_2$ with $v_j \in  V(G')\setminus V_{S'}$. 
	Recall that by \Cref{le:stEqEx} we have assumed that in our input instance, 
	every vertex not occupied by 
	a stubborn agent  is adjacent to at least one stubborn agent of each type. 
	Thus, $v_j$ is adjacent to $v_i \in V_{S'_1}$ occupied by agent $i$. It holds 
	that 
	$v_i \in A$, hence we have that agent $i$ must be from  $T_1$. The agents $i$ 
	and $j$ have $u_{i}(\mathbf{v})<1$ and $u_{j}(\mathbf{v})<1$.
	
	Suppose there are $x$ agents of type $T_1$ in $B$ in $\mathbf{v}$, with 
	$|V(G')|>y\geq x>0$. We prove 
	that in this case $\mathbf{v}$ cannot be a swap-equilibrium. 
	Observe that the subgraph $G[B]$ consists of 
	$|V_{S'_2}|$ stars where the central vertices form a clique. Each star contains 
	$|V(G')|^2+1>x$ vertices. Therefore, at least one of the stars has to contain 
	agents from both types. That is, there exists an agent $i' \in T_t$ for some $t 
	\in \{1,2\}$ on a degree-one vertex in $B$ with $u_{i'}(\mathbf{v})=0$. If 
	$t=1$, then swapping agent $i'$ and agent $j$ is profitable, since 
	$u_{i'}(\mathbf{v})=0<u_{i'}(\mathbf{v}^{i' \leftrightarrow j})$ and 
	$u_{j}(\mathbf{v})<1=u_{j}(\mathbf{v}^{i' \leftrightarrow j})$. Otherwise, if 
	$t = 2$, then swapping agent $i'$ and agent $i$ is profitable, since 
	$u_{i'}(\mathbf{v})=0<u_{i'}(\mathbf{v}^{i' \leftrightarrow i})$ and 
	$u_{i}(\mathbf{v})<1=u_{i}(\mathbf{v}^{i' \leftrightarrow i})$. Therefore, 
	there exists a profitable swap and $\mathbf{v}$ cannot be a 
	swap-equilibrium.
\end{proof}

After establishing these two lemmas, we are now able to prove the backwards 
direction of the correctness of the reduction:
\begin{lemma}\label{swapEqEx:Back}
	If the constructed instance $\mathcal{I}$ of \sEq{} admits a 
	swap-equilibrium, then the given instance $\mathcal{I}'$ of 
	\stubsEq{} admits a swap-equilibrium.
\end{lemma}
\begin{proof}
	Assume there exists a swap-equilibrium $\mathbf{v}$ for the constructed 
	instance $\mathcal{I}'$ without stubborn agents. 
	We define an assignment $\mathbf{v'}$ and prove that it is a 
	swap-equilibrium  
	for the given Schelling game $\mathcal{I}'$ with stubborn agents on $G'$.  
	In 
	$\mathbf{v'}$, a vertex $v \in V(G') \setminus V_{S'}$ is occupied by a 
	strategic agent  of the same type as the agent on $v$ in $\mathbf{v}$. The 
	vertices in $V_{S'}$ have to be occupied by the respective stubborn agents. 
	Note 
	that by \Cref{swapEqEx:A,swapEqEx:B}, in the swap-equilibrium $\mathbf{v}$ 
	in $\mathcal{I}$, the vertices in  
	$V_{S'_1} 
	\subseteq A$ and  $V_{S'_2} \subseteq B$ have to be occupied by agents from 
	$T_1$ 
	and $T_2$, respectively.  Thus, in $\mathbf{v'}$, all vertices are occupied 
	by 
	agents of the same type as in $\mathbf{v}$. Note that in $\mathbf{v'}$ 
	exactly $|T'_1|$ agents from $T'_1$ and $|T'_2|$ agents from $T'_2$ are 
	assigned, as we have 
	$|T_1|=|T'_1|+|M_1|+|X_1|$ and $|T_2|=|T'_2|+|M_2|$ and as proven in 
	\Cref{swapEqEx:A,swapEqEx:B}, in  $\mathbf{v}$, all vertices from $M_1\cup 
	X_1$ are occupied by agents from $T_1$ and all vertices from $M_2$ are 
	occupied by agents from $T_2$.
	
	Now, we will prove that  $\mathbf{v'}$ is a swap-equilibrium on the given 
	$G'$. 
	Since stubborn agents never swap position, a profitable swap has to involve 
	two 
	strategic agents $i \in T'_1$ and $j \in T'_2$ with $v_i,v_j \in V(G') 
	\setminus 
	V_{S'}$. However, by construction of $G$, the neighborhoods of $v_i$ and 
	$v_j$ 
	are identical in $G$ and $G'$. Additionally, it holds that in 
	$\mathbf{v'}$, 
	all 
	vertices are occupied by agents of the same type as in $\mathbf{v}$. Since 
	$\mathbf{v}$ is a swap-equilibrium on $G$, swapping $i$ and $j$ cannot be 
	profitable. It follows that $\mathbf{v'}$ is a swap-equilibrium, which 
	completes 
	the proof.
\end{proof}

Note that the membership of \sEq in NP is trivial, as it is possible to verify 
that a given assignment is a swap equilibrium by iterating over all pairs of agents and  checking whether swapping the two is profitable.
Thus, from \Cref{swapEqEx:For} and \Cref{swapEqEx:Back}, \Cref{swapHardness:noStubborn} follows:
\seqhard*
\subsection{Jump-Equilibria} \label{sub:jq}
Inspired by the reduction for \sEq{} described above, we can 
prove that deciding the existence of a jump-equilibrium is NP-hard as well. 
While the 
constructions behind both reductions use the same underlying general ideas, the 
proof for 
\jEq{} is more involved. The main challenge here is that in every 
assignment some vertices remain unoccupied. For instance, we do not 
only need to prove that only agents from~$T_1$ are placed on vertices from $A$ 
(which is more challenging because we have to deal with possibly unoccupied 
vertices) but also that all vertices from~$A$ are occupied. This subsection is devoted to proving the following theorem: 

\begin{restatable}{theorem}{jumhard}
\label{np:jumpEqEx}
\jEq{} is NP-complete. 
\end{restatable}

To prove the theorem, we reduce from a restricted version  of \stubjEq{} as defined in 
the following lemma. To prove hardness for the restricted version, the proof of  \citeA{AGARWAL2021103576} for the general version of 
\stubjEq{} needs to be slightly adapted.

\begin{restatable}{lemma}{jstEqEx}\label{np:jump:lemma_ass}
	Let $\lambda = |V(G)|-n$ be the number of unoccupied vertices in an 
	instance of \stubjEq{}.
	We call an instance of \stubjEq{} regularized, if the following 
	five 
	properties hold. 
	\begin{enumerate}
		\item For every vertex $v \notin V_S$ not occupied by a stubborn agent, 
		there 
		exist two adjacent vertices $s_i, s_j \in V_S$ occupied by stubborn 
		agents $i 
		\in T_1$ and $j \in T_2$.
		\item Every vertex $v \in V_S$ has $\deg_G(v)< \lambda$.
		\item Every vertex $v \in V_{S}$ is adjacent to a vertex $v \notin 
		V_{S}$.
		\item It holds that $\lambda>0$.
		\item There are at least two stubborn agents of each type.
	\end{enumerate}
	\stubjEq{} remains NP-hard when restricted to regularized 
	instances.
\end{restatable}
\begin{proof}
	We prove this statement by giving a reduction that is heavily based on the 
	reduction by \citeA{AGARWAL2021103576}, which  is modified in 
	order to ensure that the constructed instance is always regularized. Most 
	importantly, we 
	modify 
	the original construction such that the first property holds. All other 
	properties already hold or are trivial to achieve. 
	
	We reduce from \textsc{Clique}. An instance of \textsc{Clique} consist of an 
	undirected graph $H=(X,Y)$ and an integer $s$. It is a yes-instance if and only 
	if $H$ contains a clique of size~$s$. Without loss of generality, we assume 
	that 
	$s \geq 6$. We construct an instance of \stubjEq{} as follows:
	
	There are two types $T_1$ and $T_2$. There are $s$ strategic agents all part of 
	$T_1$. All other agents are stubborn and defined along with the topology. 
	
	The topology $G=(V,E)$ consists of three components $G_1, 
	G_2,$ 
	and $G_3$, which are constructed as described below. 
	\begin{itemize}
		\item To define the graph $G_1=(V_1,E_1)$, let $W_v$ be a set of $s$ 
		vertices 
		for every $v \in X$. Out of the vertices in $W_v$, one vertex is occupied by a 
		stubborn agent from $T_1$ and the remaining $s-1$ vertices are occupied by 
		stubborn agents from $T_2$. We set $V_1 = X \cup \bigcup_{v \in X} W_v$ and 
		$E_1 
		= Y \cup \bigcup_{v \in X}  \{ \{v,w\} \mid w \in W_v\}$. That is, $G_1$ is an 
		extended copy of the given $H$, where every $v \in 
		X$ is adjacent to $s$ degree-one vertices in $W_v$, which are occupied by 
		stubborn agents  from both types. 
		\item The graph $G_2$ is a bipartite graph with parts $L$ and $R$. Let $L$ 
		be 
		a set of $s-2$ vertices. For every $v \in L$, the set $R$ contains $4s$ 
		vertices 
		only connected to~$v$. Out of these $4s$ vertices, $2s+1$ are occupied by 
		stubborn agents from $T_1$ and the remaining $2s-1$ vertices are occupied by 
		stubborn agents from $T_2$. 
		\item In $G_3$, only three vertices $x,y$ and $z$ are not occupied by 
		stubborn agents. The vertices $x$ and $y$ are connected by an edge. The 
		remaining vertices are occupied by stubborn agents and defined in the 
		following. 
		First, the vertex $x$ is connected to one degree-one vertex occupied by a 
		stubborn agent from $T_1$ and two degree-one vertices occupied by stubborn 
		agents from $T_2$. The vertex $y$ is connected to $41$ degree-one vertices 
		occupied by stubborn agents from $T_1$ and $80$ degree-one vertices occupied by 
		stubborn agents from $T_2$.   Finally, $z$ is connected to $5$ degree-one 
		vertices occupied by stubborn agents from $T_1$ and $7$ degree-one vertices 
		occupied by stubborn agents from $T_2$. 
	\end{itemize}
	Lastly, we pick an arbitrary vertex occupied by a stubborn agent from each of 
	the three components $G_1$, $G_2$, $G_3$ and connect the three vertices to form 
	a clique.
	
	It is easy to verify that the constructed instance is regularized.
	The correctness of the reduction follows analogous to the proof by  \citeA{AGARWAL2021103576}. The only difference is the exact utility of 
	agents on $G_1$, however, the same inequalities still hold.
\end{proof}

In the following, we start by describing the construction of the reduction to 
prove \Cref{np:jumpEqEx} before we prove its correctness.
The construction is very similar to the one for \Cref{swapHardness:noStubborn} from \Cref{sub:sq}. 
In particular, the sketch from \Cref{np:G} still applies here.
The only major difference between the two reductions is that the number of vertices in $X_1$ is different. 
As a consequence, also the number of agents from both types is different than before.
For the sake of completeness and to avoid possible confusions, we provide the full description of the construction here: 

\paragraph{Construction.}
We are given an instance $\mathcal{I}'$ of \stubjEq consisting of a 
connected topology 
$G'$, a 
set $N'=R' \cup S'$ of agents partitioned into types $T_1'$ and $T_2'$ and 
a set 
of vertices $V_{S'} =\{s_i \in V(G') \mid i \in S'\}$. The agents in $R'$ 
are strategic and the agents from $S'$ are stubborn agents, with  stubborn 
agent $i \in S'$ occupying $s_i \in V_{S'}$ in any assignment.  
We assume that the given instance is 
regularized 
and fulfills the properties from \Cref{np:jump:lemma_ass}.
In the following, we denote the sets of vertices occupied by stubborn 
agents from $T_1'$ and 
$T_2'$ as $V_{S'_1}$ and $V_{S'_2}$, respectively. From this, we construct 
an instance $\mathcal{I}$
of 
\jEq{} consisting of a topology $G=(V,E)$ and types $T_1$ and $T_2$ as 
follows. 

The graph $G$ (sketched in \Cref{np:G}) is an extended copy of 
the given graph~$G'$ and contains all vertices and edges from $G'$. We add 
three 
sets of vertices $M_1, X_1$, and $M_2$ as specified below. For every vertex $v 
\in V_{S'_2}$, we insert $|V(G')|^2$ degree-one 
vertices only adjacent to $v$ and add them to $M_2$. We connect the
vertices in $ V_{S'_2}$ to form a clique.
	
	Next, we define $q$, which, as argued later, is an upper bound for the 
	degree of a vertex in $(V(G')\setminus V_{S'_1}) \cup M_2$:
	\begin{align*}
	q \coloneqq \Delta(G')+|V(G')|^2 + |V_{S'_2}|.
	\end{align*}
	Now, we define $s$ and $z$, which are important in the 
	proof of 
	\Cref{np:jump:claim1.1}:
	\begin{align*}
	s &\coloneqq q \cdot (|T'_2|+|M_2| + \Delta(G) ) + 1,\\
	z &\coloneqq s+|V(G')|+|M_2|.
	\end{align*}
	Let  $X_1$ be a set of $z -|V_{S'_1}|$ vertices (this is different than in the construction from \Cref{swapHardness:noStubborn}). We add the vertices $X_1$ to $G$ and connect the vertices in $V_{S'_1}\cup X_1$ to form a clique.   
	Note that it holds that $|X_1 \cup 
	V_{S'_1}|= z$, which we use to introduce an asymmetry between the two types.
	
	Finally, we define the number $p$ of added degree-one neighbors for 
	vertices 
	in $V_{S'_1} \cup X_1$ (again, the choice of $p$ is used in 
	\Cref{np:jump:claim1.1}):
	\begin{align*}
	p \coloneqq |T_2'|+|M_2| -2 > |V(G')|^2.
	\end{align*}
	For every vertex $v \in V_{S'_1} \cup X_1$, we add $p$ degree-one 
	vertices 
	only adjacent to $v$ in $G$ and add them to $M_1$. 
	
	The set of agents $N =T_1 \dot \cup T_2$ is defined as follows. 
	We have $|T_1|=|T'_1|+|M_1|+|X_1|$ agents in $T_1$ and 
	$|T_2|=|T'_2|+|M_2|$ agents in $T_2$. By the construction of $X_1$, we 
	have that:
	\begin{align*}
	s=q \cdot (|T'_2|+|M_2| + \Delta(G') )  + 1 = q \cdot (|T_2|+ 
	\Delta(G'))+1.
	\end{align*}
	It also holds that $p = |T_2'|+|M_2| -2 = |T_2| -2$. Let $\lambda':= 
	|V(G')|-|N'|$ be the number of unoccupied vertices in the given instance 
	and $\lambda := |V(G)|-|N|$ the number of unoccupied vertices 
	in the constructed instance.  Note that 
	there are equally many unoccupied vertices in the constructed and the 
	given instance, since it holds that 
	$|V(G)|-|V(G')|=|M_1|+|X_1|+|M_2|=|N|-|N'|$. That is, $\lambda= 
	\lambda' < |V(G')|$. 
	
	\paragraph{Proof of Correctness.}
	Next, we address the correctness of the reduction. We approach the 
proof 
in four steps. First, we prove in \Cref{le:jumpfd} the forward direction of the correctness. 
Next, we make some basic observations about the 
constructed graph
$G$. Then, we prove \Cref{np:jump:claim0,np:jump:claim1,np:jump:claim2}, 
which 
state useful properties of all jump-equilibria in the constructed game. 
Finally, 
using these lemmas, we prove that the constructed game admits a 
jump-equilibrium only if the original game admits a jump-equilibrium.  

We define the sets $A,B \subseteq V$ as $A \coloneqq M_1 \cup V_{S'_1} \cup 
X_1$ and $B \coloneqq M_2 \cup V_{S'_2}$. Observe that the subgraph $G[A]$ 
can 
be partitioned into $z$ stars which each have $|T_2|-1$ vertices as 
follows. The vertices in $X_1 \cup V_{S'_1}$ are the central vertices and 
form a clique. Each central vertex is adjacent to 
$|T_2|-2$ degree-one vertices in $M_1$. Recall that $z > |V(G')| \geq  
|V_{S'}|> 3$, hence we have at least $3$ stars in $G[A]$. Note that 
$V(G)\setminus A = (V(G') \setminus V_{S'}) \cup B = (V(G') \setminus 
V_{S'}) 
\cup (M_2 \cup V_{S'_2})$. The subgraph $G[V(G)\setminus A]$ is connected, 
since 
every vertex in $V(G') \setminus V_{S'}$ is connected to a vertex in 
$V_{S'_2}$ 
(by our assumption that the given instance is regularized; see 
\Cref{np:jump:lemma_ass}) and the vertices 
in $V_{S'_2}$ form a clique (by the construction of $G$). Additionally, all 
vertices in $M_2$ are adjacent to exactly one vertex in $V_{S'_2}$. 
 
We start with showing the (easier) forward direction of the correctness
of the reduction: 

\begin{lemma} \label{le:jumpfd}
	If the given instance $\mathcal{I}'$ of \stubjEq{} admits a 
	jump-equilibrium, then the 
	constructed 
	instance $\mathcal{I}$ of \jEq{} admits a swap-equilibrium.
\end{lemma}
\begin{proof}
	Assume that there exists a jump-equilibrium 
	$\mathbf{v}'$ for the given instance $\mathcal{I}'$ with stubborn agents. 
	Note that 
	in $\mathbf{v}'$ 
	the vertices in $V_{S'_1}$ and $V_{S'_2}$ are occupied by stubborn agents 
	from 
	$T_1'$ and $T_2'$, respectively. We define an assignment $\mathbf{v}$ for 
	the constructed
	Schelling game without stubborn agents as follows and afterwards prove that 
	it is a 
	jump-equilibrium. In $\mathbf{v}$, the occupied vertices in $V(G) \cap 
	V(G')$ are occupied 
	by 
	agents of the same type as the agents in $\mathbf{v}'$. If a vertex is 
	unoccupied in $\mathbf{v}'$, then it is also unoccupied in $\mathbf{v}$.  
	The vertices in $X_1$ are occupied by agents from $T_1$. For $t \in 
	\{1,2\}$, the added degree-one vertices in $M_t$ are occupied by agents 
	from $T_t$. Hence, the vertices in $A=M_1 \cup V_{S'_1} \cup X_1$ are all 
	occupied by agents from $T_1$ and the vertices in $B=M_2 \cup V_{S'_2}$ are 
	all 
	occupied by agents from $T_2$ (note that, in $\mathbf{v}$, we assigned 
	exactly $|T_1|$ agents from $T_1$ and $|T_2|$ agents from $T_2$, as  
	$|T_1|=|T'_1|+|M_1|+|X_1|$ and 
	$|T_2|=|T'_2|+|M_2|$).
	
	Next, we prove that $\mathbf{v}$ is  a jump-equilibrium on $G$ by showing 
	that no profitable jump exists. We first observe that the utility of an 
	agent $i$ on a vertex $w \in M_1 \cup X_1$ is $u_i(\mathbf{v})=1$, since 
	$N_G(w)\subseteq A$ and all agents in $A$ are friends of $i$.  The same 
	holds analogously for any agent on a vertex in $M_2$. Therefore, the agents 
	on vertices in $M_1 \cup M_2 \cup X_1$ do not want to jump to an unoccupied 
	vertex. Furthermore, note that the only unoccupied vertices in $\mathbf{v}$ 
	are in $V(G') 
	\setminus V_{S'}$ and that the neighborhood of all vertices in $V(G') 
	\setminus V_{S'}$ is identical in $G$ and $G'$. 
	Therefore, no agent on a vertex in $V(G') \setminus V_{S'}$ can have a 
	profitable jump in $\mathbf{v}$, since this jump would then also be 
	profitable in  
	$\mathbf{v}'$. Thus, in $\mathbf{v}$, a profitable 
	jump can only exist for an agent $i$ on a vertex  $v_i \in V_{S'}$ to an 
	unoccupied vertex $v \in V(G') \setminus V_{S'}$.
	
	Let $Y=V_{S'_1} \cup X_1$ if $v_i \in V_{S'_1}$ and $Y=V_{S'_2}$ if $v_i 
	\in V_{S'_2}$. 
	Denote the 
	number of degree-one neighbors of $v_i$ that were added to $G'$ in the 
	construction of $G$ by $x$. By 
	construction of $G$, it holds that $x>\Delta(G)^2$. The agent $i$ is, among 
	others,
	adjacent to the vertices in $Y \setminus \{v_i\}$ and the $x>\Delta(G')^2$ 
	degree-one neighbors, which, by construction of $\mathbf{v}$, are all 
	occupied by friends. It holds that 
	$\deg_G(v_i) \leq x + |Y \setminus \{v_i\}| + \Delta(G')$. Thus, the 
	utility of agent $i$ on $v_i$ is at least:
	$$u_i^{G}(\mathbf{v})\geq \frac{x + |Y \setminus \{v_i\}|}{x + |Y \setminus 
		\{v_i\}| + \Delta(G')} > \frac{x}{x+\Delta(G')}> 
	\frac{\Delta(G')^2}{\Delta(G')^2+\Delta(G')}=\frac{\Delta(G')}{\Delta(G')+1}.$$
	Now consider an unoccupied vertex $v \in V(G') \setminus V_{S'}$. Recall 
	that as we assume that the given instance is regularized (see 
	\Cref{np:jump:lemma_ass}), for every vertex $v \notin V_{S'}$ not occupied 
	by a stubborn agent, it holds that
	there 
	exist two adjacent vertices $s_i, s_j \in V_{S'}$ occupied by stubborn 
	agents $i 
	\in T'_1$ and $j \in T'_2$.  Since $v \in 
	V(G') \setminus V_{S'}$, the vertex $v$ is thus adjacent to agents from 
	both $T_1$ 
	and $T_2$ in $A$ and $B$, respectively. By construction of $G$, the 
	neighborhood of $v$ is identical in $G$ and $G'$. We thus have that 
	$\deg_{G}(v)=\deg_{G'}(v)\leq \Delta(G')$. Since at least one agent 
	adjacent to $v$ is not a friend of $i$, the jump of $i$ to $v$ cannot be 
	profitable in $\mathbf{v}$:
	\begin{align*}
	u_i^{G}(\mathbf{v}^{i \rightarrow v}) \leq 
	\frac{\Delta(G')-1}{\Delta(G')}<\frac{\Delta(G')}{\Delta(G')+1}\leq  
	u_i^{G}(\mathbf{v}).
	\end{align*}
	To sum up, no profitable jump exists and $\mathbf{v}$ is a 
	jump-equilibrium  for the constructed Schelling game.
\end{proof}
 
 To show the correction of the backwards direction, we start by showing the 
following lemma, 
which states that no agent on a 
degree-one 
vertex in $A$ or $B$ can have no adjacent friends and all central vertices 
have 
to be occupied. This property of jump-equilibria is then later used in the 
proof of \Cref{np:jump:claim1,np:jump:claim2}, where we prove that all 
vertices in $A$ are occupied by agents from $T_1$ and all vertices from $B$ 
are occupied by agents from $B$.
\begin{lemma}\label{np:jump:claim0}
	In a jump-equilibrium $\mathbf{v}$ in $\mathcal{I}$, all agents on a 
	degree-one vertex $v \in M_1 \cup M_2$ are adjacent to a friend and all 
	central vertices $w \in V_{S'_1} \cup X_1 \cup V_{S'_2}$ are occupied.
\end{lemma}
\begin{proof}
	Suppose for the sake of a contradiction that there exists an agent $i 
	\in T_t$ for some $t \in \{1,2\}$ on a 
	degree-one vertex $v$ with no adjacent friend in $\mathbf{v}$. Let $S$ be 
	the set of 
	vertices of 
	the star which contains $v$. Recall that every star contains at least 
	$|V(G')|^2+1\geq |V(G')|+3\geq \lambda +3$ vertices. Thus, at least three 
	vertices in $S$ have to be occupied. Since $i$ has no adjacent friends, 
	the 
	central vertex $w\in S$ is either unoccupied or occupied by an agent of 
	the 
	other 
	type. We distinguish these two cases and prove 
	that 
	there exists a profitable jump in both cases.
	
	\paragraph{Case 1:}First, assume that $w$ is unoccupied. As mentioned 
	above, at least three vertices in $S$ have to be occupied by agents. 
	Since $w$ is unoccupied, all occupied vertices are degree-one vertices. 
	By the pigeonhole principle, there exist two agents of the same type on 
	degree-one vertices in $S$. Both agents have no adjacent friends in 
	$\mathbf{v}$ and can increase their utility by jumping to $w$. 
	This concludes the first case and furthermore proves that no central 
	vertex can be unoccupied.
	\paragraph{Case 2:} Second, we consider the case where $w$ is occupied 
	by an agent $j \in T_{t'}$ with $t' \neq t$. Note that both $|T_1|>|S|$ 
	and $|T_2|>|S|$, thus there are agents from both types in 
	$G[V(G)\setminus S]$. Next, 
	we argue that $G[V(G)\setminus S]$ is connected. Let $Y=A$ if $S 
	\subseteq A$ and $Y=B$ otherwise. It is easy to see that 
	$G[V(G)\setminus Y]$ is 
	connected (as argued for $G[V(G)\setminus A]$ above and analogous for 
	$G[V(G)\setminus B]$).  The connected subgraph $G[Y \setminus S]$ 
	consists of 
	the remaining stars in $Y \setminus S$, where the central vertices form 
	a clique. Note that there is at least one vertex $x \in V_{S'}$ in $Y 
	\setminus S$, since by \Cref{np:jump:lemma_ass}, we have assumed that 
	there 
	are at least two stubborn agents of each type. 
	Furthermore, by \Cref{np:jump:lemma_ass}), we have assumed that $x \in 
	V_{S'}$ is 
	adjacent to a vertex $y \in V(G') \setminus V_{S'}$ in $G[V(G)\setminus 
	Y]$. Thus, as $G[V(G)\setminus 
	Y]$ and $G[Y \setminus S]$  are each connected and connected by an 
	edge, $G[V(G)\setminus S]$ is connected.
	
	First, suppose that there is no unoccupied vertex in $V(G)\setminus 
	S$, which implies that there need to be unoccupied vertices in $S$. 
	Then, there 
	are two adjacent agents $i' \in T_t$ and $j' \in T_{t'}$  in 
	$G[V(G)\setminus S]$. It 
	holds that $u_{j'}(\mathbf{v})<1$. Agent $j'$ can increase its utility 
	to $1$ by jumping to any unoccupied vertex in $S$.
	
	Therefore, there has to exist at least one unoccupied vertex in 
	$G[V(G)\setminus S]$. 
	If one of the unoccupied vertices in $G[V(G)\setminus S]$ is adjacent 
	to an agent in 
	$T_t$, then agent $i \in T_t$ can increase its utility by jumping to 
	this vertex. Thus, the unoccupied vertices in $G[V(G)\setminus S]$ can 
	only be 
	adjacent to agents in $T_{t'}$ and unoccupied vertices. Since 
	$G[V(G)\setminus S]$ is 
	connected, there exists 
	an unoccupied vertex $w'$ that is adjacent to at least one agent in 
	$T_{t'}$ on a vertex in $G[V(G)\setminus S]$. Note that agent $j$ on 
	the central 
	vertex in $S$ has $u_j(\mathbf{v})<1$, since $j \in T_{t'}$ is adjacent 
	to $i \in T_t$. Hence, agent $j$ can increase its utility by jumping to 
	$w'$, where $j$ is only adjacent to friends.
\end{proof}

Next, we prove that all vertices in $A$ are occupied by agents from $T_1$.
\begin{lemma}\label{np:jump:claim1}
	In every jump-equilibrium $\mathbf{v}$ in $\mathcal{I}$, there is no 
	unoccupied vertex in $A$ and every vertex in $A$ is 
	occupied by an agent 
	from $T_1$. 
\end{lemma}
\begin{proof}
	We prove this lemma by splitting it into Claims
	\ref{np:jump:claim1.1} and \ref{np:jump:claim1.2}. First, we prove that 
	there 
	are no 
	agents from $T_2$ in $A$. Then, we prove that all vertices in $A$ are 
	occupied. 
	\begin{claimS}\label{np:jump:claim1.1}
		In every jump-equilibrium $\mathbf{v}$ in $\mathcal{I}$, no vertex 
		from  $A$ is occupied by an agent from~$T_2$. 
	\end{claimS}
	\begin{proof}[Proof of Claim]\renewcommand{\qedsymbol}{$\blacklozenge$}
		Recall that by the construction, we have that 
		$|T_1|\geq |A|= 
		|X_1 \cup V_{S'_1}|\cdot (|T_2|-1) =z \cdot (|T_2|-1)= 
		(s+|V(G')|+|M_2|) \cdot 
		(|T_2|-1)$. Thus, even if all vertices in $V(G) \setminus A=(V(G') 
		\setminus 
		V_{S'_1}) \cup M_2$ are occupied by agents from $T_1$, there are at 
		least $s \cdot 
		(|T_2|-1)$~agents from $T_1$ in $A$. Since each star in $A$ 
		contains $|T_2|-1$ 
		vertices, there have to be agents from $T_1$ on at least $s$ stars. 
		By 
		\Cref{np:jump:claim0}, it holds that at least $s$ central vertices 
		are occupied 
		by agents from~$T_1$.   
		
		Now suppose there are $x>0$ agents from $T_2$ in $A$. We 
		distinguish the 
		following three cases based on the value of $x$ (see 
		\Cref{fig:claim2_1}) and prove in each case that  $\mathbf{v}$ 
		cannot be a jump-equilibrium.
		
		\begin{figure}[t]
			\centering
			\begin{subfigure}[t]{0.3\textwidth}
				\tikzstyle{alter}=[circle, minimum size=12.5pt, draw, inner 
				sep=1pt, semithick] 
				\tikzstyle{majarr}=[draw=black, thick]
				\centering
				\begin{tikzpicture}[auto]
				
				\draw[rounded corners] (0, 0) rectangle (\textwidth*0.4, 
				\textwidth); 
				\node[] at (\textwidth*0.2, \textwidth+2ex) { $A$};
				
				\node[alter, fill=c_T2] at (\textwidth*0.3,\textwidth*0.7) 
				(a) 
				{};
				
				\node[alter] at (\textwidth*0.07,\textwidth*0.6) (a_1) {};
				\node[alter, fill=c_T2] at 
				(\textwidth*0.06,\textwidth*0.775) 
				(a_2) {};
				\node[alter, fill=c_T2] at (\textwidth*0.15,\textwidth*0.9) 
				(a_3) {};
				
				\draw[majarr] (a) edge (a_1);
				\draw[majarr] (a) edge (a_2);
				\draw[majarr] (a) edge (a_3);
				\draw[majarr] (a) edge  ($(a)+(0ex, 2ex)$);
				\draw[majarr] (a) edge  ($(a)+(1ex, 2ex)$);
				\draw[majarr] (a) edge  ($(a)+(2ex, 2ex)$);
				
				\node[alter, fill=c_T1] at 
				(\textwidth*0.075,\textwidth*0.3) 
				(x) {};
				\node[alter, fill=c_T1] at (\textwidth*0.2,\textwidth*0.2) 
				(y) 
				{};
				\node[alter, fill=c_T1] at 
				(\textwidth*0.325,\textwidth*0.3) 
				(z) {};
				
				\draw[thick,decorate,decoration={brace,mirror, 
					amplitude=5pt}]
				($(x)+(-7pt,-5ex)$) --  ($(z)+(7pt,-5ex)$) node[midway, 
				below=2pt] {\tiny s $T_1$-neighb.};
				
				\draw[majarr] (a) edge (x);
				\draw[majarr] (a) edge (y);
				\draw[majarr] (a) edge (z);
				
				\draw[majarr] (x) edge (y);
				\draw[majarr] (x) edge (z);
				\draw[majarr] (y) edge (z);
				\foreach \x in {x,y,z} {
					\draw[majarr] (\x) edge  ($(\x)-(1ex, 2ex)$);
					\draw[majarr] (\x) edge  ($(\x)-(0, 2ex)$);
					\draw[majarr] (\x) edge  ($(\x)-(-1ex, 2ex)$);
				}
				
				\tikzset{shift={(0.5\textwidth,0)}}
				
				\draw[rounded corners] (0, 0) rectangle (\textwidth*0.45, 
				\textwidth); 
				\node[] at (\textwidth*0.2, \textwidth+2ex) { 
					$G[V(G)\setminus 
					A]$};
				
				\node[alter, label=below:$w$] at 
				(\textwidth*0.15,\textwidth*0.45) (b) {};
				\node[alter, fill=c_T2] at (\textwidth*0.3,\textwidth*0.65) 
				(c) 
				{};
				\draw[majarr] (b) edge (c);
				
				\draw[->, thick , bend left=10] (a) to (b);
				
				\end{tikzpicture}
				\caption{ $x<|T_2|-1$. }
			\end{subfigure}
			\hspace{0.5cm}   
			\begin{subfigure}[t]{0.2\textwidth}
				\tikzstyle{alter}=[circle, minimum size=12.5pt, draw, inner 
				sep=1pt, semithick] 
				\tikzstyle{majarr}=[draw=black, thick]
				\centering
				\begin{tikzpicture}[auto]
				\draw[rounded corners] (0, 0) rectangle (\textwidth*0.9, 
				\textwidth*1.5); 
				\node[] at (\textwidth*0.45, 1.5*\textwidth+2ex) { $A$};
				
				\node[alter, fill=c_T2] at (\textwidth*0.3,\textwidth*1.1) 
				(a) 
				{};
				\node[alter, fill=c_T2] at (\textwidth*0.3,\textwidth*0.4) 
				(b) 
				{};
				\node[alter, fill=c_T1] at (\textwidth*0.7,\textwidth*0.75) 
				(c) 
				{};
				
				\draw[majarr] (a) edge (b);
				\draw[majarr] (a) edge (c);
				\draw[majarr] (b) edge (c);
				
				\draw[majarr] (c) edge ($(c)+(2ex, -2ex)$);
				\draw[majarr] (c) edge ($(c)+(3ex, 0ex)$);
				\draw[majarr] (c) edge ($(c)+(2ex, 2ex)$);
				
				\draw[majarr] (a) edge ($(a)+(1ex, 2ex)$);
				\draw[majarr] (a) edge ($(a)+(2ex, 2ex)$);
				\draw[majarr] (a) edge ($(a)+(3ex, 2ex)$);
				
				\draw[majarr] (b) edge ($(b)+(1ex, -2ex)$);
				\draw[majarr] (b) edge ($(b)+(2ex, -2ex)$);
				\draw[majarr] (b) edge ($(b)+(3ex, -2ex)$);
				
				\node[alter, fill=c_T2] at (\textwidth*0.1,\textwidth*0.9) 
				(a_1) {};
				\node[alter, fill=c_T2] at (\textwidth*0.1,\textwidth*1.1) 
				(a_2) {};
				\node[alter, fill=c_T2] at 
				(\textwidth*0.225,\textwidth*1.35) 
				(a_3) {};
				\draw[majarr] (a) edge (a_1);
				\draw[majarr] (a) edge (a_2);
				\draw[majarr] (a) edge (a_3);
				
				\node[alter] at (\textwidth*0.1,\textwidth*0.55) (b_1) {};
				\node[alter, fill=c_T2] at (\textwidth*0.1,\textwidth*0.3) 
				(b_2) {};
				\node[alter, fill=c_T2] at 
				(\textwidth*0.225,\textwidth*0.1) 
				(b_3) {};
				\draw[majarr] (b) edge (b_1);
				\draw[majarr] (b) edge (b_2);
				\draw[majarr] (b) edge (b_3);
				
				\draw[->, thick , bend left=10] (a) to (b_1);
				
				\end{tikzpicture}
				\caption{ $x=|T_2|$. }
			\end{subfigure}
			\hspace{0.5cm}
			\begin{subfigure}[t]{0.3\textwidth}
				\tikzstyle{alter}=[circle, minimum size=12.5pt, draw, inner 
				sep=1pt, semithick] 
				\tikzstyle{majarr}=[draw=black, thick]
				\centering
				\begin{tikzpicture}[auto]
				\draw[rounded corners] (0, 0) rectangle (\textwidth*0.4, 
				\textwidth); 
				\node[] at (\textwidth*0.2, \textwidth+2ex) { $A$};
				
				\node[alter, fill=c_T2] at (\textwidth*0.1,\textwidth*0.7) 
				(a) 
				{};
				\node[alter, fill=c_T1] at (\textwidth*0.1,\textwidth*0.3) 
				(b) 
				{};
				\node[alter, fill=c_T1] at (\textwidth*0.3,\textwidth*0.5) 
				(c) 
				{};
				
				\draw[majarr] (a) edge (b);
				\draw[majarr] (a) edge (c);
				\draw[majarr] (b) edge (c);
				
				\draw[majarr] (c) edge ($(c)+(1ex, -2ex)$);
				\draw[majarr] (c) edge ($(c)+(2ex, 0ex)$);
				\draw[majarr] (c) edge ($(c)+(1ex, 2ex)$);
				
				\draw[majarr] (a) edge ($(a)+(-1ex, -2ex)$);
				\draw[majarr] (a) edge ($(a)+(-2ex, -2ex)$);
				\draw[majarr] (a) edge ($(a)+(-2.5ex, -1ex)$);
				
				\draw[majarr] (b) edge ($(b)+(1ex, -2ex)$);
				\draw[majarr] (b) edge ($(b)+(2ex, -2ex)$);
				\draw[majarr] (b) edge ($(b)+(3ex, -2ex)$);
				
				\node[alter, fill=c_T2] at (\textwidth*0.07,\textwidth*0.9) 
				(a_1) {};
				\node[alter, fill=c_T2] at (\textwidth*0.2,\textwidth*0.9) 
				(a_2) {};
				\node[alter, fill=c_T2] at (\textwidth*0.33,\textwidth*0.9) 
				(a_3) {};
				\draw[majarr] (a) edge (a_1);
				\draw[majarr] (a) edge (a_2);
				\draw[majarr] (a) edge (a_3);
				
				\tikzset{shift={(0.45\textwidth,0)}}
				
				\draw[rounded corners] (0, 0) rectangle (\textwidth*0.3, 
				\textwidth); 
				\node[] at (\textwidth*0.15, \textwidth+2ex) { $G' 
					\setminus 
					V_{S'}$};
				
				\node[alter, label=below:$w$] at 
				(\textwidth*0.15,\textwidth*0.5) (x) {};
				
				\tikzset{shift={(0.35\textwidth,0)}}
				
				\draw[rounded corners] (0, 0) rectangle (\textwidth*0.4, 
				\textwidth); 
				\node[] at (\textwidth*0.2, \textwidth+2ex) { $B$};
				
				\node[alter, fill=c_T1] at (\textwidth*0.1,\textwidth*0.5) 
				(u) 
				{};
				\node[alter, fill=c_T1] at 
				(\textwidth*0.275,\textwidth*0.75) 
				(v) {};  
				\node[alter, fill=c_T1] at 
				(\textwidth*0.275,\textwidth*0.25) 
				(w) {};
				
				\draw[majarr] (x) edge (u);
				\draw[majarr] (u) edge (v);
				\draw[majarr] (u) edge (w);
				\draw[majarr] (v) edge (w);
				
				\foreach \x in {u,v,w} {
					
					\draw[majarr] (\x) edge  ($(\x)+(2ex, 1ex)$);
					\draw[majarr] (\x) edge  ($(\x)+(2ex, 0ex)$);
					\draw[majarr] (\x) edge  ($(\x)+(2ex, -1ex)$);
					
				}
				
				\draw[->, thick , bend left=25, thick] (c) to (x);
				
				\end{tikzpicture}
				\caption{$x=|T_2|-1$. }
			\end{subfigure}
			
			\caption{\Cref{np:jump:claim1.1}: If there are $x>0$ agents from 
				$T_2$ in 
				$A$, then a profitable jump exists. }
			\label{fig:claim2_1}
		\end{figure}
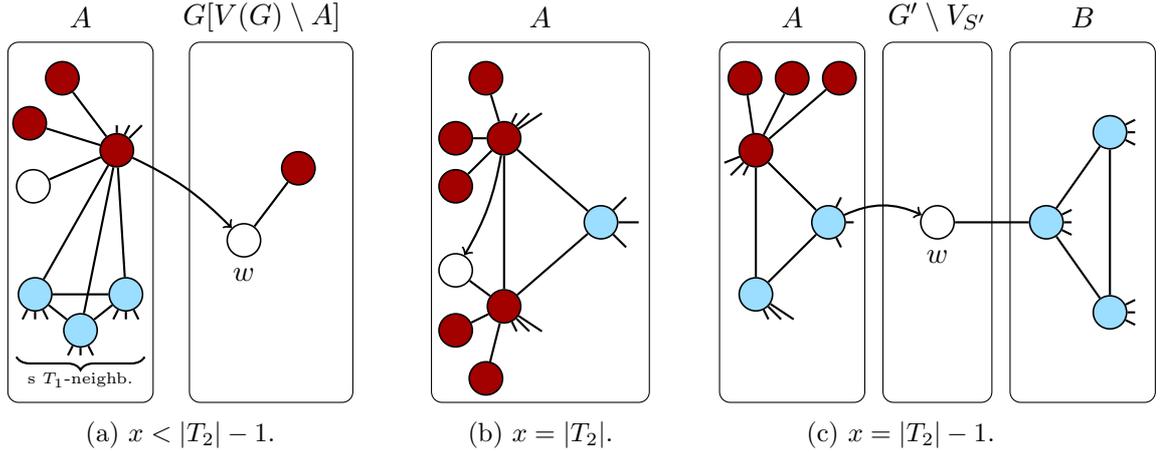
		
		\paragraph{Case 1:} First, assume there are $x<|T_2|-1$ agents from 
		$T_2$ in 
		$A$. Recall \Cref{np:jump:claim0}, which states that no agent on a 
		degree-one 
		vertex in $A$ has no adjacent friend in a jump-equilibrium. 
		Thus, 
		there exists an agent $i \in T_2$ on a central vertex  in $A$. Let 
		$S$ be the 
		set of vertices of the star which contains $v_i$. Since it holds 
		that 
		$x<|T_2|-1=|S|$, not all vertices in $S$ can be occupied by agents 
		from $T_2$. 
		By \Cref{np:jump:claim0}, these vertices have to be unoccupied. 
		Summarizing, 
		there exists an unoccupied degree-one vertex $w\in S$ adjacent to 
		the 
		central 
		vertex occupied by $i\in T_2$ (note that an agent from 
		$T_2\setminus \{i\}$ would get utility $1$ from jumping to $w$. 
		However, as there exist unoccupied vertices, we cannot be sure that 
		there exists an agent with utility smaller than $1$ on a vertex from 
		$V(G)\setminus N_G(w)$.
		
		Next, we upper-bound the utility of agent $i$. Note that $i \in 
		T_2$ is adjacent to all other central vertices in $A$, of which at 
		least $s=q \cdot (|T_2| + \Delta(G') ) + 1$ are occupied by agents 
		from~$T_1$. It therefore holds that:
		\begin{align*} 
		u_i(\mathbf{v})\leq \frac{|T_2|}{q \cdot (|T_2| + \Delta(G') ) + 1 
			+|T_2|} < \frac{1}{q}.
		\end{align*}
		Since $x<|T_2|$ and $|T_1| \geq |A|$, there are agents from both 
		types in $V(G)\setminus A$. Furthermore, recall that 
		$G[V(G)\setminus A]$ is connected. Now 
		consider a path between two arbitrary agents from $T_1$ and $T_2$ 
		in $G[V(G)\setminus A]$. If there are two adjacent agents $i' \in 
		T_1$ and $j' \in 
		T_2$ on this path, then it holds that $u_{j'}(\mathbf{v})<1$ and 
		jumping to $w$ is profitable for $j'$. Thus, no such two agents can 
		exist. However then, there exists a unoccupied vertex $w'$ on the 
		path that is adjacent to an agent in $T_2$.
		Note that $\deg_G(w')\leq  \Delta(G')+|V(G')|^2 + |V_{S'_2}|=q$.
		Jumping to $w'$ is profitable for agent $i$:
		\begin{align*}
		u_i(\mathbf{v}^{i \rightarrow w'}) \geq \frac{1}{q}>u_i(\mathbf{v}).
		\end{align*}
		
		\paragraph{Case 2:} Next, we consider the case where $x=|T_2|$. 
		Since $2 \cdot 
		(|T_2|-1)>x>|T_2|-1$, the agents from $T_2$ occupy vertices on at 
		least two 
		stars in $A$, but cannot occupy all vertices of these stars. With 
		\Cref{np:jump:claim0}, there exists an unoccupied degree-one vertex 
		$w$ 
		adjacent to a central vertex occupied by an agent from $T_2$. Now 
		consider an 
		agent $i \in T_2$ on another central vertex in $A$ not adjacent to 
		$w$. We have 
		that $u_i(\mathbf{v})<1$, since $i$ is adjacent to agents from 
		$T_1$ on other 
		central vertices. Then, agent $i$ can increase its utility to $1$ 
		by jumping to~$w$.
		
		\paragraph{Case 3:} Finally, we address the case where $x=|T_2|-1$. 
		Again, by 
		\Cref{np:jump:claim0}, at least one of the central vertices in $A$ 
		has to be 
		occupied by an agent $i \in T_2$. Furthermore, if there are agents 
		from $T_2$ 
		on two or more stars, then there exists an unoccupied degree-one 
		vertex 
		adjacent to a central vertex occupied by an agent from $T_2$ and an 
		agent from 
		$T_2$ on another central vertex with utility less than $1$. 
		Analogous to Case 
		2, such an assignment cannot be a jump-equilibrium. Therefore, the 
		$|T_2|-1$ 
		agents from $T_2$ occupy all vertices of one star in $A$. We denote 
		the set of 
		vertices of this star by $S$. 
		
		All central vertices in $A$ and $B$ have to be occupied by 
		\Cref{np:jump:claim0}. Note that there is only one agent from $T_2$ 
		outside of 
		$A$. This agent cannot occupy a degree-one vertex in $B$, since 
		it would have 
		no adjacent friends on such a vertex. If it occupies a central 
		vertex in $B$, then 
		there are agents from $T_1$ with no adjacent friends on degree-one 
		vertices in 
		this star (recall that at least three vertices of each star have to 
		be 
		occupied). Both possibilities contradict \Cref{np:jump:claim0}. 
		Hence, this 
		agent occupies a vertex in $V(G')\setminus V_{S'}$ and all central 
		vertices in 
		$B$ are 
		occupied by agents from $T_1$.
		
		Next, we argue that both $A$ and $B$ are fully occupied in this 
		case. First, 
		suppose there 
		exits an unoccupied vertex $w$ in $A$. By \Cref{np:jump:claim0}, 
		$w$ has to be 
		a degree-one vertex. Furthermore, since all $|T_2|-1$ agents from 
		$T_2$ in $A$ 
		fully occupy one star, $w$ is adjacent to a central vertex occupied 
		by an agent 
		from $T_1$. Then, an agent from $T_1$ on a central vertex in $A$ 
		not adjacent 
		to $w$ can increase its utility to $1$ by jumping to $w$. If there 
		is an 
		unoccupied degree-one vertex $w$ in $B$, then this vertex is 
		adjacent to a 
		central vertex occupied by an agent from $T_1$ (recall that all 
		central 
		vertices in $B$ are occupied by agents from $T_1$). Then again, any 
		agent from 
		$T_1$ on a central vertex in $A$ has a profitable jump to $w$. 
		Thus, both $A$ 
		and $B$ are fully occupied and all $\lambda$ unoccupied vertices 
		are in 
		$V(G')\setminus V_{S'}$.
		
		Now consider the central vertex $v_i\in S$ of the star in $A$ 
		occupied 
		by the agents from~$T_2$. Recall that as our given instance is 
		regularized (see \Cref{np:jump:lemma_ass}) we can assume that 
		$\deg_{G'}(v)<\lambda$ for all $v \in V_{S'}$ in our input 
		instance. In our construction, we only add edges within $A$ to a 
		vertex in $A$. Hence, as all $\lambda$ unoccupied vertices need to 
		be from $V(G')\setminus V_{S'}$,  there exists an 
		unoccupied vertex $w$ in 
		$G'-V_{S'}$ which is not adjacent to $v_i$. If $w$ is adjacent to 
		an agent in $T_2$, jumping to $w$ is profitable for agent $i$, as, 
		with 
		an argument analogous to Case 1, we get that: 
		\begin{align*}
		u_i(\mathbf{v}^{i \rightarrow w}) \geq \frac{1}{q}>u_i(\mathbf{v}).
		\end{align*}
		Hence, $w$ is not adjacent to any agent from $T_2$. Recall that as 
		we assume that our input instance is regularized (see 
		\Cref{np:jump:lemma_ass}),  
		in the input instance, for every vertex $v \notin V_{S'}$ not 
		occupied by a stubborn agent, 
		there 
		exist two adjacent vertices $s_i, s_j \in V_{S'}$ occupied by 
		stubborn 
		agents $i 
		\in T'_1$ and $j \in T'_2$. Thus, since $w \in V(G') 
		\setminus 
		V_{S'}$, the vertex $w$ is adjacent to a central vertex in $B$, 
		which is 
		occupied by an agent from $T_1$. Then again, any agent from $T_1$ 
		on a central 
		vertex in $A$ can increase its utility to $1$ by jumping to $w$. 
		This concludes 
		Case 3. Since we have exhausted all possible cases, Claim 
		\ref{np:jump:claim1.1} 
		follows.
	\end{proof}
	Next, we prove the second part of \Cref{np:jump:claim1} by proving that 
	all 
	vertices in $A$ are occupied.
	\begin{claimS}\label{np:jump:claim1.2}
		In every jump-equilibrium $\mathbf{v}$ in $\mathcal{I}$, no vertex 
		in $A$ is unoccupied. 
	\end{claimS}
	\begin{proof}[Proof of Claim]\renewcommand{\qedsymbol}{$\blacklozenge$}
		Suppose there exists an unoccupied vertex $v \in A$. Again, by 
		\Cref{np:jump:claim0}, $v$ has 
		to be a 
		degree-one vertex and by \Cref{np:jump:claim0} and 
		\Cref{np:jump:claim1.1} $v$ needs to be adjacent to a vertex 
		occupied by an agent from $T_1$. Note that $|T_2|>|V(G')|$, 
		therefore by \Cref{np:jump:claim1.1} there have 
		to be 
		agents from $T_2$ in $B$ (as there are no agents from $T_2$ in $A$). By \Cref{np:jump:claim0}, this implies 
		that a 
		central vertex in $B$ is occupied by an agent from $T_2$. Since we 
		also have 
		that $|T_1|\geq |A|$, there are agents from $T_1$ outside of $A$. 
		If there is an agent from $T_1$ in $B$, there 
		also 
		exists an agent $i \in T_1$ on a central vertex in $B$. We have 
		that 
		$u_i(\mathbf{v})<1$, since $i$ is adjacent to the central vertex 
		occupied by 
		an agent from $T_2$. However, agent $i$ can then increase its 
		utility to $1$ 
		by jumping to $v$. Hence, all remaining agents from $T_1$ not in 
		$A$ are in 
		$V(G') \setminus V_{S'}$. Since all central vertices in $B$ have to 
		be 
		occupied, there are agents from $T_2$ on all central vertices in 
		$B$.
		
		Now consider an arbitrary agent $i \in  T_1$ in $V(G') \setminus 
		V_{S'}$.  Since we assume the given instance to be regularized (see 
		\Cref{np:jump:lemma_ass}), $i$ is adjacent to a central vertex in 
		$B$, which is occupied by an agent from $T_2$. We thus have 
		$u_i(\mathbf{v})<1$. Agent $i$ can increase its utility to~$1$ by 
		jumping to $v$, which completes the proof.
	\end{proof}
	With Claim \ref{np:jump:claim1.1} and Claim \ref{np:jump:claim1.2}, 
	\Cref{np:jump:claim1} 
	follows.
\end{proof}

Next, we prove the following analogous claim for $B$, where we show that 
all vertices in~$B$ are occupied by agents from $T_2$.

\begin{lemma}\label{np:jump:claim2}
	In every jump-equilibrium $\mathbf{v}$ in $\mathcal{I}$, there is no 
	unoccupied vertex in $B$ and every vertex in $B$ is 
	occupied by an agent 
	from $T_2$. 
\end{lemma}
\begin{proof}
	Again, we prove this claim by dividing it into 
	\Cref{np:jump:claim2.1} and \Cref{np:jump:claim2.2}.
	\begin{claimS}\label{np:jump:claim2.1}
		In every jump-equilibrium $\mathbf{v}$ in $\mathcal{I}$, no vertex 
		from  $B$ is occupied by an agent from~$T_1$. 
	\end{claimS}
	\begin{proof}[Proof of Claim]\renewcommand{\qedsymbol}{$\blacklozenge$}
		Since by \Cref{np:jump:claim1} $A$ is fully occupied by agents 
		from $T_1$, there are $y 
		\coloneqq 
		|T_1|-|A|<|V(G')|$ agents from $T_1$ outside of $A$. Suppose there 
		are agents 
		from $T_1$ in $B$.  By  \Cref{np:jump:claim0}, there exists a 
		central vertex 
		$v \in B$ which is occupied by an agent from $T_1$. Furthermore, 
		since 
		$y<|V(G')|^2$, there exists an unoccupied degree-one vertex 
		adjacent to $v$. 
		Note that $|T_2| \geq |B|$. As not all vertices in $B$ are 
		occupied by 
		agents from $T_2$, there exists an agent $j \in T_2$ on a vertex in 
		$V(G') 
		\setminus V_{S'}$. Since the given instance is regularized (see 
		\Cref{np:jump:lemma_ass}), this agent is adjacent to a central 
		vertex in 
		$A$. Let $i' \in T_1$ be the agent on 
		this 
		central 
		vertex. We 
		have $u_{i'}(\mathbf{v})<1$, since $i' \in T_1$ is adjacent to $j 
		\in T_2$. 
		Then, jumping to $v$ is profitable for agent $i'$, since 
		$u_{i'}(\mathbf{v})<1=u_{i'}(\mathbf{v}^{i' \rightarrow v})$.
	\end{proof}
	
	\begin{claimS}\label{np:jump:claim2.2}
		In every jump-equilibrium $\mathbf{v}$ in $\mathcal{I}$, no vertex 
		in $B$ is unoccupied. 
	\end{claimS}
	\begin{proof}[Proof of Claim]\renewcommand{\qedsymbol}{$\blacklozenge$}
		Suppose there exists an unoccupied vertex $v \in B$. By 
		\Cref{np:jump:claim0}, the vertex $v$ 
		is
		a degree-one vertex. Since all central vertices have to be occupied 
		and by \Cref{np:jump:claim2.1}, $B$ 
		only contains agents from $T_2$, the vertex $v$ is adjacent to a 
		central 
		vertex occupied by an agent from $T_2$. It holds that $|T_2| \geq 
		|B|$. Since 
		at least one vertex in $B$ is unoccupied, there exists an agent $i 
		\in T_2$ on 
		a vertex in $V(G') \setminus V_{S'}$. We have that 
		$u_{i}(\mathbf{v})<1$, 
		since $i$ is adjacent to a central vertex in $A$ (by our assumption 
		that the given instance is regularized; see 
		\Cref{np:jump:lemma_ass}), which is occupied 
		by an 
		agent from $T_1$ by \Cref{np:jump:claim1}. Agent $i$ can increase 
		her utility 
		to $1$ by jumping to $v$.
	\end{proof}
	
	\Cref{np:jump:claim2} follows from 
	\Cref{np:jump:claim2.1} and \Cref{np:jump:claim2.2}.
\end{proof}
Finally, we prove the correction of the backwards direction of the construction, i.e., that the input Schelling game with stubborn agents admits 
a 
jump-equilibrium only if the constructed Schelling game without 
stubborn agents admits a 
jump-equilibrium:

\begin{lemma} \label{le:jumpbd}
	If the constructed instance $\mathcal{I}$ of \jEq{} admits a 
	jump-equilibrium, then the given instance $\mathcal{I}'$ of 
	\stubjEq{} admits a jump-equilibrium.
\end{lemma}
\begin{proof}
	Assume that there exists a jump-equilibrium 
	$\mathbf{v}$ in the constructed instance $\mathcal{I}$. We define an 
	assignment 
	$\mathbf{v'}$ 
	on $G'$ for the given instance $\mathcal{I}'$  with stubborn agents as 
	follows and afterwards 
	prove that 
	it is a jump-equilibrium.  In $\mathbf{v'}$, an occupied vertex $v \in 
	V(G') 
	\setminus V_{S'}$ is occupied by a strategic agent  of the same type as the 
	agent on $v$ in $\mathbf{v}$. If a vertex is unoccupied in $\mathbf{v}$, 
	then 
	it is also unoccupied in $\mathbf{v}'$. The vertices in $V_{S'}$ have to be 
	occupied by the respective stubborn agents. Note that by 
	\Cref{np:jump:claim1,np:jump:claim2}, in $\mathbf{v}$, the vertices in  
	$V_{S'_1} 
	\subseteq A$ and  $V_{S'_2} \subseteq B$ have to be occupied by agents from 
	$T_1$ 
	and $T_2$, respectively.  Thus, in $\mathbf{v'}$, all occupied vertices are 
	occupied by agents of the same type as in $\mathbf{v}$. Note that in 
	$\mathbf{v'}$ we have assigned exactly $|T'_1|$ agents from $T'_1$ and 
	$|T'_2|$ agents from $T'_2$ as it holds that $|T_1|=|T'_1|+|M_1|+|X_1|$  
	and 
	$|T_2|=|T'_2|+|M_2|$ and by \Cref{np:jump:claim1,np:jump:claim2} we know 
	that in $\mathbf{v}$ all vertices from $M_1\cup X_1$ are occupied by agents 
	from $T_1$ and all vertices from $M_2$ are occupied by agents from $T_2$.
	
	Next, we prove that $\mathbf{v}'$ is a jump-equilibrium in the given 
	Schelling game. Since the 
	agents on vertices in $V_{S'}$ are stubborn, only agents on vertices in 
	$V(G') 
	\setminus V_{S'}$ can have a profitable jump. However, all 
	unoccupied vertices are in $V(G') \setminus V_{S'}$ in both $\mathbf{v}$ 
	and  $\mathbf{v}'$ and the neighborhood of vertices from $V(G') \setminus 
	V_{S'}$ is the 
	same in $G$ and $G'$. A 
	profitable jump for an agent on a vertex in $V(G') \setminus V_{S'}$ to an unoccupied 
	vertex from $V(G') \setminus V_{S'}$ in $\mathbf{v}'$ would 
	therefore also be a profitable jump in $\mathbf{v}$, thereby contradicting 
	that $\mathbf{v}$ is a jump-equilibrium. Thus, $\mathbf{v'}$ is a 
	jump-equilibrium.
\end{proof}

Note that the membership of \jEq in NP is trivial, as it is possible to verify 
that a given assignment is a jump equilibrium by iterating over pairs of 
occupied vertices and unoccupied vertices and checking whether the jump of the 
agent from the occupied vertex to the unoccupied vertex is profitable.
Thus, from \Cref{le:jumpfd} and \Cref{le:jumpbd}, \Cref{np:jumpEqEx} follows: 
\jumhard*

\section{Robustness of Equilibria} \label{ch:rob}
Having established in \Cref{ch:np_nostubborn} that deciding the existence 
of an equilibrium is NP-hard, we now introduce the concept of 
\textit{robustness} of an 
equilibrium. We consider both the robustness of an equilibrium with respect 
to the deletion of edges and with respect to the deletion of vertices, where 
deleting a vertex 
also implies deleting the agent occupying the vertex from the Schelling game.

\begin{definition}\label{r-robust} 
	For a Schelling game $I$ on a graph $G$, an equilibrium $\mathbf{v}$ in $I$ 
	is \emph{$r$-edge-robust} (\emph{$r$-vertex-robust})
	for some $r \in \NN_0$ if $\mathbf{v}$ is an
	equilibrium in $I$ on the topology~$G-S$ for all subsets of edges $S \subseteq 
	E(G)$ (for all subsets of vertices $S\subseteq V(G)$)  with $|S| \leq 
	r$. The \emph{edge-robustness} (\emph{vertex-robustness}) 
	of 
	$\mathbf{v}$ is the largest~$r\in[0,|E(G)|]$ ($r\in[0,|V(G)|]$) for which  
	$\mathbf{v}$ is $r$-edge-robust ($r$-vertex-robust).
\end{definition} 
Note that given an equilibrium $\mathbf{v}$ that becomes 
unstable after deleting~$r$ edges/vertices, deleting further edges/vertices can 
make 
$\mathbf{v}$ 
stable again, as any assignment is stable if we delete all edges/vertices. That is why in \Cref{r-robust} we require that $\mathbf{v}$ is stable for all $S$ with 
$|S| \leq r$ and not only for all~$S$ with $|S|=r$. 

We focus on the robustness of swap-equilibria. While for jump-equilibria the introduced concepts are 
also meaningful, already obtaining lower and upper bounds on the robustness 
of a jump-equilibrium on a single fixed graph is problematic, as the robustness of a
jump-equilibrium significantly depends on the number of unoccupied vertices. 
For instance, if the number of agents of both types is small, on most graphs the agents 
can be placed such that all agents are only adjacent to friends making the 
equilibrium~quite~robust.

Note that as we restrict our attention to swap-equilibria, deleting agents and 
vertices is equivalent, as in this context a once unoccupied vertex can never 
become occupied again and is also irrelevant for computing utilities. Thus, all 
our results on vertex-robustness also apply to the robustness with 
respect to the deletion of agents.
However, when analyzing jump-equilibria one would have to distinguish the two 
and consider the robustness of a jump-equilibrium with respect to the deletion 
of 
agents (but not vertices) and the deletion of vertices 
(and the agents occupying them).
\subsection{First Observations}

We start by observing that if for one type there exists only a single agent,
then there never is a profitable swap.
Hence, in this case, every assignment trivially is a swap-equilibrium of 
edge-robustness~$|E(G)|$ and vertex-robustness~$|V(G)|$. That is why, in the 
following, we assume 
that~$\min\{|T_1|,|T_2|\}\ge 2$. Further, if we consider vertex-robustness, 
then in case that $|T_1|=2$ and $|T_2|=2$, after deleting an agent from 
one type the other agent from this type will always have zero utility and 
cannot be part of a profitable swap, implying that any swap-equilibrium 
has vertex-robustness~$|V(G)|$.  Thus, in  the following, considering 
vertex-robustness, we assume that $|T_1|\geq 3$ and $|T_2|\geq 2$. 

Focusing on edge-robustness for a moment, only the deletion of 
edges between agents of the same type has 
an influence on the stability of a swap-equilibrium.
This is stated more precisely in the following proposition.

\begin{proposition}\label{swap:corollaryEdgeSameType}
	Let $\mathbf{v}$ be a swap-equilibrium  for a Schelling game on 
	topology~$G$. Let $S \subseteq E(G)$ be a set of edges such that 
	$\mathbf{v}$ 
	is not a swap-equilibrium on~$G-S$. Then,
	\begin{enumerate}[label=(\roman*)]
		\item $S$ contains at least one edge between agents of the same type.
		\item  $\mathbf{v}$ is also not a swap-equilibrium on 
		$G-S'$, where $S'\subseteq S$ is the subset of edges from $S$ that 
		connect agents of the same type.
		\item For every set $A\subseteq \{ \{v_i,v_j\} \in 
		E(G) \mid i,j \in T_1  \vee i,j \in T_2 \}$ of 
		edges 
		between 
		agents of the same type, $\mathbf{v}$ is also not a swap-equilibrium on
		$G-(S 
		\cup A)$.
	\end{enumerate}
\end{proposition}
\begin{proof}
	We start by proving statement~(i). Let $S$ be a set of edges between 
	agents of different types. Consider the game $I$ on topology 
	$G-S$ and assume for the sake of contradiction that $\mathbf{v}$ is not a 
	swap-equilibrium, that is, there exist agents $i \in 
	T_1$ and $j \in T_2$ that want to swap. We therefore have 
	$u_i^{G-S}(\mathbf{v})<u_i^{G-S}(\mathbf{v}^{i \leftrightarrow j })$ and 
	$u_j^{G-S}(\mathbf{v})<u_j^{G-S}(\mathbf{v}^{i \leftrightarrow j })$.
	
	On the original topology $G$, it holds that $u_i^{G}(\mathbf{v})\geq 
	u_i^{G}(\mathbf{v}^{i \leftrightarrow j })$, because $\mathbf{v}$ is a 
	swap-equilibrium on $G$. Since all of the edges in $S$ are between agents 
	of different types, no edges to friends of $i$ are deleted in $G-S$ and it 
	therefore holds that $u_i^{G}(\mathbf{v})\leq u_i^{G-S}(\mathbf{v})$. 
	Hence, we have:
	\begin{equation} \label{swap:lemmaEdgeSameType:ineq}
	u_i^{G}(\mathbf{v}^{i \leftrightarrow j })\leq u_i^{G}(\mathbf{v})\leq 
	u_i^{G-S}(\mathbf{v}) < u_i^{G-S}(\mathbf{v}^{i \leftrightarrow j 
	}).                                                              
	\end{equation}
	Consider the vertex $v_j$ that is occupied by agent $j\in T_2$ in 
	assignment $\mathbf{v}$.  Since all of the edges in $S$ are between agents 
	of different types, the only edges incident to $v_j$ that have been deleted 
	in $G-S$ are edges to agents in $T_1$. Therefore,
	$u_i^{G-S}(\mathbf{v}^{i \leftrightarrow j }) \leq u_i^{G}(\mathbf{v}^{i 
		\leftrightarrow j })$. This contradicts \Cref{swap:lemmaEdgeSameType:ineq} 
	and completes the proof of the first statement.
	
	We now turn to statement~(ii).  Assume for the sake of contradiction, 
	that $\mathbf{v}$ is a swap-equilibrium on $G-S'$. Notice that 
	$X=S\setminus S'$ only contains edges between agents of different types.  
	Hence, we can apply the first statement to the fact that $\mathbf{v}$ is a 
	swap-equilibrium on $G-S'$ and get that $\mathbf{v}$ is a 
	swap-equilibrium on $G-(S'\cup X)=G-S$, leading to a contradiction.
	
	Lastly, we consider statement~(iii). Let $i\in T_1$ and $j 
	\in T_2$ be a pair of agents that has a profitable swap on $G-S$ in 
	$\mathbf{v}$. We will 
	now argue that the swap is also profitable on $G-(S \cup A)$. Consider the 
	vertex $v_i$ that is occupied by agent~$i$. Since $A$ only contains edges 
	between agents of the same type, we only delete edges to friends of $i$ in 
	the neighborhood of $v_i$. Hence, the utility of agent $i$ in $\mathbf{v}$ 
	on topology $G-(S \cup A)$ is at most as high as the utility on 
	$G-S$. In the neighborhood of $v_j$, we only delete edges to agents in 
	$T_2$. Therefore, the utility of $i$ after swapping to $v_j$ on $G-(S \cup 
	A)$ has to be at least as high as on $G-S$. By symmetry, the same holds for 
	agent $j$. Hence, the swap is profitable and $\mathbf{v}$ is also not a 
	swap-equilibrium on $G-(S \cup A)$.
\end{proof}

For vertex-robustness, one can similarly observe that only deleting 
a vertex occupied by an agent $a$
adjacent to at least one vertex occupied by a friend of $a$ can 
make a swap-equilibrium unstable. 

Next, note that the utility of an agent only depends on its neighborhood. 
Thus, whether two agents $i$ and $j$ have a profitable swap in $G-S$ only 
depends on the edges/vertices incident/adjacent to $v_i$ and $v_j$ in $S$. Combining 
this with 
the observation that no profitable swap can involve an agent on an 
isolated vertex, it follows that if 
a swap-equilibrium cannot be made unstable by deleting $2\cdot (\Delta(G)-1)$ 
edges/vertices, then it cannot be made 
unstable by deleting an arbitrary number of edges/vertices:
\begin{observation} 
	Let $\mathbf{v}$ be a swap-equilibrium for a Schelling game on
	$G$. If $\mathbf{v}$ is $2\cdot( \Delta(G)-1)$-edge-robust,  
	$\mathbf{v}$ has edge-robustness $|E(G)|$ and if $\mathbf{v}$ is~$2\cdot( 
	\Delta(G)-1)$-vertex-robust,  
	$\mathbf{v}$ has vertex-robustness $|V(G)|$.
\end{observation}

The simple fact that the utility of an agent only depends on 
its neighborhood leads to a polynomial-time algorithm to 
determine whether a given swap-equilibrium $\mathbf{v}$ has edge-robustness $r 
\in 
\NN_0$: We simply iterate
over all pairs of agents~$i$ and $j$ and check whether we can delete at most 
$r$ edges between~$v_i$ and adjacent vertices occupied by friends of $i$ and 
between $v_j$ and adjacent vertices occupied by friends of $j$
such that the swap of~$i$ and~$j$ becomes profitable (note that the stability 
of $\mathbf{v}$ only depends on the number of such deleted edges in the 
neighborhood of each agent, not 
the exact subset of edges).\footnote{A very similar algorithm can also be used for determining the robustness of jump-equilibria. Instead of 
iterating 
over pairs of vertices, we need to iterate over all pairs of agents and unoccupied 
vertices.}
\begin{proposition}\label{comp:rRobust}
	Given a Schelling game with $n$ agents, a swap-equilibrium~$\mathbf{v}$, 
	and 
	an integer $r \in \NN_0$, one can decide in $\bigO(n^2\cdot r)$ time 
	whether $\mathbf{v}$ is $r$-edge/vertex-robust.
\end{proposition}
\begin{proof}
	We describe in detail how to solve the problem for edge-robustness. 
	For vertex-robustness, an analogous approach can be used. Recall the definition 
	of $a_i(\mathbf{v})$ as the number of 
	friends adjacent to  
	agent $i$ in $\mathbf{v}$ and let  
	$b_i(\mathbf{v})\coloneqq|N_{i}(\mathbf{v})|-a_i(\mathbf{v})$ be the 
	number of 
	agents of the other type in the neighborhood of $i$. We define 
	$\mathbbm{1}_{i,j}=1$ if the agents $i$ and $j$ are neighbors and 
	$\mathbbm{1}_{i,j}=0$ otherwise.
	
	We solve the problem using \Cref{alg:equirobustness} for which we prove the 
	correctness and running time in the following:
	First, we prove that \Cref{alg:equirobustness} outputs \textit{yes} if 
	$\mathbf{v}$ is $r$-edge-robust and \textit{no} otherwise. Assume that 
	$\mathbf{v}$ 
	is $r$-edge-robust. It therefore holds for all $i \in T_1$ and $j \in T_2$ and 
	$S 
	\subseteq E(G)$ with $|S|\leq r$ that swapping $i$ and $j$ is not 
	profitable in 
	$G-S$. Hence, it holds that the swap of $i$ and $j$ is not profitable if we 
	delete $x$ 
	edges between $v_i$ and vertices that are occupied by friends of $i$ and 
	$y$ edges between $v_j$ and vertices that are occupied by friends of $j$  
	for all $x,y \in 
	\NN_0$ with 
	$x+y \leq r$ and  $x\leq a_i(\mathbf{v})$ 
	and 
	$y\leq a_j(\mathbf{v})$. 
	After deleting these edges, the utilities of $i$ and $j$ are given by  
	$u'_i(\mathbf{v}) = \frac{a_i(\mathbf{v})-x}{|N_i(\mathbf{v})|-x}$ and 
	$u'_j(\mathbf{v}) = \frac{a_j(\mathbf{v})-y}{|N_j(\mathbf{v})|-y}$. 
	Swapping 
	positions results in utility $u'_i(\mathbf{v}^{i \leftrightarrow j}) = 
	\frac{b_j(\mathbf{v})-\mathbbm{1}_{i,j}}{|N_j(\mathbf{v})|-y}$ for $i$ 
	and 
	$u'_j(\mathbf{v}^{i \leftrightarrow j}) = 
	\frac{b_i(\mathbf{v})-\mathbbm{1}_{i,j}}{|N_i(\mathbf{v})|-x}$ for $j$. 
	Notice 
	that we subtract $\mathbbm{1}_{i,j}=1$ in the numerator if $i$ and $j$ 
	are 
	adjacent, since the vertex previously occupied by $i$ or $j$ in 
	$\mathbf{v}$ is 
	occupied by the other agent of a different type in $\mathbf{v}^{i 
		\leftrightarrow j}$. Since, by our assumption, $\mathbf{v}$ is 
	$r$-edge-robust, it 
	has to hold that $u'_i(\mathbf{v}) \geq u'_i(\mathbf{v}^{i 
		\leftrightarrow j})$ 
	or $u'_j(\mathbf{v}) \geq u'_j(\mathbf{v}^{i \leftrightarrow j})$. 
	Hence, \Cref{alg:equirobustness} outputs \textit{yes}.
	
	Assume that $\mathbf{v}$ is not $r$-edge-robust. Then, there exists a pair 
	of agents 
	$i,j \in N$ and a set of edges $S \subseteq E(G)$ with $|S|\leq r$ such 
	that 
	the swap involving $i$ and $j$ is profitable on $G-S$. Note that, as argued before, 
	we only 
	have to consider deleting edges to friends in the neighborhoods of $i$ 
	and $j$. 
	Therefore, there exist $w,z \in \NN_0$ with $w+z \leq r$, $w\leq 
	a_i(\mathbf{v})$ and $z\leq a_j(\mathbf{v})$ such that swapping $i$ and 
	$j$ is 
	profitable after deleting $w$ edges to adjacent friends of $i$ and $z$ edges to 
	adjacent friends 
	of $j$. As proven in \Cref{swap:corollaryEdgeSameType}, deleting additional 
	edges between $i$ and friends of $i$ and $j$ and friends of $j$ cannot make 
	$\mathbf{v}$ stable again. Thus, it holds for all $w',z' \in \NN_0$ with 
	$w\leq w'\leq a_i(\mathbf{v})$ and $z\leq z' \leq a_j(\mathbf{v})$ that 
	swapping $i$ and 
	$j$ is 
	profitable after deleting $w'$ edges to adjacent friends of $i$ and $z'$ 
	edges to 
	adjacent friends 
	of $j$. Thus, in \Cref{alg:equirobustness} with $x=w$ and 
	$y=\min\{r-x,a_j(\mathbf{v})\}\geq z$, we have $u'_i(\mathbf{v}) < 
	u'_i(\mathbf{v}^{i \leftrightarrow 
		j})$ and 
	$u'_j(\mathbf{v}) < u'_j(\mathbf{v}^{i \leftrightarrow j})$ and 
	therefore 
	return \textit{no}.
	
	Next, we analyze the running time. We first iterate over all pairs of 
	agents $i 
	\in T_1$ and $j \in T_2$ that can potentially be involved in a 
	profitable swap, 
	the number of pairs is upper-bounded by $n^2$. For each pair, we 
	iterate 
	over at most $r$ possible values for $x$ from $\min\{r,a\}$ to zero. 
	All other 
	operations are simple arithmetic operations that can be computed in 
	constant 
	time (assuming we precomputed all $|N_i(\mathbf{v})|$ and $a_i(\mathbf{v})$ in linear time), hence our algorithm runs in $\bigO(n^2\cdot r)$ time.
	\begin{algorithm}[t]
		\caption{Robustness of a Swap-Equilibrium}\label{alg:equirobustness}
		\begin{algorithmic}[1]
			\Input{Topology $G$, swap-equilibrium $\mathbf{v}$, sets of agents 
				$T_1$ and $T_2$ and $r 
				\in \NN_0$}
			\Output{Yes if $\mathbf{v}$ is $r$-edge-robust and No otherwise}
			\Function{Robustness}{$G,\mathbf{v},T_1,T_2,r$}
			\For{each pair $i \in T_1$ and $j \in T_2$} \Comment{Iterate over 
				all possible swaps}
			\For{$x=\min\{r,a_i(\mathbf{v})\}$ \textbf{to} 0} \Comment{Number 
				of edges we delete in $N_G(v_i)$}
			\State $y \gets \min\{r-x,a_j(\mathbf{v})\}$ \Comment{Delete 
				remaining edges in $N_G(v_j)$}
			\State $u'_i(\mathbf{v}) \gets 
			\frac{a_i(\mathbf{v})-x}{|N_i(\mathbf{v})|-x}$ \Comment{Utility of 
				$i$ after deleting edges to $x$ friends}
			\State $u'_i(\mathbf{v}^{i \leftrightarrow j}) \gets 
			\frac{b_j(\mathbf{v})-\mathbbm{1}_{i,j}}{|N_j(\mathbf{v})|-y}$ 
			\Comment{\parbox[t]{.5\linewidth}{Utility of $i$ on $v_j$ after deleting $y$ edges in 
					$N_G(v_j)$}}
			\State $u'_j(\mathbf{v}) \gets 
			\frac{a_j(\mathbf{v})-y}{|N_j(\mathbf{v})|-y}$
			\State $u'_j(\mathbf{v}^{i \leftrightarrow j}) \gets 
			\frac{b_i(\mathbf{v})-\mathbbm{1}_{i,j}}{|N_i(\mathbf{v})|-x}$
			\If {$u'_i(\mathbf{v}) < u'_i(\mathbf{v}^{i \leftrightarrow j})$ 
				\textbf{and} $u'_j(\mathbf{v}) < u'_j(\mathbf{v}^{i 
					\leftrightarrow 
					j})$}
			\State \Return no \Comment{Swapping $i,j$ is a profitable swap}
			\EndIf
			\EndFor
			\EndFor
			\State \Return yes \Comment{No profitable swap is possible}
			\EndFunction
		\end{algorithmic}
	\end{algorithm}
\end{proof}

Note, however, that finding a swap-equilibrium whose vertex- or edge-robustness 
is as high as possible is NP-hard, as we have proven in 
\Cref{swapHardness:noStubborn} that already deciding whether a Schelling game 
admits some swap-equilibrium is NP-hard.

\subsection{Robustness of Equilibria on Different Graph 
	Classes}\label{sec:swap:topInfluence}
In this subsection, we analyze the influence of the topology on the robustness 
of 
swap-equilibria. We first analyze cliques where each swap-equilibrium has 
edge-robustness zero and vertex-robustness $|V(G)|$.
Subsequently, we turn to  
cycles, paths, and grids and find 
that there exists a swap-equilibrium on all these graphs with edge-robustness 
and vertex-robustness zero. 
For paths, we observe that 
the difference between the edge/vertex-robustness of the most and least robust 
equilibrium  can be arbitrarily large. Finally, with 
$\alpha$-star-constellation 
graphs for $\alpha\in \mathbb{N}_0$, we present a class of graphs on which all 
swap-equilibria 
have at least edge/vertex-robustness~$\alpha$.
For $\alpha$-star-constellation graphs we also obtain a precise characterization of swap equilibria and a polynomial-time algorithm for checking whether a swap equilibrium exists. 

\paragraph{Cliques.} We start by observing that on a clique every assignment is a swap-equilibrium. 
From this it directly follows that every swap-equilibrium has 
vertex-robustness $|V(G)|$, as deleting a vertex from a clique results in 
another clique. 
In contrast, each swap-equilibrium can be made unstable by deleting one
edge.
Thereby, the following observation also proves that the difference between the 
edge- and vertex-robustness of a swap-equilibrium can be arbitrarily large:

\begin{observation}\label{edgedel:clique}
	In a Schelling game on a clique $G$ with~$|T_1|\ge 2$ and $|T_2|\ge 2$, every swap-equilibrium $\mathbf{v}$ 
	has 
	edge-robustness zero and vertex-robustness $|V(G)|$.
\end{observation}
\begin{proof}
	It remains to prove that the edge-robustness is always zero.
	Let $i\neq j \in T_1$, $e:=\{v_i,v_j\}\in E(G)$, and $l\in T_2$. As $G$ is 
	a clique, it holds that 
	$u^G_i(\mathbf{v})= 
	\frac{|T_1|-1}{|T_1|+|T_2|-1}$ and $u^G_l(\mathbf{v})= 
	\frac{|T_2|-1}{|T_1|+|T_2|-1}$. Swapping $i$ and $l$ is 
	profitable 
	in $\mathbf{v}$ on $G-\{e\}$ for 
	both $i$ and~$l$, as
	\begin{align*}&u^{G-\{e\}}_i(\mathbf{v}^{i 
		\leftrightarrow
		l})=\frac{|T_1|-1}{|T_1|+|T_2|-1}>\frac{|T_1|-2}{|T_1|+|T_2|-2}=u^{G-\{e\}}
	_i(\mathbf{v}) \text{ and }\\
	&u_l^{G-\{e\}}(\mathbf{v}^{i 
		\leftrightarrow
		l})=\frac{|T_2|-1}{|T_1|+|T_2|-2}>\frac{|T_2|-1}{|T_1|+|T_2|-1}=u^{G-\{e\}}
	_l(\mathbf{v}).
	\end{align*}
\end{proof}

\paragraph{Cycles.} For a cycle $G$, we can show that in a swap-equilibrium $\mathbf{v}$, every 
agent is adjacent to at least one friend. Then, picking an arbitrary agent 
$i\in 
T_1$ that has utility $\nicefrac{1}{2}$ in~$\mathbf{v}$ and deleting $i$'s 
neighbor from~$T_1$ 
or the edge between $i$ and 
its neighbor from~$T_1$ makes $\mathbf{v}$ 
unstable.
\begin{proposition}\label{swap:cycle}
	In a Schelling game on a cycle $G$ with~$|T_1|\ge 2$ and $|T_2|\ge 2$, every swap-equilibrium~$\mathbf{v}$ 
	has 
	edge-robustness zero. For $|T_1|\geq 3$ and 
	$|T_2|\geq 2$, every swap-equilibrium~$\mathbf{v}$ has vertex-robustness~zero.
\end{proposition}
\begin{proof}
	First, we show that, if $|T_1|\geq 2$ and $|T_2|\geq 2$, then each agent is adjacent to at least one friend in a 
	swap-equilibrium~$\mathbf{v}$. For the sake 
	of contradiction, assume without loss of generality that there exists an agent 
	$i\in T_1$ that has $u_i(\mathbf{v})=0$. 
	Let $i'\neq i\in T_1$ be a different agent from $T_1$ with 
	$u_{i'}(\mathbf{v})\leq \frac{1}{2}$ (such an agent needs to exist, as we 
	assume that $|T_1|\geq 2$, $|T_2|\geq 2$, and that $G$ is a cycle). 
	We now distinguish two cases. If $i'$ is adjacent to an agent $j\in T_2$ who 
	is also adjacent to $i$, then swapping $i$ and $j$ is profitable, as 
	$u_i(\mathbf{v}^{i 
		\leftrightarrow j})=\frac{1}{2}>0=u_i(\mathbf{v})$ and $u_j(\mathbf{v}^{i 
		\leftrightarrow j})=\frac{1}{2}>0=u_j(\mathbf{v})$. 
	Otherwise, let $j\in T_2$ be an agent from $T_2$ adjacent to $i'$. Then, 
	swapping $i$ and $j$ is profitable, as $u_i(\mathbf{v}^{i 
		\leftrightarrow j})\geq\frac{1}{2}>0=u_i(\mathbf{v})$ and 
	$u_j(\mathbf{v}^{i 
		\leftrightarrow j})=1>\frac{1}{2}\geq u_j(\mathbf{v})$. 
	
	Let $w_1, \dots, w_l$ be a maximal subpath 
	of vertices occupied by agents from $T_1$ and let $i\in T_1$ be the agent 
	occupying $w_1$ and $j\in T_1$ be 
	the agent occupying $w_l$ in $\mathbf{v}$. By 
	the above argument it holds that $l\geq 2$. Since we have at least two agents 
	from $T_2$, agent $i$ on $w_1$ and 
	agent $j$ on $w_l$ 
	need to be adjacent to agents $i' \in 
	T_2$ and $j' \in T_2$, respectively, with $i' \neq j'$.
	
	We first consider edge-robustness. Let $S:=\{\{w_1,w_2\}\}$.
	On $G-S$, agent $i$ and agent $j'$ have a profitable swap, as 
	they are not adjacent and thus it holds that $u^{G-S}_i(\mathbf{v}^{i 
		\leftrightarrow j'})\ge\frac{1}{2}>0=u^{G-S}_i(\mathbf{v})$ and 
	$u^{G-S}_{j'}(\mathbf{v}^{i 
		\leftrightarrow j'})=1>\frac{1}{2}\geq u_{j'}(\mathbf{v})$. 
	
	We now turn to vertex-robustness. Let $S:=\{w_2\}$. If $l\ge 3$, then $j$ has 
	not been deleted and again swapping 
	$i$ and $j'$ is profitable, as  $u^{G-S}_i(\mathbf{v}^{i 
		\leftrightarrow j'})\ge\frac{1}{2}>0=u^{G-S}_i(\mathbf{v})$ and 
	$u^{G-S}_{j'}(\mathbf{v}^{i 
		\leftrightarrow j'})=1>\frac{1}{2}\geq u_{j'}(\mathbf{v})$. Otherwise, if 
	$l=2$, as $G$ is a cycle and $|T_1|\geq 3$, there needs to exists an agent 
	$p\in T_1\setminus \{i,j\}$ that is adjacent to an agent $q\in T_2$. By our 
	observation from above that each agent in $\mathbf{v}$ needs to be adjacent to 
	at least one friend, it follows that $q\neq i'$ and $q\neq j'$. This implies 
	that $q$ and $i$ have a profitable swap, as $u^{G-S}_i(\mathbf{v}^{i 
		\leftrightarrow q})\ge\frac{1}{2}>0=u^{G-S}_i(\mathbf{v})$ and 
	$u^{G-S}_q(\mathbf{v}^{i 
		\leftrightarrow q})=1>\frac{1}{2}\geq u_q(\mathbf{v})$.
\end{proof}

\paragraph{Paths.} Next, we turn to paths and prove that every Schelling game on a path with 
sufficiently many agents from both types has an equilibrium with 
edge-/vertex-robustness zero and one with edge-robustness 
$|E(G)|$ and vertex-robustness $|V(G)|$. This puts paths in 
a surprisingly sharp contrast to 
cycles. The reason for this is that on a path, we can 
always position the agents such that there exists only one edge between agents of 
different types, yielding a swap-equilibrium with edge-robustness $|E(G)|$ and 
vertex-robustness~$|V(G)|$.
This is not possible on a cycle.

\begin{theorem}\label{path}
	For a Schelling game on a path $G$ with~$|T_1|\ge 4$ and~$|T_2|\ge 2$, 
	there exists a swap-equilibrium 
	$\mathbf{v}$ that has 
	edge-robustness and vertex-robustness zero and a 
	swap-equilibrium~$\mathbf{v}'$ that has edge-robustness 
	$|E(G)|$ and vertex-robustness $|V(G)|$.
\end{theorem}     
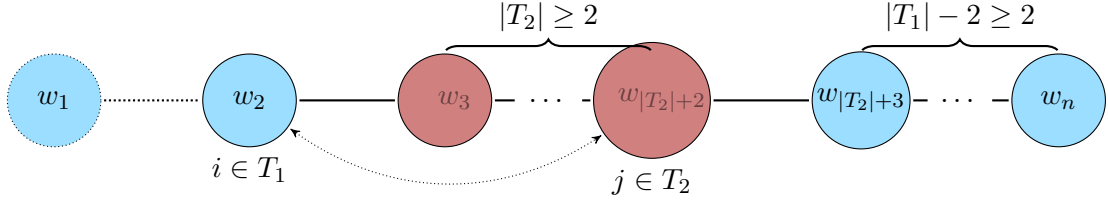
\begin{figure}[t]
	\tikzstyle{alter}=[circle, minimum size=35pt, draw, inner sep=1pt] 
	\tikzstyle{majarr}=[draw=black, thick]
	\centering
	\begin{tikzpicture}[auto, >=stealth',shorten <=1pt, shorten >=1pt]
		\node[alter, densely dotted , fill=c_T1] at (0,0) (a) {$w_1$};
		\node[alter, right = 8ex of a, fill=c_T1, label=below:$i \in T_1$] (b) 
		{$w_2$};
		\node[alter, right = 8ex of b, fill=c_T2, fill opacity=0.5, text opacity=1] (c) {$w_3$};
		\node[alter, right = 8ex of c, fill=c_T2, fill opacity=0.5, text opacity=1, label=below:$j \in T_2$] (d) 
		{$w_{|T_2|+2}$};
		\node[alter, right = 8ex of d, fill=c_T1] (e) {$w_{|T_2|+3}$};
		\node[alter, right = 8ex of e, fill=c_T1] (f) {$w_{n}$};
		
		\draw[majarr,  densely dotted   ] (a) edge (b);
		\draw[majarr] (b) edge (c);
		\edgeDotted(c,d)
		\draw[majarr] (d) edge (e);
		\edgeDotted(e,f)
		
		\draw[thick,decorate,decoration={brace, amplitude=5pt}]
		($(c)+(0,18pt)$) --  ($(d)+(0,18pt)$) node[midway, above=4pt] 
		{$|T_2| \geq 2$};
		\draw[thick,decorate,decoration={brace, amplitude=5pt}]
		($(e)+(0,18pt)$) --  ($(f)+(0,18pt)$) node[midway, above=4pt] 
		{$|T_1|-2 \geq 2$};
		\draw[<->, bend right=35, densely dotted] (b) to (d);
		
		\end{tikzpicture}
	\caption{The swap-equilibrium with robustness zero from \Cref{path}. After 
		deleting $\{w_1, w_2\} \in E(G)$ or $w_1\in V(G)$, swapping $i$ and $j$ is 
		profitable.}
	\label{fig:top:paths:0}
\end{figure}
\begin{proof}
	Let~$V(G)=\{w_1,\ldots,w_n\}$ and $E(G)=\{\{w_i,w_{i+1}\}\mid i\in [n-1]\}$.
	In $\mathbf{v}$, vertices~$w_1$ and~$w_2$ are occupied by agents from~$T_1$, 
	vertices $w_3$ to $w_{|T_2|+2}$ are occupied by agents from $T_2$, and the 
	remaining $|T_1|-2\geq2$ vertices  are occupied by agents from~$T_1$ (see 
	\Cref{fig:top:paths:0} for a visualization). As all 
	agents have at 
	most  one 
	neighbor of the other type and at least one neighbor of the same type, for each 
	pair $i$, $j$ of agents of different types it holds that~$u_i(\mathbf{v}^{i 
		\leftrightarrow j})\le\nicefrac{1}{2}\leq u_i(\mathbf{v})$.
	Thus, $\mathbf{v}$ is a swap-equilibrium.
	Further, after deleting the edge between $w_1$ and $w_2$ or deleting the 
	vertex $w_1$, swapping the agent 
	on~$w_2$ with the agent on $w_{|T_2|+2}$ is profitable. It follows that 
	$\mathbf{v}$ has edge-robustness and vertex-robustness zero. 
	
	In $\mathbf{v}'$, the agents from $T_1$ occupy the first $|T_1|$ vertices and 
	agents from $T_2$ the remaining vertices. 
	Let $S \subseteq E(G)$ or $S\subseteq V(G)$ and consider $G - S$.
	As for $j\in [|T_1|-1] \cup  [|T_1|+2,n]$, in $G-S$, the agent on~$w_j$ got 
	deleted,
	has no neighbor, or is only adjacent to friends, it can never be involved 
	in a profitable swap. Further, swapping the agent on 
	$w_{|T_1|}$ and the agent on $w_{|T_1|+1}$ can also never be profitable, since 
	after the swap none of the two is adjacent to a friend. Thus, $\mathbf{v}'$ is a 
	swap-equilibrium on $G-S$.
\end{proof}
If $\max\{|T_1|,|T_2|\}\le 3$, which is not covered by \Cref{path}, then in every swap-equilibrium the path is split into two 
subpaths and agents from $T_1$ occupy one subpath and agents from $T_2$ occupy 
the 
other subpath. As argued in the proof of \Cref{path}, such an assignment  has 
edge-robustness $|E(G)|$ and vertex-robustness~$|V(G)|$. 

\paragraph{Grids.} Turning to grids, which besides paths form the class which has been 
most 
often 
considered in the context of Schelling's segregation model, for both vertex- 
and edge-robustness, we show using 
some more involved arguments that every 
swap-equilibrium has either robustness one or zero and that there exists an 
infinite class of Schelling games on grids admitting a swap-equilibrium  with 
robustness zero and one with robustness one.

To prove that every swap-equilibrium on a grid has edge-robustness and vertex-robustness at most one, we 
first define the concept of \textit{frames} of a grid (this notion also appears 
in the analysis of the price of anarchy of swap-equilibria on grids in the work of
\citeA{bilo2020topological}). Let $G$ be an $(x\times y)$-grid with
$V(G)=\{ (a,b) \in \NN \times \NN 
\mid a\leq x, b\leq y\}$. We refer to the set $B(G)$ of
vertices in the top row, bottom row, left column and right column as 
\emph{border vertices} (formally,  $B(G)=\{ (a,b) \in V(G) \mid 
a \in \{1,x\} \text{ or } b  \in \{1,y\} \}$). The first frame $F_1$ of~$G$ 
is the set of border vertices. The second frame $F_2$ of~$G$ is the set of 
border vertices of the grid that results from deleting the first frame from 
$G$. Further, for all $i>1$, the frame $F_i$ of~$G$ is the set of border 
vertices of the grid that results from deleting the frames $F_1,\dots , 
F_{i-1}$ from $G$.

\begin{theorem}\label{grid}
	\begin{enumerate}
		\item 	In a Schelling game with ~$|T_1|\ge 5$ and~$|T_2|\ge 
		5$ on an $(x\times y)$-grid with $x\geq 4$, 
		$y\geq 4$, the edge-robustness of a 
		swap-equilibrium is at 
		most one.
		\item In a Schelling game with $|T_1|\ge 4$ and~$|T_2|\ge 
		4$ on an $(x\times y)$-grid with  $x\geq 4$ and~$y\geq 4$, the 
		vertex-robustness of a 
		swap-equilibrium is at 
		most one. 
		\item   In a Schelling game with $|T_1|=|T_2|$ on an $(x\times y)$-grid 
		with  
		even $x\geq 4$ and~$y\geq 2$, there exists a 
		swap-equilibrium $\mathbf{v}$ with 
		edge- and vertex-robustness zero and a swap-equilibrium 
		$\mathbf{v}'$ with edge- and vertex-robustness one.
	\end{enumerate}
\end{theorem}
\begin{proof}
	In the following, let~$G$ be an~$(x\times y)$-grid. We start by proving the first statement.
	
	\paragraph*{\textbf{Edge-Robustness at Most~1.}} For the sake of 
	contradiction, assume that there exists a swap-equilibrium 
	$\mathbf{v}$ which is $2$-edge-robust. By induction over the frames of $G$, we 
	show that all vertices have to be occupied by agents from one type in 
	$\mathbf{v}$, which contradicts that we have at least five agents from both 
	types.
	
	\paragraph*{Base Case:} First, consider the frame $F_1$ and assume for the sake of contradiction  
	that there are agents from both types on vertices from $F_1$. Then, there exist two 
	adjacent agents $i \in T_1$ and $j\in T_2$ on vertices from $F_1$. 
	Note that $i$ has at most three neighbors and is thus adjacent to at most two friends. 
	Let $i'\in T_1\setminus \{i\}$ and let $j'\in T_2\setminus \{j\}$ (such agents exist by the assumption that we have at least five agents per type). 
	Since $G$ is an $(x\times y)$-grid with $x,y\geq 4$, there 
	exists a path from $v_{i'}$ to $v_{j'}$ that does not visit the two adjacent vertices $v_i$ and $v_j$ that are part of $F_1$. 
	Given that the two endpoints of the path are occupied by agents of different types, on this path, there need to be agents $i''\in T_1\setminus \{i\}$ and $j''\in T_2\setminus \{j\}$ that are placed on adjacent vertices on this path. 
	Now, if we remove the edge(s) between $v_i$ and its at most two neighbors different from $v_j$, then swapping $i$ and $j''$ is profitable, as $u_i(\mathbf{v})=0<u_i(\mathbf{v}^{i 
		\leftrightarrow j''})$ and $u_{j''}(\mathbf{v})<1=u_{j''}(\mathbf{v}^{i 
		\leftrightarrow j''})$. 
	
	\paragraph*{Induction Step:}
    Assume that it holds for some $\ell \in 
	\NN$,  that all $F_k$ with $k\leq \ell$ are occupied only by agents from 
	$T_t$, $t\in\{1,2\}$.
	
	We now show that then $F_{\ell+1}$ also needs to be only occupied by agents 
	from~$T_t$. We do so by 
	showing that if there are agents from the other type $T_{t'}$, $t'\in\{1,2\}\setminus\{t\}$ in $F_{\ell+1}$, 
	then there exists a profitable swap after deleting at most two edges. Without loss of generality, assume that $|F_{\ell+1}|\geq 2$. If $|F_{\ell+1}|< 
	2$, then we have already reached a contradiction, since all other frames are 
	occupied by agents 
	from $T_t$ only, but there are at least five agents from~$T_{t'}$ that have to be positioned. 
	
    First, suppose that there are only agents from $T_{t'}$ in $F_{\ell+1}$. Let 
	$i \in T_{t'}$ be the agent in the bottom-left corner of $F_{\ell+1}$.
	Formally, let $i \in T_{t'}$ with $v_i = (a,b)  \in F_{\ell+1}$ such that it 
	holds for all other $(a',b') \in F_{\ell+1}$ that $a'\geq a$ and $b' \geq b$. 
	Note that $i$ is adjacent to at least two agents of the other type in $F_{\ell}$ and has 
	at most two adjacent friends. Thus, after deleting at most two edges, we 
	have that $u_i(\mathbf{v})=0$.  Now consider another agent $i' \in T_{t'}$ 
	in $F_{\ell+1}$. Agent $i'$ is adjacent to an agent $j\in T_t$ in $F_{\ell}$. It 
	holds that $u_j(\mathbf{v})<1$. Then, swapping $i$ and $j$ is profitable, 
	since $u_j(\mathbf{v})<1=u_j(\mathbf{v}^{i \leftrightarrow j})$  and 
	$u_i(\mathbf{v})=0<u_i(\mathbf{v}^{i \leftrightarrow j})$. Therefore, there have to be agents from $T_t$ in $F_{\ell+1}$. 
	
	Next, suppose for the sake of contradiction that there are agents from both types in~$F_{\ell+1}$. Then, there exists an agent $i \in T_{t'}$ in $F_{\ell+1}$ that is adjacent to an agent from $T_{t}$  on a vertex also in $F_{\ell+1}$. Note that agent $i$ is adjacent to at least two agents of the other type (the adjacent agent from $T_t$ in $F_{\ell+1}$ and some agent in $F_{\ell}$) and thus has at most two adjacent friends. Next, let $j$ be some agent from $T_t$ that is not adjacent to $i$ and that is not positioned on one of the corner-vertices (i.e., vertices $(a,b)$ of an $(x\times y)$-grid where $a$ is $1$ or $x$ and $b$ is $1$ or $x$) of the grid. Such an agent exists as the frame $F_1$ of the $(x\times y)$-grid with $x,y\geq 4$ is fully occupied by agents from $T_t$. As  $|T_{t'}|\geq 5$, there furthermore exists some other agent $i'$ from $T_{t'}$ not adjacent to $i$ (on some vertex in a frame $F_k$ with $k\geq\ell+1$). As $G$ is a $(x\times y)$-grid with $x,y\geq 4$, it is easy to see that there exists a path from $v_{i’}$ to $v_j$ that does not visit one of the neighbors of $v_i$ (or $v_i$).
	Analogously to the argument for the induction base, given that the two endpoints of the path are occupied by agents of different types, on this path, there need to be agents $i''\in T_{t'}\setminus \{i\}$ and $j''\in T_t$ that are placed on adjacent vertices on this path. Note that by definition of the path $j''$ is not adjacent to $i$.
	Now, if we remove the edge(s) between $v_i$ and its at most two adjacent friends, then swapping $i$ and $j''$ is profitable, as $u_i(\mathbf{v})=0<u_i(\mathbf{v}^{i 
		\leftrightarrow j''})$ and $u_{j''}(\mathbf{v})<1=u_{j''}(\mathbf{v}^{i 
		\leftrightarrow j''})$. 
	
	Thus, the induction step follows. 
	Consequently, in case there is a $2$-edge-robust equilibrium,
	all vertices need to be occupied by 
	agents from the same type. As, however, we assume that $|T_1|,|T_2|\geq 5$, 
	this leads to a contradiction.
	
	\paragraph*{\textbf{Vertex-Robustness at Most~1.}}
	Similarly, we can show that in a Schelling game with $|T_1|\ge 4$ and~$|T_2|\ge 
	4$ on an $(x\times y)$-grid with  $x\geq 4$ and $y\geq 4$, the vertex-robustness of a 
	swap-equilibrium is at most one. For the sake of contradiction, assume that there exists a swap-equilibrium $\mathbf{v}$ which is $2$-vertex-robust. By induction over the frames of~$G$, we 
	show that all vertices have to be occupied by agents from one type in 
	$\mathbf{v}$, which contradicts that we have at least four agents from both 
	types.
	
	\paragraph*{Base Case:} 
	First, consider the frame $F_1$ and assume for the sake of contradiction  
	that there are agents from both types on vertices from $F_1$. Then, there exist two 
	adjacent agents $i \in T_1$ and $j\in T_2$ on vertices from $F_1$. 
	Note that $i$ has at most three neighbors. 
	Let $i'\in T_1$ be an agent from $T_1$ that is not placed on a vertex from $N(v_i)\cup \{v_i\}$ and let $j'\in T_2$ be an agent from $T_2$ which is not placed on a vertex from $N(v_i)\cup \{v_i\}$ (such vertices exist as there exist four agents per type, $|N(v_i)\cup \{v_i\}|\leq 4$, and for both types at least one of its agents is placed on a vertex from $N(v_i)\cup \{v_i\}$). 
	Since~$G$ is an $(x\times y)$-grid with $x,y\geq 4$, there 
	exists a path from $v_{i'}$ to $v_{j'}$ that does not visit a vertex from $N(v_i)\cup \{v_i\}$. 
	Given that the two endpoints of the path are occupied by agents of different types, on this path, there need to be agents $i''\in T_1$ and $j''\in T_2$ that are placed on adjacent vertices on this path and not on a vertex from $N(v_i)\cup \{v_i\}$. 
	Now, if we remove the at most two vertices in $N(v_i)\setminus \{v_j\}$, then $i$ has utility zero and neither $v_{i''}$ nor $v_{j''}$ has been deleted. 
	Thus, swapping $i$ and $j''$ is profitable, as $u_i(\mathbf{v})=0<u_i(\mathbf{v}^{i 
		\leftrightarrow j''})$ and $u_{j''}(\mathbf{v})<1=u_{j''}(\mathbf{v}^{i 
		\leftrightarrow j''})$. 
	
	\paragraph*{Induction Step:} Assume that it holds for some $\ell \in 
	\NN$,  that all $F_k$ with $k\leq \ell$ are occupied only by agents from 
	$T_t$, $ t \in \{1,2\}$.
	
	We now show that then $F_{\ell+1}$ also needs to be only occupied by agents 
	from~$T_t$. Similarly to the proof for edge-robustness, we do so by 
	showing that if there are agents from the other type $T_{t'}$, $t'\in\{1,2\}\setminus\{t\}$ in $F_{\ell+1}$, 
	then there exists a profitable swap after deleting at most two vertices. 
	Without loss of generality, assume that $|F_{\ell+1}|\geq 4$. If $|F_{\ell+1}|< 
	4$, then we have already reached a contradiction, since all other frames are 
	occupied by agents 
	from $T_t$ only, but there are at least four agents from the other 
	type 
	$T_{t'}$ that have to be positioned. 
	
	First, suppose that there are only agents from $T_{t'}$ in $F_{\ell+1}$. Let 
	$i \in T_{t'}$ be the agent in the bottom-left corner of $F_{\ell+1}$. 
	Formally, let $i \in T_{t'}$ with $v_i = (a,b)  \in F_{\ell+1}$ such that it 
	holds for all other $(a',b') \in F_{\ell+1}$ that $a'\geq a$ and $b' \geq b$. 
	Note that $i$ is adjacent to at least two agents of the other type in $F_{\ell}$ and has 
	at most two adjacent friends. Thus, after deleting at most two vertices, we 
	have that $u_i(\mathbf{v})=0$.  As $|F_{\ell+1}|\geq 4$, there exists another agent $i' \in T_{t'}$ 
	in $F_{\ell+1}$ outside of the neighborhood of $i$. Agent $i'$ is adjacent to an agent $j\in T_t$ in $F_{\ell}$. It 
	holds that $u_j(\mathbf{v})<1$. Then, swapping $i$ and $j$ is profitable, 
	since $u_j(\mathbf{v})<1=u_j(\mathbf{v}^{i \leftrightarrow j})$  and 
	$u_i(\mathbf{v})=0<u_i(\mathbf{v}^{i \leftrightarrow j})$.
	
	Therefore, there have to be agents from $T_t$ in $F_{\ell+1}$. Next, suppose for the sake of contradiction that there are agents from both types in $F_{\ell+1}$. Then, there exists an agent $i \in T_{t'}$ in $F_{\ell+1}$ that is adjacent to an agent from $T_{t}$  on a vertex also in $F_{\ell+1}$. Note that agent $i$ is adjacent to at least two agents of the other type (the adjacent agent from $T_t$ in $F_{\ell+1}$ and some agent in $F_{\ell}$) and thus has at most two adjacent friends. Moreover, $i$ has at most four neighbors.
	Let $j\in T_{t}$ be an agent from $T_{t}$ which is not adjacent to $i$ and not placed on a corner-vertex of the grid (such an agent exists as the frame $F_1$ of the $(x\times y)$-grid with $x,y\geq 4$ is fully occupied by agents from $T_t$). Furthermore, let $i' \in T_{t'}$ be an agent from $T_{t'}$ that is not adjacent to $i$ (such an agent exists as $|T_{t'}|\geq 4$ and two of the five vertices in $N(v_i)\cup \{v_i\}$ are occupied by agents from $T_t$). Note that $i'$ has to be placed on a vertex in a frame $F_j$ with $j>i$ (and thus, in particular, can not be placed on a corner-vertex of the grid). As $G$ is an $(x\times y)$-grid with $x,y\geq 4$ and neither $v_j$ nor $v_{i’}$ is a corner-vertex, there exists a path from $v_j$ to $v_{i’}$ that does not go through $N(v_i)\cup \{v_i\}$. 
	Given that the two endpoints of the path are occupied by agents of different types, on this path, there need to be agents $i''\in T_{t’}$ and $j''\in T_t$ that are placed on adjacent vertices on this path and not on a vertex from $N(v_i)\cup \{v_i\}$. 
	Now, if we remove the at most two vertices adjacent to $v_i$ occupied by friends of agent $i$, then $i$ has utility zero and neither $v_{i''}$ nor $v_{j''}$ has been deleted. 
	Thus, swapping $i$ and $j''$ is profitable, as $u_i(\mathbf{v})=0<u_i(\mathbf{v}^{i 
		\leftrightarrow j''})$ and $u_{j''}(\mathbf{v})<1=u_{j''}(\mathbf{v}^{i 
		\leftrightarrow j''})$. 
	
	Thus, the induction step follows. 
	Consequently, in case there is a $2$-vertex-robust equilibrium,
	all vertices need to be occupied by 
	agents from the same type. As, however, we assume that $|T_1|,|T_2|\geq 4$, 
	this leads to a contradiction.
	
	\paragraph*{Schelling Games With Tight Robustness.}
	Let $I$ be 
	a Schelling game with $|T_1| = |T_2|$ on a grid $G$ with $x\geq4$ 
	columns and $y\geq 2$ rows and $x$ being even.
	
	We start by defining the assignment $\mathbf{v}$ as 
	follows. The first column is occupied by agents from $T_1$, the following 
	columns up to column $\frac{x}{2}+1$ are occupied by the agents from $T_2$ 
	and the remaining agents from $T_1$ are placed on columns $\frac{x}{2}+2$ 
	to $x$. 
	
	We now show that $\mathbf{v}$ is a swap-equilibrium. Observe that it holds 
	for all agents $i \in T_2$ that $u_i(\mathbf{v})\geq \frac{1}{2}$. An agent 
	$j \in T_1$ has at most one neighbor in $T_2$. Therefore, agent $i$ could 
	achieve at most $u_i(\mathbf{v}^{i \leftrightarrow j})=\frac{1}{2}\leq 
	u_i(\mathbf{v})$ by swapping with an agent $j\in T_1$. Thus, no profitable 
	swap exists. 
	
	In $\mathbf{v}$, there exists an agent $i \in T_1$ on the vertex at the top 
	of the first column with only one adjacent friend $i'\in T_1$. If we now 
	either delete vertex $v_{i'}$ or the edge $\{v_i,v_{i'}\}$, then agent $i$ and 
	agent $j \in T_2$ at 
	the bottom of column $\frac{x}{2}+1$, who has utility less than $1$ in 
	$\mathbf{v}$, have a profitable swap. It follows that $\mathbf{v}$ has edge- 
	and vertex-robustness zero. 
	
	In the following, we construct an assignment $\mathbf{v}'$ with 
	edge- and vertex-robustness 
	one: The agents from $T_1$ occupy the first 
	$\frac{x}{2}$ columns and the agents from $T_2$ are placed on the remaining 
	columns. We now argue that $\mathbf{v}'$ is a swap-equilibrium: Observe 
	that only the agents on columns $\frac{x}{2}$ and 
	$\frac{x}{2}+1$ have a utility of less than $1$. Therefore, only these 
	agents can be involved in a profitable swap. Consider some agent $i\in T_1$ 
	from 
	column $\frac{x}{2}$. Since $\frac{x}{2}\geq2$, we have 
	$u_i(\mathbf{v}')=\frac{2}{3}$ if $i$ occupies the vertex on the top or 
	bottom of the column and $u_i(\mathbf{v}')=\frac{3}{4}$ otherwise. If $i$ 
	swaps with some agent $j$ from column $\frac{x}{2}+1$, then 
	$u_i(\mathbf{v'}^{i \leftrightarrow j })\leq\frac{1}{3}$ if $v_j$ is the 
	vertex at the top or bottom of the column and $u_i(\mathbf{v'}^{i 
		\leftrightarrow j })\leq\frac{1}{4}$ otherwise. Hence, no 
	profitable swap 
	involving an agent from $T_1$ exists and thus $\mathbf{v}'$ is a 
	swap-equilibrium assignment.
	
	We now show that $\mathbf{v}'$ has edge-robustness $1$. We have already 
	argued 
	above that no swap-equilibrium can be be  
	$2$-edge-robust. 
	To show that $\mathbf{v}'$ is $1$-edge-robust, consider deleting an edge $e 
	\in 
	E(G)$.  From \Cref{swap:corollaryEdgeSameType}, we know that deleting an 
	edge 
	between agents of different types cannot make $\mathbf{v}'$ unstable. 
	Therefore, we only consider edges between agents of the same type. Hence, 
	by 
	symmetry, without loss of generality, let $e$ be an edge between two agents 
	from $T_1$.  Still, only 
	the 
	agents on columns $\frac{x}{2}$ and $\frac{x}{2}+1$ have utility of less 
	than 
	$1$ after the deletion of $e$. If an agent $i \in T_1$  on vertex $v_i$ 
	incident to $e$ is positioned 
	on 
	the top or bottom of column $\frac{x}{2}$, then 
	$u^{G-\{e\}}_i(\mathbf{v}')=\frac{1}{2}$ 
	and $u^{G-\{e\}}_i(\mathbf{v}')=\frac{2}{3}$ otherwise. However, by 
	swapping with some agent 
	$j 
	\in 
	T_2$, agent $i$ can  get at most $u^{G-\{e\}}_i(\mathbf{v'}^{i 
		\leftrightarrow 
		j})=\frac{1}{3}$. Therefore, the swap cannot be 
	profitable 
	and $\mathbf{v}'$ is $1$-edge-robust and hence has edge-robustness~$1$.
	
	The argument why $\mathbf{v}'$ has vertex-robustness $1$ works quite 
	similar: Let $w\in V(G)$ be the vertex we deleted. In $G-w$, again only agents 
	on columns $\frac{x}{2}$ and $\frac{x}{2}+1$ have utility of less 
	than 
	$1$ and can thus be involved in a profitable swap. Assume without loss of 
	generality that $w$ is occupied by an agent from $T_1$ in $\mathbf{v}$. 
	Then, for an 
	agent $i \in T_1$ positioned on vertex $v_i$ from column $\frac{x}{2}$ adjacent 
	to $w$ in $G$, it holds that if $v_i$ is on 
	the top or bottom of column $\frac{x}{2}$, then 
	$u^{G-\{w\}}_i(\mathbf{v}')=\frac{1}{2}$ 
	and $u^{G-\{w\}}_i(\mathbf{v}')=\frac{2}{3}$ otherwise. However, by 
	swapping with some agent 
	$j 
	\in 
	T_2$, agent $i$ can  get at most $u^{G-\{w\}}_i(\mathbf{v'}^{i 
		\leftrightarrow 
		j})=\frac{1}{3}$. Therefore, the swap cannot be 
	profitable 
	and $\mathbf{v}'$ is $1$-vertex-robust and hence has vertex-robustness~$1$.
\end{proof}

\paragraph{$\boldsymbol{\alpha}$-Star-Constellation Graphs.} Lastly, motivated by the observation that on all previously considered graph 
classes there exist swap-equilibria with zero edge-robustness and on all 
considered graph classes except cliques there exist swap-equilibria
with zero 
vertex-robustness, we investigate 
$\alpha$-star-constellation graphs, a generalization of 
stars and $\alpha$-caterpillars. We prove that every 
swap-equilibrium in a Schelling game on an $\alpha$-star-constellation graph is 
$\alpha$-vertex-robust and $\alpha$-edge-robust. We 
also 
show that a swap-equilibrium on an 
$\alpha$-star-constellation graph may fail to exist but that we can precisely 
characterize swap-equilibria on such graphs. Using this 
characterization, we design a polynomial-time algorithm for \sEq{} on 
$\alpha$-star-constellation graphs and show that there always exists a 
swap-equilibrium on an $\alpha$-caterpillar, that is, an 
$\alpha$-star-constellation graph which restricted to non-degree-one 
vertices forms a path.

\begin{theorem} \label{constellation:rob}
	In a Schelling game on an $\alpha$-star-constellation graph for $\alpha \in 
	\NN_0$, every swap-equilibrium $\mathbf{v}$
	is $\alpha$-edge and $\alpha$-vertex-robust.
\end{theorem}
\begin{proof}
	Let $\mathbf{v}$ be a swap-equilibrium on an $\alpha$-star-constellation 
	graph~$G$ for some $\alpha \in \NN_0$.
	We make a case distinction based on whether or not there exists an agent~$i$ on 
	a degree-one vertex adjacent to an agent $j$ of the other type in $\mathbf{v}$.
	If this is the case, then assume without loss of generality that
	$i\in T_1$ and $j\in T_2$ and observe that it needs to hold that all 
	agents $j'\in T_2\setminus \{j\}$ are only adjacent to friends, as 
	otherwise~$j'$ and~$i$ have a profitable swap.  Now consider the 
	topology $G-S$ for some subset $S\subseteq E(G)$ or some subset $S\subseteq 
	V(G)$. Then, for all $j'\in 
	T_2 \setminus  \{j\}$, agent $j'$ cannot be involved in a 
	profitable swap in $G-S$, as $j'$ got deleted, is only adjacent to friends, or 
	placed 
	on an isolated vertex. Moreover, there also cannot exist a profitable swap for 
	$j$, as no agent from~$T_1$ is adjacent to an agent from $T_2\setminus \{j\}$.
	Hence, $\mathbf{v}$ is~$|E(G)|$-edge-robust and $|V(G)|$-vertex-robust.
	
	Now, assume that all agents on a degree-one vertex are only adjacent to 
	friends in~$\mathbf{v}$ and consider the topology $G-S$ for $S \subseteq E(G)$ 
	or $S \subseteq V(G)$ 
	with 
	$|S|\leq \alpha$. Note that, in $G-S$,
	only agents $i \in T_1$ and $j \in T_2$ with $\deg_{G}(v_i)>1$ 
	and 
	$\deg_{G}(v_j)>1$ and $\deg_{G-S}(v_i)\geq 1$ and 
	$\deg_{G-S}(v_j)\geq 1$ can be involved in a profitable swap, 
	since all other agents either occupy an isolated vertex or are only 
	adjacent to friends in $G-S$. For vertex-robustness, it additionally needs to hold 
	that $v_i,v_j\notin S$. Since $G$ is an $\alpha$-star-constellation graph and 
	we 
	delete 
	at most $\alpha$ edges or $\alpha$ vertices, it holds that both $v_i$ and $v_j$ 
	are adjacent to at 
	least as many degree-one vertices as non-degree-one vertices in $G-S$. By our 
	assumption, 
	the 
	agents on degree-one vertices adjacent to $v_i$  are friends of $i$ and the 
	agents on degree-one vertices adjacent to $v_j$ are friends of $j$. Hence, 
	swapping $i$ and $j$ 
	cannot be profitable, as $u_k^{G-S}(\mathbf{v})\geq 
	\nicefrac{1}{2}$ and $u_k^{G-S}(\mathbf{v}^{i \leftrightarrow j})\leq 
	\nicefrac{1}{2}$ 
	for~$k\in \{i,j\}$.
\end{proof}

\Cref{constellation:rob} has no implications for the existence 
of 
swap-equilibria on $\alpha$-star-constellation graphs. Indeed, we observe that
there is no swap-equilibrium in a Schelling game 
with $|T_1|=5$ and $|T_2|=7$ 
on the $1$-star-constellation graph depicted in \Cref{fig:cater}. Notably, to 
the 
best of our 
knowledge the graph from \Cref{fig:cater} is the first known 
graph without a swap-equilibrium that is not a tree. 

\begin{proposition}\label{
		constellation:triangle}
	A Schelling game on an $\alpha$-star-constellation 
	graph $G$
	may fail to admit a swap-equilibrium, even if $G$ is a split 
	graph, that is, 
	the vertices of the graph can be partitioned into a clique and an independent 
	set. 
\end{proposition}
\begin{figure}[t]
	\tikzstyle{alter}=[circle, minimum size=12.5pt, draw, inner sep=1pt, 
	semithick] 
	\tikzstyle{majarr}=[draw=black, thick]
	\centering
	\begin{tikzpicture}[auto]
		
		\node[alter] at (0ex,0ex) (a) {};
		
		\node[alter] at (-6ex,6ex) (a1) {};
		
		\node[alter] at (0ex,7ex) (a2) {};
		
		\node[alter] at (6ex,6ex) (a3) {};
		
		\node[alter] at (30ex,0ex) (b) {};
		
		\node[alter] at (24ex,6ex) (b1) {};
		
		\node[alter] at (30ex,7ex) (b2) {};
		
		\node[alter] at (36ex,6ex) (b3) {};
		
		\node[alter] at (60ex,0ex) (c) {};
		
		\node[alter] at (54ex,6ex) (c1) {};
		
		\node[alter] at (60ex,7ex) (c2) {};
		
		\node[alter] at (66ex,6ex) (c3) {};
		
		\draw[majarr] (a) edge  (a1);
		\draw[majarr] (a) edge  (a2);
		\draw[majarr] (a) edge  (a3);
		\draw[majarr] (b) edge  (b1);
		\draw[majarr] (b) edge  (b2);
		\draw[majarr] (b) edge  (b3);
		\draw[majarr] (c) edge  (c1);
		\draw[majarr] (c) edge  (c2);
		\draw[majarr] (c) edge  (c3);
		\draw[majarr] (a) edge  (b);
		\draw[majarr] (b) edge  (c);
		\draw[majarr] (a) edge [bend right=10] (c);
		
 		\end{tikzpicture}
	\caption{\label{fig:cater} There is no swap-equilibrium in a Schelling game 
		with $|T_1|=5$ and $|T_2|=7$ 
		on this $1$-star-constellation graph.}
\end{figure}
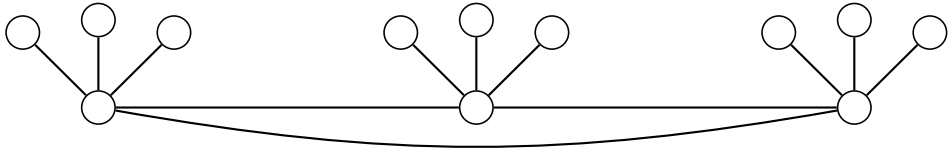
\begin{proof}
	Consider the Schelling game with $|T_1|=5$ many agents of type $T_1$ and $|T_2|=7$ many agents of type $T_2$ on the graph $G$  
	from \Cref{fig:cater}, which consists of three $3$-stars whose central 
	vertices form a clique.  Observe 
	that as
	all stars in $G$ consist of four vertices and neither $|T_1|=5$ nor 
	$|T_2|=7$ are  
	divisible by four, in any assignment~$\mathbf{v}$, there exists 
	a 
	degree-one vertex occupied by an agent $i \in T_l$ such that the adjacent 
	central vertex is occupied by an agent $j \in T_{l'}$ of the other type 
	with 
	$l \neq l'$. Let $v \neq v'$ be the other two central 
	vertices. We make a case distinction based on whether the agents on the 
	degree-one vertices  adjacent to $v$ 
	and~$v'$ have the same type as their respective neighbor on the central 
	vertex.
	If this is the case, then since we have $|T_1|<8$ and $|T_2|<8$, the 
	vertices $v$ and $v'$ 
	cannot be occupied by agents of the same type. Assume by symmetry, without 
	loss of generality, that an 
	agent $j' \in T_{l'}$ occupies vertex $v$ and an agent $i' \in 
	T_l$ occupies vertex $v'$. Then, we have $u_{j'}(\mathbf{v})<1$ and 
	swapping $i$ and $j'$ is 
	profitable, as it holds that $u_{i}(\mathbf{v})=0<u_{i}(\mathbf{v}^{i 
		\leftrightarrow j'})$ and  
	$u_{j'}(\mathbf{v})<1=u_{j'}(\mathbf{v}^{i 
		\leftrightarrow j'})$.
	
	Otherwise, there is an agent on a degree-one vertex that has a 
	different type than the agent on the adjacent central vertex $v$ or $v'$. 
	This implies that there is an agent $j'\neq j$ from $T_{l'}$ 
	with $u_{j'}(\mathbf{v})<1$. Then, similarly to the case above, swapping~$i$ and $j'$ is profitable.
\end{proof}

On the positive side, we can precisely characterize
swap-equilibria in Schelling games on $\alpha$-star-constellation graphs.

\begin{theorem}\label{constellation:equi}
	Let $G$ be an $\alpha$-star-constellation graph with $\alpha \in 
	\NN_0$ 
	and let $\mathbf{v}$ be an assignment in some Schelling game on $G$. 
	The assignment $\mathbf{v}$ is a swap-equilibrium if and only if at least one 
	of the following two conditions holds.
	\begin{enumerate}
		\item Every vertex $v \in V(G)$ with $\deg_G(v)=1$ is occupied by an agent 
		from 
		the same type as the only adjacent agent in $\mathbf{v}$.
		\item There exists an agent $i \in T_l$ for some $l \in \{1,2\}$ 
		such that all other agents $i' \in T_l \setminus \{i\}$ are only 
		adjacent to friends in $\mathbf{v}$.
	\end{enumerate}
\end{theorem}
\begin{proof}
	First, we prove that any assignment that fulfills at least one of the two 
	conditions is always a swap-equilibrium.
	Let $\mathbf{v}$ be an assignment satisfying the first condition, that is, 
	every 
	$v \in V$ with $\deg_G(v)=1$ is occupied by an agent from the same type as the 
	agent on the only adjacent vertex in $\mathbf{v}$. Thus, for all $i \in N$ with 
	$\deg_G(v_i)=1$, we 
	have $u_i(\mathbf{v})=1$. No such agent $i$ can be involved in a 
	profitable swap. Recall that by definition of $G$, we have $|\{ w \in N_G(v) 
	\mid \deg_G(w)=1\}| \geq |\{ w \in N_G(v) \mid \deg_G(w)>1\}| + \alpha$ for all 
	$v \in V$ with $\deg_G(v) > 1$. Since all agents on vertices adjacent to $v_j$ 
	with degree one are friends, we have $u_j(\mathbf{v})\geq \frac{1}{2}$ for all 
	$j \in N$ with $\deg_G(v_j)>1$. Now, consider an agent $j \in T_l$ and an agent 
	$j' \in T_{l'}$ of the other type $l' \neq l$ on vertex $v_{j'}$ with  
	$\deg_G(v_{j'})>1$. If we swap $j$ and $j'$, we have $u_j(\mathbf{v}^{j 
		\leftrightarrow j'})\leq \frac{1}{2}\leq u_j(\mathbf{v})$. Thus, no profitable 
	swap is possible and $\mathbf{v}$ is a 
	swap-equilibrium. 
	
	Now, we consider an assignment that fulfills the second condition. Let 
	$\mathbf{v}$ be an assignment such that there exists an agent $i \in T_l$ for 
	one of the types $l \in \{1,2\}$ such that all other agents $i' \in T_l 
	\setminus \{i\}$ of type $T_l$ are only adjacent to friends. Assume without 
	loss of generality 
	that $l=1$. For all $i' \in T_1 \setminus \{i\}$, we have 
	$u_{i'}(\mathbf{v})=1$. Similarly, for all $j \in T_2$ with $v_j \notin 
	N_G(v_i)$, we also have $u_{j}(\mathbf{v})=1$. Hence, only agent $i \in T_1$ 
	and 
	an agent $j' \in T_2$ with $v_{j'} \in N_G(v_i)$ can have a profitable swap. 
	However, after swapping $i$ and $j'$, agent $i \in T_1$ is only adjacent to 
	agents from $T_2$ and has $u_i(\mathbf{v}^{i \leftrightarrow j})=0$. Therefore, 
	no profitable swap is possible.
	
	Next, we will argue that any assignment $\mathbf{v}$ for which both conditions 
	do not hold cannot be a swap-equilibrium. Thus, in assignment $\mathbf{v}$, 
	there exists an agent $i \in N$ with $\deg_G(v_i)=1$ such that the only 
	adjacent 
	agent is of the other type. Assume without 
	loss of generality that $i \in T_1$. Additionally, for 
	both types there exist two agents  $x,x' \in T_1$ with $x \neq x'$ and $y,y' 
	\in 
	T_2$ with $y \neq y'$ such that $\{v_x,v_y\} \in E$ and $\{v_{x'},v_{y'}\} \in 
	E$. We have $u_{i}(\mathbf{v})=0$. Furthermore, since $\deg_G(v_i)=1$, at least 
	one of the agents $y,y' \in T_2$ has to be positioned outside of the 
	neighborhood of $v_i$. Assume without loss of generality that $v_y \notin N_G(v_i)$ and thus also 
	$x \neq i$. Then, swapping $i$ and $y$ is profitable: We have 
	$u_{i}(\mathbf{v})=0<u_{i}(\mathbf{v}^{i \leftrightarrow y})$, since $i$ is 
	adjacent to $x$ in $\mathbf{v}^{i \leftrightarrow y}$. It also holds that 
	$u_{y}(\mathbf{v})<1=u_{y}(\mathbf{v}^{i \leftrightarrow y})$, since $y \in 
	T_2$ 
	is adjacent to $x \in T_1$ in $\mathbf{v}$ and the only neighbor of $y$ in 
	$\mathbf{v}^{i \leftrightarrow y}$ has the same type as $y$. Hence, the 
	assignment 
	$\mathbf{v}$ cannot be a swap-equilibrium.
\end{proof}

The characterization of swap-equilibria from 
\Cref{constellation:equi} also yields a polynomial-time algorithm (using dynamic programming for \textsc{Subset Sum}) to decide for 
a Schelling game on an $\alpha$-star-constellation graph whether it admits a 
swap-equilibrium.
\begin{corollary}\label{constellation:comp}
	For a Schelling game on an $\alpha$-star-constellation graph with $\alpha \in 
	\NN_0$, one can decide 
	in polynomial time whether a swap-equilibrium exists.
\end{corollary}
\begin{proof}
	Recall that one can think of an $\alpha$-star-constellation graph $G$ as a 
	collection of stars with additional edges between the central vertices of 
	the stars. We define the \textit{star-partition} $S(G)$ of $G$ as the 
	partitioning of $V(G)$ into these stars, that is, $S(G)=\{ \{v\} \cup \{ w 
	\in 	
	N_G(v) \mid \deg_G(w)=1 \} \mid  v\in V(G) \wedge \deg_G(v)>1 \}$.
	We note that the star-partition is unique and can be easily computed in  
	polynomial time.  Next, observe that there exists an assignment 
	$\mathbf{v}$ on $G$ that meets the first 
	condition 
	from \Cref{constellation:equi} if and only if $S(G)$ can be 
	partitioned into two sets $A$ and $B$ such that $|\bigcup_{S \in 
		A}S|=|T_1|$ 
	and $|\bigcup_{S \in B}S|=|T_2|$. Deciding this reduces to 
	solving an instance of subset sum with integers bounded by $n$, which can 
	be 
	solved in $\bigO(n^{2})$ by using dynamic programming. To decide whether an 
	assignment $\mathbf{v}$ fulfilling the second condition exists we compute, 
	for each 
	vertex $v 
	\in V$, the connected components $C_1, \dots, C_m$ of $G[V(G)\setminus 
	\{v\}]$ and
	check 
	whether the components can be partitioned into two subsets $A$ and $B$ such 
	that $|\bigcup_{C \in A} V(C)|=|T_1|-1$ and $|\bigcup_{C \in B} 
	V(C)|=|T_2|$  
	or $|\bigcup_{C \in A} V(C)|=|T_1|$ and $|\bigcup_{C \in B} V(C)|=|T_2|-1$. 
	Again, this can be solved by using dynamic programming in~$\bigO(n^{2})$ time.
\end{proof}

\paragraph{$\boldsymbol{\alpha}$-Caterpillars.} Using \Cref{constellation:equi}, we now argue that there is a subclass of 
$\alpha$-star-constellation graphs, namely $\alpha$-caterpillars, on which a 
swap-equilibrium always exists.
Consider a Schelling game on an $\alpha$-caterpillar $G$ with $w_1,\dots, 
w_{\ell}$ 
being the non-degree-one vertices forming the central path 
($\{\{w_{i},w_{i+1}\}\mid i\in [\ell-1]\}\subseteq E(G)$). It is easy to 
construct a swap-equilibrium $\mathbf{v}$ on $G$ by assigning for each $i \in 
\{1, \dots, 
\ell \}$ agents from $T_1$ to~$w_i$ and to adjacent degree-one vertices, until all agents from~$T_1$ have been assigned; in 
which case the remaining vertices are filled with agents from $T_2$. 
As $\mathbf{v}$ fulfills Condition 2 from \Cref{constellation:equi}, 
$\mathbf{v}$ is a swap-equilibrium and it is easy 
to see that $\mathbf{v}$ has edge-robustness~$|E(G)|$ and vertex-robustness 
$|V(G)|$. Notably, this assignment 
somewhat resembles the swap-equilibrium with edge-robustness $|E(G)|$ and 
vertex-robustness 
$|V(G)|$ on a 
path from 
\Cref{path}. In contrast, extending the swap-equilibrium  with 
robustness zero on a path from \Cref{path} such that all agents on degree-one 
vertices are of the same type as their only neighbor, in some Schelling games 
on 
$\alpha$-caterpillars, it is possible to create a 
swap-equilibrium with edge- and vertex-robustness~only~$\alpha$:

\begin{proposition}\label{cater}
	For a Schelling game on an $\alpha$-caterpillar with $\alpha \in 
	\NN_0$, there is a swap-equilibrium with 
	edge-robustness~$|E(G)|$ and vertex-robustness $|V(G)|$. For every 
	$\alpha\in\NN_0$, there is a 
	Schelling game on an $\alpha$-caterpillar with 
	a swap-equilibrium with edge-robustness and vertex-robustness $\alpha$.
\end{proposition}
\begin{proof}
	We start with the first part of the proposition.
	Consider a Schelling game on an $\alpha$-caterpillar $G$ with $w_1,\dots, 
	w_{\ell}$ 
	being the non-degree-one vertices forming a path and $\{\{w_i,w_{i+1}\} 
	\mid i\in [\ell-1]\}\subseteq E(G)$. It is easy to 
	construct a swap-equilibrium $\mathbf{v}$ by assigning for each $i 
	\in 
	\{1, \dots, 
	\ell \}$ agents from $T_1$ first to $w_i$ and then to the degree-one 
	vertices  
	adjacent to $w_i$, until there are no remaining unassigned agents from 
	$T_1$. In this case the remaining vertices are filled with agents from 
	$T_2$. 
	By \Cref{constellation:equi}, $\mathbf{v}$ is a swap-equilibrium (as it satisfies the second condition). Note that 
	in $\mathbf{v}$ there exists only one agent from $T_1$, say $i\in T_1$, who 
	is adjacent to an agent from $T_2$. Let $S\subseteq E(G)$ be a subset of 
	edges 
	or a subset of vertices $S\subseteq V(G)$ of arbitrary size. We now argue 
	that $\mathbf{v}$ is a swap-equilibrium on 
	$G-S$. First, note that $i$ is still the only agent from~$T_1$ who is 
	adjacent to an agent from $T_2$ in  $\mathbf{v}$ on $G-S$. 	
	For all $i' \in T_1 \setminus \{i\}$, we have 
	$u^{G-S}_{i'}(\mathbf{v})=1$ or $v_{i'}$ is an isolated vertex. Similarly, 
	for all $j \in T_2$ with $v_j \notin 
	N_{G-S}(v_i)$, we also have $u^{G-S}_{j}(\mathbf{v})=1$ or $v_{j}$ is an 
	isolated vertex. Hence, 
	only 
	agent $i \in T_1$ 
	and 
	an agent $j' \in T_2$ with $v_{j'} \in N_{G-S}(v_i)$ can have a profitable 
	swap. 
	However, after swapping $i$ and $j'$, agent $i \in T_1$ is only adjacent to 
	agents from $T_2$ and has $u^{G-S}_i(\mathbf{v}^{i \leftrightarrow j})=0$. 
	Therefore, 
	no profitable swap is possible and $\mathbf{v}$ has 
	edge-robustness~$|E(G)|$ and vertex-robustness~$|V(G)|$.
	
	Let us now come to the second part of the proposition. For $\alpha\in 
	\mathbb{N}_0$ let $G$ be an $\alpha$-caterpillar with $w_1,\dots, 
	w_{4}$ 
	being the non-degree-one vertices forming the central path with 
	$\{\{w_i,w_{i+1}\}\mid i\in [3]\}\subseteq E(G)$. For 
	$i\in \{1,4\}$, $w_i$ is adjacent to $\alpha+1$ degree-one vertices, and for $i\in \{2,3\}$, $w_i$ is adjacent to $\alpha+2$ degree-one vertices. We 
	consider  
	the Schelling game on $G$ with $|T_1|=2\cdot(\alpha+1)$ and 
	$|T_2|=2\cdot(\alpha+2)$. Let~$\mathbf{v}$ be the assignment where agents from $T_1$ occupy the $2\cdot(\alpha+1)$ vertices from the stars with central vertices $w_1, w_4$ and the agents from $T_2$ occupy the $2\cdot(\alpha+2)$ vertices from the stars with central vertices $w_2, w_3$.  Note that 
	$\mathbf{v}$ fulfills the first condition from \Cref{constellation:equi} 
	and is thus a swap-equilibrium which is by \Cref{constellation:rob} 
	$\alpha$-vertex-robust and $\alpha$-edge-robust. To show that $\mathbf{v}$ 
	has edge-robustness $\alpha$ it 
	remains to specify a set of $\alpha+1$ edges whose deletion make 
	$\mathbf{v}$ unstable (that is, $\mathbf{v}$ is not $\alpha+1$-edge-robust). Let 
	$S$ be the set of $\alpha+1$ edges containing 
	all edges between $w_1$ and its degree-one neighbors. Then, on $G-S$, swapping agent $i \in T_1$ on $w_1$ and  
	agent $j \in T_2$ on $w_3$ is profitable, as 
	$u_i^{G-S}(\mathbf{v})=0<u_i^{G-S}(\mathbf{v}^{i \leftrightarrow j})$ and 
	$u_j^{G-S}(\mathbf{v})<1=u_j^{G-S}(\mathbf{v}^{i \leftrightarrow j})$.
	To show that $\mathbf{v}$ has vertex-robustness $\alpha$, let $S$ be the 
	set of $w_1$'s degree-one neighbors. Then similar as above, on $G-S$, swapping 
	$i \in T_1$ on $w_1$ and  
	agent $j \in T_2$ on $w_3$ is profitable.
\end{proof}

\section{Conclusion}
We 
proved that even in the simplest variant of 
Schelling 
games where all agents want to maximize the fraction of agents of their type in 
their occupied neighborhood, deciding the existence of a 
swap- or jump-equilibrium is NP-complete.
Moreover, we introduced a notion for the robustness of 
an equilibrium under vertex or edge deletions and proved that the robustness of different 
swap-equilibria on the same topology can vary significantly. In addition, we 
found that the 
minimum and the maximum robustness of swap-equilibria vary depending on the 
underlying topology.

There are multiple possible directions for future research. First, independent of properties of the given graph,
in our reduction showing the NP-hardness of deciding the existence of a swap- or jump-equilibrium, we construct a graph that is not planar and has unbounded maximum degree. 
The same holds for graphs constructed in the reductions from  \citeA{AGARWAL2021103576} for showing NP-hardness in the presence of stubborn agents.
Thus, the computational complexity of deciding the existence of equilibria on planar or bounded-degree graphs (properties that typically occur in the real world) in Schelling games with our without stubborn agents is open. 
Second,  \citeA{bilo2020topological} recently introduced 
the notions of local swap (jump)-equilibria where only adjacent agents are 
allowed to swap places (agents are only allowed to jump to adjacent vertices). 
To the best of our knowledge, the computational complexity of deciding the 
existence of a 
local swap- or jump-equilibrium 
is unknown even if we allow for stubborn agents.  
Third, while we showed that on most considered graphs swap-equilibria can be very non-robust, it might be interesting to search for graphs guaranteeing a higher equilibrium robustness; here, graphs with a high minimum degree and/or high connectivity seem to be promising candidates. 
Fourth, besides 
looking 
at the robustness of equilibria with respect to the deletion of edges or 
vertices, one may 
also study adding or contracting edges or vertices.
Fifth, instead of analyzing the robustness of a specific equilibrium, 
one could also investigate the robustness of a topology regarding 
the existence of an equilibrium.
Lastly, for an equilibrium, it would also be interesting to analyze empirically or theoretically how many reallocations of agents take place on average after a certain change has been performed until an equilibrium is reached again. 

\section*{Acknowledgements}
NB was supported by the DFG project MaMu (NI 369/19) and by the DFG project ComSoc-MPMS (NI 369/22). This paper is largely based on the Bachelor thesis of the first author \cite{kreisel2021equilibria}. We thank an anonymous AAMAS '22 reviewer who pointed us to two possible simplifications in the proof of \Cref{grid}.

\bibliographystyle{theapa}

\end{document}